\documentclass[acmsmall,natbib,dvipsnames]{acmart}

\usepackage{microtype}
\usepackage{etex}
\usepackage{amsmath}
\usepackage{amsthm}
\usepackage{stmaryrd}
\usepackage{xspace}
\usepackage{mathpartir}
\usepackage{tikz}
\usepackage{mathbbol}
\DeclareMathSymbol{\fcmp}{\mathrel}{bbold}{\lq\;}
\usepackage{hyperref}
\usepackage{wasysym}
\usepackage{listings}
\newcommand\dhxrightarrow[2][]{\mathrel{\ooalign{$\xrightarrow[#1\mkern4mu]{#2\mkern4mu}$\cr \hidewidth$\rightarrow\mkern4mu$}}
}

\ccsdesc[500]{Theory of computation~Type structures}

\setcopyright{none}
\acmJournal{TOPLAS}
\acmYear{2021} \acmVolume{43} \acmNumber{1} \acmArticle{4} \acmMonth{4} \acmPrice{15.00}\acmDOI{10.1145/3450272}

\begin{document}

\title{Polymorphic Iterable Sequential Effect Systems}
\author{Colin S.\ Gordon}
\orcid{0000-0002-9012-4490}
\affiliation{
    \department{Department of Computer Science}
    \institution{Drexel University}
    \city{Philadelphia}\state{PA}
    \country{USA}
}
\email{csgordon@drexel.edu}

\begin{abstract}
Effect systems are lightweight extensions to type systems that can verify a wide range of important properties with modest developer burden.  But our general understanding of effect systems is limited primarily to systems where the order of effects is irrelevant.
Understanding such systems in terms of a semilattice of effects grounds understanding of the essential issues, and provides guidance when designing new effect systems.
By contrast, sequential effect systems --- where the order of effects is important --- lack an established algebraic structure on effects.

We present an abstract polymorphic effect system parameterized by an effect quantale --- an algebraic structure with well-defined properties that can model the effects of a range of existing sequential effect systems.  We define effect quantales, derive useful properties, and show how they cleanly model a variety of known sequential effect systems.

We show that for most effect quantales, there is an induced notion of iterating a sequential effect; that for systems we consider the derived iteration agrees with the manually designed iteration operators in prior work; and that this induced notion of iteration is as precise as possible when defined.
We also position effect quantales with respect to work on categorical semantics for sequential effect systems, clarifying the distinctions between these systems and our own in the course of giving a thorough survey of these frameworks.
Our derived iteration construct should generalize to these semantic structures, addressing limitations of that work.
Finally, we consider the relationship between sequential effects and Kleene Algebras, where the latter may be used as instances of the former.
\end{abstract}
\keywords{Type systems, effect systems, quantales, polymorphism}\maketitle

\section{Introduction}
Effect systems are a well-known lightweight extension to standard type systems, which are capable of verifying an array of useful program properties with modest developer effort.
They have proven useful for enforcing error handling~\cite{vanDooren2005,benton2007exceptions,gosling2014java},
ensuring a variety of safety properties for concurrent programs~\cite{safelocking99,rccjava00,objtyrace99,boyapati01,boyapati02,flanagan2003atomicity},
purity~\cite{Hunt2007sealing,Fahndrich2006language},
safe region-based memory management~\cite{lucassen88,talpin1992polymorphic,Tofte1994Regions}, and more.
Effect systems extend type systems to track not only the shape of and constraints on data, but also a summary of the side effects caused by an expression's evaluation.
Java's checked exceptions are the best-known example of an effect system --- the effect of an expression is the set of (checked) exceptions it may throw --- and other effects have a similar flavor, like the set of heap regions accessed by parallel code, or the set of locks that must be held to execute an expression without data races.

However, our understanding of effect systems is concentrated in the space of systems like Java's checked exceptions, where the \emph{order of effects} is irrelevant: the effect system does not care that an \texttt{IllegalArgumentException} would be thrown before any possible \texttt{IOException}.
Effects in such systems are generally structured as a join semilattice, which captures exactly systems where ordering is irrelevant and all that matters is the set of possible behaviors: the join operation is commutative and associative, and the join is used for both alternative branches of execution (i.e., conditionals) \emph{and} sequential composition.
This is an impressively large and useful class of systems, but the assumption that order is irrelevant leaves some of the more sophisticated effect systems for checking more powerful properties out of reach.  We refer to effect systems that disregard program order as \emph{commutative} effect systems, to contrast against the class we study in this paper.\footnote{Section \ref{sec:bg} points out that not all commutative effect systems use only a join semilattice.}  The alternative class of effect systems, where the order in which effects occur matters --- \emph{sequential} effect systems, following Tate's terminology~\cite{tate13}\footnote{These effect systems have been alternately referred to as flow-sensitive~\cite{marino09}, as they are often formalized using flow-sensitive type judgments (with pre- and post-effect) rather than effects in the traditional sense.  However, this term suggests a greater degree of path sensitivity and awareness of branch conditions than most such systems have.  We use Tate's terminology as it avoids technical quibbles.} --- reason directly about the proper ordering of program events.  Examples include
non-block-structured reasoning about synchronization for data races and
deadlock freedom~\cite{boyapati02,tldi12,suenaga2008type},
atomicity~\cite{flanagan2003tldi,flanagan2003atomicity},
memory management~\cite{crary1999typed}, and execution trace properties~\cite{skalka2008types,Skalka2008,Koskinen14LTR}.

Effect system design for the traditional commutative effect systems has been greatly aided in both theory and practice by the recognition that effects in such systems form a bounded join semilattice with top --- a set with all binary joins (least-upper-bound) and greatest and least elements.
On the theory side, this permits general formulations of effect systems to study common properties across a range of systems~\cite{marino09,rytz12,BanadosSchwerter2014gradual}.  On the practical side, this guides the design and implementation of working effect systems.  If an effect system is not a join semilattice, why not? Without a strong reason, perhaps this indicates a mistake, since the semilattice structure coincides with commonly assumed program equivalences, refactorings, and compiler optimizations.  Effect system frameworks can be implemented generically with respect to an effect lattice~\cite{rytz12,toro2015customizable}, and in the common case where effects are viewed as sets of required capabilities, simply specifying the capabilities and exploiting the default powerset lattice makes core design choices straightforward.  In the research literature, the ubiquity of lattice-based (commutative) effect systems simplifies explanations and presentations.

Sequential effect systems so far have no such \emph{established} common basis in terms of an algebraic structure to guide design, implementation, and comparison, in the sense of having a go-to framework that captures enough structure to readily apply an instance to real languages.
This makes design, implementation, and comparison more difficult than we would like. 
Recent work on semantic approaches to modeling sequential effect systems~\cite{tate13,katsumata14,mycroft16} has produced very general characterizations of the mathematics behind key \emph{necessary} constructs (namely, sequencing effects), but with one recent exception~\cite{mycroft16} does not produce a description that is \emph{sufficient} to model full details of a sequential effect system for a real language. This is all the more surprising given that specific effect systems distinguishing sequential and alternating (i.e., conditional branch) composition appeared as early as 1993 in the context of Concurrent ML~\cite{nielson1993cml,amtoft1999}.
Partly this stems from the fact that the accounts of such abstract work proceed primarily by generalizing categorical structures used to model sequential computation, rather than implementing complete known effect systems.  None of this work has directly considered effect polymorphism (essential for any real use); singleton~\cite{aspinall1994subtyping} effects (a.k.a.\ value-dependent~\cite{swamy2011secure} effects, necessary for prominent effect systems both commutative and sequential where effects mention program values like locks); or iteration constructs.  So there is currently a gap between this powerful semantic work, and understanding real sequential effect systems in a systematic way.

We generalize directly from concrete type-and-effect systems to give an abstract algebraic formulation for sequential effects, suitable for modeling some well-known sequential type-and-effect disciplines, and (we hope) useful for guiding the design of future sequential effect systems.  This yields a characterization that captures the full structure common across a range of example systems.  We give important derived constructions (products, and inducing an iteration operation on effects), and put them to use with explicit translations from prior work~\cite{skalka2008types} into our generic core calculus.

Overall, our contributions include:
\begin{itemize}
\item A new algebraic structure for sequential effects --- effect quantales --- that is consistent with existing semantic notions and subsumes commutative effects
\item A syntactic motivation for effect quantales by generalizing from concrete, full-featured sequential effect systems. As a result, we are the first to investigate the interplay between singleton (i.e., value-dependent) effects and sequential effect systems in the abstract (not yet addressed by semantic work).  This reveals subtlety in the metatheory of sequential effects that depend on program values, which should inform further semantic models of sequential effects.
\item Demonstration that effect quantales are not only general, but also sufficient to modularly define the structure of existing non-trivial effect systems.
\item A general construction of effect iteration for most sequential effects system given by an effect quantale. We validate it by showing that applying the iteration construct to prior systems (as effect quantales) gives exactly the hand-derived operations from those works. Applying it to systems that did not consider general iteration constructs yields sensible results.
We also show this construction is optimally precise: when defined (as it is for all effect quantales we have considered from the literature) it gives the most precise result that satisfies some basic axioms.
This contrasts with prior work's speculation that no such general construction of iteration could be given.
\item The first generic \emph{sequential} effect system with effect polymorphism. Prior categorical approaches did not consider parametric polymorphism over effects, so we are the first to elaborate on the complications introduced by the use of partial operations on effects in this setting. We prove syntactic type safety for a language with effect polymorphism and sequential effects, with respect to an arbitrary effect quantale and associated primitives. As part of this, we highlight the difficulties that arise from mixing effect polymorphism and partial effect operators.
\item Precise characterization of the relationship between effect quantales and related notions, ultimately connecting the syntax of established effect systems to semantic work, closing a gap in our understanding.
\item Discussion of the relationship between Kleene Algebras~\cite{KOZEN1994366} --- an algebraic model of program \emph{semantics} --- and effect quantales, which have similar axiomatizations. In particular, we explain the distinctions while showing that any Kleene Algebra can be used as an effect system.
\end{itemize}

\paragraph{A Note on Proofs}
This work stands at the intersection of research communities that sometimes have significantly different assumptions about ``common knowledge'' --- applied type(-and-effect) theorists interested in the proof theory of effect systems, and categorical semanticists interested in the categorical denotational semantics of such systems.  In particular, these two groups seems to have different opinions on how obvious certain elements of order theory applied in this paper might be.  We have erred on the side of being overly explicit when applying concepts from order theory, to make the paper more self-contained.

\paragraph{Relation to Prior Work}
This paper is an extended and updated version of a paper by the same author published in ECOOP 2017~\cite{ecoop17}.  
In addition to including full proofs and more examples, this paper strengthens the original results on deriving iteration for sequential effect systems, and extends the original syntactic type safety proof to also show that the effect quantales correctly enforce a semantic interpretation of effects.
It also adds a discussion of the relationship between sequential effect systems and the established notion of a Kleene Algebra~\cite{KOZEN1994366}, which algebraically models the semantics of imperative programs.
A more detailed comparison between this paper and the original publication appears in Appendix \ref{apdx:comparison}.

\section{Background on Commutative and Sequential Effect Systems}
\label{sec:bg}
Here we derive the basic form of a new algebraic characterization of sequential effects based on generalizing from the use of effects in extant source-level sequential effect systems.  The details of this form are given in Section \ref{sec:quantales}, with a corresponding generic type-and-effect system in Section \ref{sec:soundness}.  We refer to the two together as a framework for sequential effect systems.

By now, the standard mechanisms of traditional effect systems that ignore program order --- what is typically meant by the phrase ``type-and-effect system'' --- are well understood.  The type judgment $\Gamma\vdash e : \tau$ of a language is augmented with a component $\chi$ describing the overall effect of the term at hand: $\Gamma\vdash e : \tau \mid \chi$.  Type rules for composite expressions, such as forming a pair, join the effects of the child expressions by taking the least upper bound of those effects (with respect to the effect lattice).  The final essential adjustment is to handle the \emph{latent effect} of a function --- the effect of the function body, which is deferred until the function is invoked.  Function types are extended to include this latent effect, and this latent effect is included in the effect of function application.  Allocating a closure itself typically\footnote{Strictly speaking, effect systems for reasoning about low-level behaviors such as allocation of closures~\cite{tofte1997region,Tofte1994Regions,Tofte1998Inference,Birkedal1996Representation,BIRKEDAL2001299} do assign non-trivial effects to closure creation, such as allocating into a particular region of memory. Our focus here is on the general case, where effects do not directly reflect internal details of a language implementation.} has no meaningful effect, and is typically given the bottom effect in the semilattice:
\begin{mathpar}
\inferrule*[left=T-Fun]{\Gamma,x:\tau\vdash e : \tau' \mid \chi}{\Gamma\vdash(\lambda x\ldotp e) : \tau\overset{\chi}{\rightarrow}\tau' \mid \bot}
\and
\inferrule*[left=T-Call]{
    \Gamma\vdash e_1 : \tau\overset{\chi}{\rightarrow}\tau' \mid \chi_1\\
    \Gamma\vdash e_2 : \tau \mid \chi_2
}{
    \Gamma\vdash e_1\;e_2 : \tau' \mid \chi_1\sqcup\chi_2\sqcup\chi
}
\end{mathpar}
Consider the interpretation for concrete effect systems.  Java's checked exceptions are an effect system~\cite{gosling2014java,vanDooren2005}: to a first approximation\footnote{Because Java permits inheritance among exception types, the details are actually a bit more subtle. The sets are restricted such that no set contains both $A$ and $B$ where $A<:B$. The join is a slight twist on union, dropping exceptions that have supertypes also present in the na\"ive union.} the effects are sets of checked exception types, ordered by inclusion, with set union as the semilattice join.  The \texttt{throws} clause of a method states its \emph{latent} effect --- the effect of actually \emph{executing} the method (roughly $\chi$ in \textsc{T-Fun} above).  The exceptions thrown by a composite expression such as invoking a method is the union of the exceptions thrown by subexpressions (e.g., the receiver object expression and method arguments) and the latent effect of the code being executed (as in \textsc{T-Call} above).  Most effect systems for treating data race freedom (for block-structured synchronization like Java's \texttt{synchronized} blocks, such as \textsc{RCC/Java}~\cite{rccjava00,Abadi2006}) use sets of locks as effects, where an expression's effect is the set of locks guarding data that may be accessed by that expression.  The latent effect there is the set of locks a method requires to be held by its call-site.  Other effect systems follow similar structure: a binary yes/no effect for whether or not code performs a sensitive class of action like allocating memory in an interrupt handler~\cite{Hunt2007sealing,Hunt2007singularity,Fahndrich2006language} or accessing user interface elements~\cite{ecoop13}; tracking the sets of memory regions read, written, or allocated into for safe memory deallocation~\cite{talpin1992polymorphic,Tofte1994Regions} or parallelizing code safely~\cite{lucassen88,gifford86} or even deterministically~\cite{bocchino09,kawaguchi12}.

But these and many other examples do not care about ordering.  Java does not care which exception might be thrown first.  Race freedom effect systems for block-structured locking do not care about the order of object access within a \texttt{synchronized} block.  Effect systems for region-based memory management do not care about the order in which regions are accessed, or the order of operations within a region.  Because the order of combining effects in these systems is irrelevant, we refer to this style of effect system as \emph{commutative} effect systems, though later in this section we discuss additional nuances of this classification.

Sequential effect systems tend to have slightly different proof theory.  Many of the same issues
arise (latent effects, etc.) but the desire to enforce a sensible \emph{ordering} among expressions
leads to slightly richer type judgments.  Often they take the form $\Gamma;\Delta\vdash e : \tau
\mid \chi \dashv \Delta'$.  Here the $\Delta$ and $\Delta'$ are some kind of pre- and post-state
information --- for example, the sets of locks held before and after executing
$e$~\cite{suenaga2008type}, or abstractions of heap shape before and after $e$'s
execution~\cite{tldi12}.  $\chi$ as before is an element of some partial order (often, though not always, a lattice), such as Flanagan and
Qadeer's atomicity lattice~\cite{flanagan2003tldi} (Figure \ref{fig:atomicity_lattice}).  Some sequential effect systems
have both of these features (split pre/post environments and a more apparent effect), and some only one or the other.\footnote{These components never affect the
type of variables, and strictly reflect some property of the \emph{computation} performed by $e$, making them part of the effect. Occasionally a flow-sensitive type judgment with strong updates to variable types is used --- such as updating the type of a variable from unlocked to locked~\cite{suenaga2008type} --- but in these cases it is also possible to separate the basic type information (e.g., $x$ is a mutual exclusion lock) from the flow-sensitive effect-tracking (e.g., whether or not $x$ is currently locked).}
The judgments for something like a variant of Flanagan and Qadeer's atomicity type system that tracks lock sets flow-sensitively rather than using synchronized blocks
or for an effect system that tracks partial heap shapes before and after updates~\cite{tldi12} might look like the following, using $\Delta$ and $\Upsilon$ to track locks held, and tracking atomicities with $\chi$:
\renewcommand{\MathparLineskip}{\lineskiplimit=3pt\lineskip=3pt}
\begin{mathpar}
    \inferrule{\Gamma,x:\tau;\Upsilon\vdash e : \tau' \mid \chi\dashv\Upsilon'}{\Gamma;\Delta\vdash(\lambda x\ldotp e) : \tau{\xrightarrow{\Upsilon,\chi,\Upsilon'}}\tau' \mid \bot \dashv \Delta}
\and
\inferrule{
    \Gamma;\Delta\vdash e_1 : \tau{\xrightarrow{\Delta'',\chi,\Delta'''}}\tau' \mid \chi_1\dashv\Delta'\\
    \Gamma;\Delta'\vdash e_2 : \tau \mid \chi_2\dashv\Delta''
}{
    \Gamma;\Delta\vdash e_1\;e_2 : \tau' \mid \chi_1;\chi_2;\chi \dashv\Delta'''
}
\end{mathpar}
The sensitivity to evaluation order is reflected in the threading of $\Delta$s through the type rule for application, as well as through the switch to the sequencing composition $;$ of the basic effects.
Confusingly, while $\chi$ continues to be referred to as the effect of this judgment, the real effect is actually a combination of $\chi$, $\Delta$, and $\Delta'$ in the judgment form.  This distribution of the ``stateful'' aspects of the effect through a separate part of the judgment obscures that this judgment really tracks a product of \emph{two} effects --- one concerned with the self-contained $\chi$, and the other a form of effect indexed by pre- and post-computation information (i.e., a parameterized monad~\cite{atkey2009parameterised}).

Rewriting this common style of sequential effect judgment in a form closer to the typical commutative form reveals some subtleties of sequential effect systems:
\begin{mathpar}
    \inferrule{\Gamma,x:\tau\vdash e : \tau' \mid (\Upsilon\leadsto\Upsilon')\otimes\chi
}{
    \Gamma\vdash(\lambda x\ldotp e) : \tau{\xrightarrow{(\Upsilon\leadsto\Upsilon')\otimes\chi}}\tau' \mid (\Delta\leadsto\Delta)\otimes\bot
}
\and
\inferrule{
    \Gamma\vdash e_1 : \tau{\xrightarrow{(\Delta''\leadsto\Delta''')\otimes\chi}}\tau' \mid (\Delta\leadsto\Delta')\otimes\chi_1\\
    \Gamma\vdash e_2 : \tau \mid (\Delta'\leadsto\Delta'')\otimes\chi_2
}{
    \Gamma\vdash e_1\;e_2 : \tau' \mid ((\Delta\leadsto\Delta');(\Delta'\leadsto\Delta'');(\Delta''\leadsto\Delta'''))\otimes(\chi_1;\chi_2;\chi)
}
\end{mathpar}
One change that stands out is that the effect of allocating a closure is not simply the bottom effect (or product of bottom effects) in some lattice.  No sensible lattice of pre/post-state pairs has equal pairs as its bottom.  However, it makes sense that some such equal pair acts as the left and right unit for \emph{sequential} composition of these ``stateful'' effects.\footnote{Technically this style of effect is usually presented as having a potentially different unit for every effect, but we can model this with a single global unit that can be coerced to a particular choice (Section \ref{sec:param_monads}).}  In traditional commutative effect systems, sequential composition is actually least-upper-bound, for which the unit element happens to be $\bot$.
We account for this in our framework.

We also assumed, in rewriting these rules, that it was sensible to run two effect systems ``in parallel'' in the same type judgment, essentially by building a product of two effect systems.  Some sequential effect systems are in fact built this way, as two ``parallel'' systems (e.g., one for tracking locks, one for tracking atomicities, one for tracking heap shapes, etc.) that together ensure the desired properties.  The general framework we propose supports a straightforward product construction.

Another implicit assumption in the refactoring above is that the effect tracking that is typically done via flow-sensitive type judgments is equivalent to \emph{some} algebraic treatment of effects akin to how $\chi$s are managed above.  While it is clear we would \emph{want} a clean algebraic characterization of such effects, the existence of such an algebra that is adequate for modeling known sequential effect systems for non-trivial languages is not obvious.  Our proposed algebraic structures (Section \ref{sec:quantales}) are adequate to model such effects (Section \ref{sec:modeling}).

Examining the sequential variant of other rules reveals more subtleties of sequential effect system design.
For example, effect joins are still required in sequential systems:
\begin{mathpar}
\inferrule{
    \Gamma\vdash e : \mathcal{B} \mid \chi\quad
    \Gamma\vdash e_1 : \tau \mid \chi_1\quad
    \Gamma\vdash e_2 : \tau \mid \chi_2
}{
    \Gamma\vdash\mathsf{if}\;e\;e_1\;e_2 : \tau \mid \chi\sqcup\chi_1\sqcup\chi_2
}
\\\\
\textrm{becomes}
\\\\
\inferrule{
    \Gamma\vdash e : \mathcal{B} \mid \chi\quad
    \Gamma\vdash e_1 : \tau \mid \chi_1\quad
    \Gamma\vdash e_2 : \tau \mid \chi_2
}{
    \Gamma\vdash\mathsf{if}\;e\;e_1\;e_2 : \tau \mid \chi;(\chi_1\sqcup\chi_2)
}
\end{mathpar}
Nesting conditionals can quickly produce an effect that becomes a mass of alternating effect sequencing and join operations.  For a monomorphic effect system, concrete effects can always be plugged in and comparisons made.  However, for a polymorphic effect system, it is highly desirable to have a sensible way to simplify such effect expressions --- particularly for highly polymorphic code where effect variables prohibit simplifying the entire effect --- to avoid embedding the full structure of code in the effect.  Our proposal codifies natural rules for such simplifications, which are both theoretically useful and correspond to common behavior-preserving code transformations performed by developers and compiler optimizations.

\paragraph{Traditional, Commutative, and Sequential Effect Systems}
Before continuing to the formal development, we must clarify a bit more terminology. This paper focuses on sequential effect systems, which use two operators on effects to distinguish cases where ordering is ignored from cases where it is taken into account. We have positioned these in opposition to commutative effect systems, but within commutative effect systems there are further distinctions to make. Most commutative effect systems use a single commutative operator on effects, but there exist others which have two operators on effects, \emph{both} commutative.
Due to their prevalence and the fact that they arose first historically, single-operator commutative effect systems are the sort typically meant by general references to ``effect systems'' with no further qualification.
We will sometimes refer to these single-operator commutative effect systems as \emph{traditional} commutative effect systems, using \emph{commutative} without further qualification to refer to both one- and two-operator effect systems when all such binary operators are commutative. This yields the following containment relations among classes of effect systems:
\[
    \fbox{$\begin{array}{c}
        \textit{Traditional}\\
        \textrm{1 commutative operator}
    \end{array}$}
    \subseteq
    \fbox{$\begin{array}{c}
        \textit{Commutative}\\
        \textrm{2 operators}\\
        \textrm{Both commutative}
    \end{array}$}
    \subseteq
    \fbox{$\begin{array}{c}
        \textit{Sequential}\\
        \textrm{2 operators}\\
        \textrm{One must be commutative}
    \end{array}$}
\]
We give an example of a two-operator commutative effect system in Section \ref{sec:modeling}.

\section{Effect Quantales}
\label{sec:quantales}

An algebraic characterization of sequential effects, which captures the concrete examples of the previous section, clearly requires distinct operators for sequencing with respect to program order and for giving upper bounds on alternative branches, plus some laws characterizing their interactions. 
The literature of mathematics and computer science is rich with examples of two-operator algebras, but none quite meets our needs as-is.  The closest structures are unital quantales~\cite{mulvey1986,mulvey1992quantisation}, idempotent semirings (also called dioids), and Kleene algebras~\cite{KOZEN1994366}.
Each of these includes an idempotent commutative binary operator (suitable for control flow joins), often called the additive operator; as well as an associative (but possibly \emph{non-commutative}) operator with unit, suitable for sequencing and often called the multiplicative operator; and useful distributive laws relating the two operators. However, each of these also has additional structure which either excludes some examples, or is simply not required for effects.  An additional mismatch to our needs is that these operations are all total, while some examples we consider later benefit from one or both operations being partial (in particular, effects for non-reentrant locking).
Here we define \emph{effect quantales} (so named because our route to identifying them involved weakening the quantale axioms, as described in Appendix \ref{apdx:comparison}), and establish some useful basic properties.
We defer more involved examples to Section \ref{sec:modeling}.

Because existing literature on concrete effect systems uses join semilattices, we will use $\sqcup$ for the commutative operator that produces common upper bounds on alternative control flow paths (i.e., branches of a conditional).
For clarity, we also switch to the suggestive (directional) $\rhd$ for sequencing, rather than the multiplication symbol $\cdot$ or the common practice in work on semiring-like structures~\cite{yetter1990quantales,abramsky1993quantales,galatos2007residuated,kozen1997kleene,pratt1990action} of eliding the multiplicative operator entirely and writing ``strings'' $abc$ for $a\cdot b\cdot c$.
We choose to require that sequencing distribute on both sides over joins, for reasons explained shortly. And finally, we assume both operators are \emph{partial} functions.
We give examples in Section \ref{sec:modeling} of cases from the literature where certain joins or sequencings have no reasonable result, and should therefore be left undefined.

\begin{definition}[Effect Quantale]
\label{def:eq}
An \emph{effect quantale} $Q=(E,\sqcup,\rhd,I)$ is a partial (binary) join semilattice $(E,\sqcup)$ with partial monoid\footnote{Specifically, a monoid in which the composition is partial, but is required to be defined when either operand is the unit element.} $(E,\rhd,I)$, such that $\rhd$ distributes over joins in both directions ---
$a\rhd(b\sqcup c)=(a\rhd b)\sqcup(a\rhd c)$ and
$(a\sqcup b)\rhd c = (a\rhd c)\sqcup(b\rhd c)$ --- when either side is defined (i.e., if one side is defined, then so is the other).
\end{definition}
As is common with algebraic structures, there are many ways to describe this structure instead as a removal of parts of another definition (e.g., an upper unital partial binary quantale, or an unbounded idempotent semiring without the zero requirement), or as a composite of other definitions (e.g., a partial-join-semilattice-ordered partial monoid).  For brevity, and because of our route to proposing it (see Appendix \ref{apdx:quantales}), we will simply use ``effect quantale.''

As is standard in lattice theory, we induce the partial order $x\sqsubseteq y\overset{\mathsf{def}}{=} x\sqcup y = y$ from the join operation, which ensures the properties required of a partial order.
We extend this to expressions $e_1\sqsubseteq e_2$, defined as if $e_2$ is defined then so is $e_1$ and their evaluations are ordered appropriately.

We will use the semilattice to model the standard effect hierarchy, using the induced partial order for subeffecting.  The (non-commutative) monoid operation $\rhd$ will act as the sequential composition.
Intuitively, the unit $I$ is an ``empty'' effect, which need not be a bottom element.
Our structure is in some sense slightly weaker than a join semilattice, but still stronger than a partial order: not all joins are defined, but if the join \emph{is} defined the result is the least element above both arguments.  We will continue to simply refer to it as a join semilattice for brevity, though we will emphasize the partial nature of the operations when comparing against systems using total operations.
As we will explore in detail in Section \ref{sec:semantics}, this is a more restrictive notion of partial order than related abstract frameworks for sequential effects, and a middle-ground with respect to sequencing (other frameworks treat composition as either a relation, or a total function, not a partial function).
Using a partial join-semilattice both makes effect quantales a direct extension of the common (total) join-semilattice model satisfied by most commutative effects, and makes it straightforward to axiomatize the behavior using (partial) algebraic laws.

An important general contrast between effect quantales and related algebraic structures (quantales, idempotent semirings, Kleene Algebras) is that quantales do not require a greatest or least element, and do not require a zero or nilpotent element for sequencing/multiplication.  In general, the concrete effect systems we shall study often have no natural top or bottom element (Sections \ref{sec:locking} and \ref{sec:crit}). Those with natural bottom elements may not use them as zero elements for sequencing (Section \ref{sec:atomicity_quantale}), and those with nilpotent elements sometimes have them as the greatest element (Section \ref{sec:atomicity_quantale}).  Synthetic top or bottom elements with relevant properties could be added, but this would clutter every system studied in Section \ref{sec:modeling} with additional elements present only to satisfy the equations.

Distributivity of the product ($\rhd$) over joins ($\sqcup$) is worth remarking on, because it is a stronger requirement than other abstract sequential effect systems require. The most obvious benefit is that having more equivalences can permit simpler specifications, even (later) in the presence of effect variables.  More critically, it is necessary to preserve validity with respect to basic compiler optimizations and common refactorings.  The right distributivity law corresponds to the basic compiler optimizations of tail merging~\cite{maher2006merging} --- if conditional branches end with the same operations, moving a single copy of those operations to after the control flow join to reduce code size --- and tail splitting (or duplication)~\cite{rock2004architecture,gregg2001comparing} (its inverse, which increases code size but may reduce the number of jumps by inlining function returns).  
Basically, they assume the equivalence
\[ \mathsf{if}~e~(e_1;\;e_3)~(e_2;\;e_3) \equiv (\mathsf{if}~e~e_1~e_2);\;e_3\]
and the right distributivity law demands these programs have equivalent effects as well.
Other abstract effect systems do not require this, thus permitting (in principle) effect systems unsound with respect to very basic optimizations.  Developers may also perform analogous refactorings themselves, and having this change the effect of the code would be quite surprising to most --- this in fact shows up when comparing graded monads to effect quantales (Section \ref{sec:katsumata}).
The left distributivity law similarly corresponds to possible refactorings or possible results of block merging~\cite{maher2006merging}.

\label{sec:quantale-props}

Effect quantales inherit a rich equational theory of semilattice-ordered monoids and extensive results of ordered algebraic systems in general~\cite{birkhoff,galatos2007residuated,blyth2006lattices,fuchs2011partially}, providing many ready-to-use (or at least, ready-to-adapt-the-proof) properties for simplifying complex effects, and giving rise to other properties more interesting to our needs.

One such example is the expected monotonicity property: that sequential composition respects the partial order on effects.  The monotonicity proof for complete lattices~\cite[Ch.~14.4]{birkhoff} (there called \emph{isotonicity}) carries over directly to effect quantales because it requires only binary joins:
\begin{proposition}[$\rhd$ Monotonicity]
    \label{prop:isotonicity}
In an effect quantale $Q$, $a\sqsubseteq b$ and $c\sqsubseteq d$ implies that $a\rhd c \sqsubseteq b\rhd d$.
\end{proposition}
\begin{proof}
Because
$b\rhd d = b\rhd(c\sqcup d) = (b\rhd c) \sqcup (b\rhd d)$, we know
$b\rhd c \sqsubseteq b\rhd d$ by the definition of $\sqsubseteq$. Repeating the reasoning:
\mbox{$b\rhd c = (a\sqcup b)\rhd c = (a\rhd c) \sqcup (b\rhd c)$,}
so $a\rhd c \sqsubseteq b\rhd c$.  The partial order is transitive, thus
$a\rhd c \sqsubseteq b\rhd d$
\end{proof}

We will also find several corollaries of this fact useful when reasoning about whether or not certain operations are defined.

\begin{corollary}[Monotone $\rhd$ Undefinedness]
    \label{coro:undef_mono}
In an effect quantale $Q$ where $a\sqsubseteq b$ and $c\sqsubseteq d$,
if $a\rhd c$ is undefined, then $b\rhd d$ is undefined.
\end{corollary}
\begin{proof}
    Assume (for contradiction) that $b\rhd d$ is defined. Then by Proposition \ref{prop:isotonicity} $a\rhd c$ is defined, contradicting the top-level assumptions.
\end{proof}
\begin{corollary}[Antitone $\rhd$ Definedness]
    \label{coro:def_anti}
In an effect quantale $Q$ where $a\sqsubseteq b$ and $c\sqsubseteq d$,
if $b\rhd d$ is defined, then $a\rhd c$ is also defined.
\end{corollary}

Join satisfies similar properties with similar proofs:

\begin{proposition}[$\sqcup$ Monotonicity]
    \label{prop:join-isotonicity}
In an effect quantale $Q$, $a\sqsubseteq b$ and $c\sqsubseteq d$ implies that $a\sqcup c \sqsubseteq b\sqcup d$.
\end{proposition}
\begin{corollary}[Monotone $\sqcup$ Undefinedness]
    \label{coro:join_undef_mono}
In an effect quantale $Q$ where $a\sqsubseteq b$ and $c\sqsubseteq d$,
if $a\sqcup c$ is undefined, then $b\sqcup d$ is undefined.
\end{corollary}
\begin{corollary}[Antitone $\sqcup$ Definedness]
    \label{coro:join_def_anti}
In an effect quantale $Q$ where $a\sqsubseteq b$ and $c\sqsubseteq d$,
if $b\sqcup d$ is defined, then $a\sqcup c$ is also defined.
\end{corollary}

We generally summarize the corollaries as ``undefinedness is upward-closed'' and ``definedness is downward-closed'' for both operators.

A useful general construction is the \emph{product} of two effect quantales, which supports examples such as those alluded to in Section \ref{sec:bg}. Later we give examples of systems from the literature which are products of effect quantales (Section \ref{sec:atomicity_quantale}).
\begin{definition}[Products of Effect Quantales ($\otimes$)]
The product $Q\otimes R$ of effect quantales $Q$ and $R$ is given by the product of the respective
carrier sets. Other operations are lifted pointwise to each half
of the product, with the composite operations defined only when both constituent operations are defined. Unit is $(I_Q,I_R)$. \begin{itemize}
    \item $E = E_Q\times E_R$ \item $a\sqcup b = ((\mathsf{fst}~a)\sqcup_Q(\mathsf{fst}~b), (\mathsf{snd}~a)\sqcup_R(\mathsf{snd}~b))$ when both component joins are defined.
\item $a\rhd b = ((\mathsf{fst}~a)\rhd_Q(\mathsf{fst}~b), (\mathsf{snd}~a)\rhd_R(\mathsf{snd}~b))$ when both component sequencings are defined.
\item $I=(I_Q,I_R)$
\end{itemize}
\end{definition}

\subsection{Subsumption of Traditional Commutative Effect Systems}
\label{sec:subsume}
An important litmus test for a general model of sequential effects is that it should subsume traditional single-operator commutative effects as a special case.
This not only implies consistency of effect quantales with traditional effect systems, but ensures implementation frameworks for sequential effects (based on effect quantales) would be adequate for implementing commutative systems as well.
\begin{lemma}[Subsumption of Traditional Commutative Effects]
\label{lem:subsume}
For every bounded join semilattice with top $L=(E,\vee,\top,\bot)$, there is an effect quantale defined by $Q=(E,\vee,\vee,\bot)$.
\end{lemma}
\begin{proof}
$Q$ easily satisfies all laws of the effect quantale definition, as the join semilattice structure is trivially a monoid, and joins distribute over themselves.
\end{proof}

We make no further direct use of this construction (Lemma \ref{lem:subsume}) in this paper, aside from showing in Section \ref{sec:closure_operator} that our construction of iteration operators for sequential effects applies sensibly to this construction as well.  But it implies that the primary results (type safety in Section \ref{sec:soundness}) extend to treatment of traditional commutative effects as well as a special case of effect quantales.

Effect quantales also include the possibility of two-operator commutative effect systems, where both operators are commutative yet they are distinct operations, unlike the traditional case above.  We do not prove an analogous subsumption of these commutative effect systems, because there is no broadly established characterization of these either; they exist but are relatively rare, so it is difficult to generalize from examples. Section \ref{sec:must} shows the main example is an effect quantale.  We posit that \emph{commutative effect quantales} (those with distinct operators, both commutative, akin to commutative semirings) offer a suitable characterization, but it is difficult to evaluate without a richer body of concrete two-operator commutative effect systems for comparison.

\section{Modeling Prior Sequential Effect Systems with Effect Quantales}
\label{sec:modeling}
Many of the axioms of effect quantales are not particularly surprising given prior work on
sequential effect systems; one of this paper's contributions is recognizing and demonstrating that these axioms are
sufficiently general to capture many prior instances of sequential effect systems. We show here a number of prominent examples, ranging from relatively small algebras to rich behavioral effects.

\subsection{Locking with Effect Quantales}
\label{sec:locking}
A common class of effect systems is that for reasoning about synchronization --- which locks are held at various points in the program.  Most effect systems for this problem rely on scoped synchronization constructs, for which a bounded join semilattice is adequate --- the runtime semantics match lock acquire and release operations. (\textsc{RCC/Java} is the classic example of this approach~\cite{Abadi2006}.)  Here, we give an effect quantale for flow-sensitive tracking of lock sets including recursive acquisition.
The main idea is to use a multiset of locks (modeled by $\mathcal{M}(S)=S\rightarrow\mathbb{N}$, where the multiplicity of a lock is the number of claims a thread has to holding the lock --- the number of times it has acquired said lock) for the locks held before and after each expression.  We use $\emptyset$ to denote the empty multiset (where all multiplicities are 0).  We use join ($\vee$) on multisets to produce least upper bounds on multiplicities, union ($\cup$) to perform addition of multiplicities, and subtraction ($-$) for zero-bounded subtraction (least multiplicity is 0).

\begin{definition}[Synchronization Effect Quantale $\mathcal{L}$]
    \label{def:locking}
An effect quantale $\mathcal{L}$ for lock-based synchronization with explicit mutex acquire and release primitives is given by:
\begin{itemize}
\item $E=(\mathcal{M}(L)\times\mathcal{M}(L))$ for a set $L$ of possible locks (the lock claims before and after execution).
\item $(a,a')\sqcup (b,b') = (a\vee b,a'\vee b')$ when both effects acquire and release the same set of locks the same number of times: ${b-b'}={a-a'}$ and ${b'-b}={a'-a}$.  Otherwise, the join is undefined.
\item $(a,a')\rhd(b,b')$ is $(c,c')$ where
\begin{itemize}
    \item $c=a\cup(b-a')$ ($c$ is the lock holdings presumed by $a$, plus those presumed by $b$ but not provided by $a'$)
    \item $c'=(((c-(a-a'))\cup(a'-a))-(b-b'))\cup(b'-b)$ ($c'$ is $c$ less the locks released by the first action, plus the locks acquired by the first action, less the locks released by the second action, plus the locks acquired by the second action)
\end{itemize}
\item $I=(\emptyset,\emptyset)$
\end{itemize}
\end{definition}
Intuitively, the pair represents the sets of lock claims before and after some action, which models lock acquisition and release.
$\sqcup$ intuitively requires each ``alternative'' to acquire/release the same locks, while the set of locks held for the duration may vary (and the result assumes enough locks are held on entry --- enough times each --- to validate either element).
This can be intuitively justified by noticing that most effect systems for synchronization require, for example, that each branch of a conditional may access different memory locations, but reject cases where one branch changes the set of locks held while the other does not (otherwise the lock set tracked ``after'' the conditional will be inaccurate for one branch, regardless of other choices).
The resulting partial order on effects essentially allows adding the same (multi)set of lock acquisitions to be added to the pre- and post-condition lock multisets.  Note that there is no sensible join of effects that acquire and release different sets of locks, making the join partial.
Sequencing two lock actions, roughly, pushes the locks required by the second action to the precondition of the compound action (unless such locks were released by the first action, i.e.\ in $a-a'$), and pushes locks held after the first action through the second --- roughly a form of bi-abduction~\cite{calcagno2011compositional}.

With this scheme, lock acquisition for some lock $\ell$ would have (at least) effect $(\emptyset,\{\ell\})$, indicating that it requires no locks to execute safely, and terminates holding lock $\ell$.  A release of $\ell$ would have swapped components --- $(\{\ell\},\emptyset)$ --- indicating it requires a claim on $\ell$ to execute safely, and gives up that claim.  Sequencing the acquisition and release would have effect $(\emptyset,\{\ell\})\rhd(\{\ell\},\emptyset)=(\emptyset,\emptyset)$.  Sequencing acquisitions for two locks $\ell_1$ and $\ell_2$ would have effect $(\emptyset,\{\ell_1\})\rhd(\emptyset,\{\ell_2\})=(\emptyset,\{\ell_1,\ell_2\})$, propagating the extra claim on $\ell_1$ that is not used by the acquisition of $\ell_2$.  This is true even when $\ell_1=\ell_2=\ell$ --- the overall effect would represent the recursive acquisition as two outstanding claims to hold $\ell$: $(\emptyset,\{\ell,\ell\})$.

A slightly more subtle example is the acquisition of a lock $\ell_2$ just prior to releasing a lock $\ell_1$, as would occur in the inner loop of hand-over-hand locking on a linked list: $(\textsf{acquire}\;\ell_2;\textsf{release}\;\ell_1)$ has effect $(\emptyset,\{\ell_2\})\rhd(\{\ell_1\},\emptyset)=(\{\ell_1\},\{\ell_2\})$.  The definition of $\rhd$ propagates the precondition for the release through the actions of the acquire; it essentially computes the minimal lock multiset required to execute both actions safely, and computes the final result of both actions' behavior on that multiset.

While use of sets rather than multisets would be appealing, and variants of this are possible (we see an example shortly), the interaction of true sets with substitution is quite subtle.
We explore this subtlety further in Section \ref{sec:relwork}, but for now consider what happens when supplying the same actual lock argument for two formal lock parameters that are both acquired.  With sets, substitution can behave poorly (substituting the same lock for multiple variables loses information about the number of lock acquisitions when tracked as a mere set), while with multisets the multiplicities can simply be added.

\subsection{An Effect Quantale for Atomicity}
\label{sec:atomicity_quantale}
One of the best-known sequential effect systems is Flanagan and Qadeer's extension of \textsc{RCC/Java} to reason about atomicity~\cite{flanagan2003atomicity}, based on Lipton's theory of \emph{reduction}~\cite{lipton75} (called \emph{movers} in the paper).
The full details of the movers would be a substantial and lengthy digression from our purpose, but the essential ideas were developed for a simpler language and effect system in an earlier paper~\cite{flanagan2003tldi}, for which we give an effect quantale.

\begin{figure}[t]
\vspace{-0.5cm}
\begin{center}
$
\begin{minipage}{0.14\textwidth}
\begin{tikzpicture}
\node(Top){$\top$};
\node(A)[below of=Top]{$A$};
\node(R)[below right of=A]{$R$};
\node(L)[below left of=A]{$L$};
\node(B)[below right of=L]{$B$};
\draw(B)--(L);
\draw(B)--(R);
\draw(R)--(A);
\draw(L)--(A);
\draw(A)--(Top);
\end{tikzpicture}
\end{minipage}
\qquad
\begin{array}{|c|ccccc|}
\hline
; & B & L & R & A & \top\\
\hline
B & B & L & R & A & \top\\
R & R & A & R & A & \top\\
L & L & L & \top & \top & \top\\
A & A & A & \top & \top & \top\\
\top & \top & \top & \top & \top & \top\\
\hline
\end{array}
$
\end{center}
\caption{Atomicity effects~\cite{flanagan2003tldi}: lattice and sequential composition.}
\label{fig:atomicity_lattice}
\end{figure}

The core idea is that in a well-synchronized (i.e., data race free) execution, each action of one thread can be categorized by how it commutes with actions of other threads: a left ($L$) mover commutes left (earlier) with other threads' actions (e.g., a lock release), a right $R$ mover commutes later (e.g., lock acquire), a both $B$ mover commutes either direction (e.g., a well-synchronized field access).  A sequence of right movers, then both-movers, then left-movers \emph{reduces} to an atomic action ($A$).  Repeating the process wrapping movers around an atomic action can again reduce to an atomic action, verifying atomicity for even non-trivial code blocks including multiple lock acquisitions.  As a regular expression, any sequence of movers matching the regular expression $(R^*B^*)^*A(B^*L^*)^*$ reduces to an atomic action.  Effect trace fragments of this form demarcate expressions that evaluate as if they were physically atomic.

\begin{definition}[Atomicity Effect Quantale $\mathcal{A}$]
    \label{def:atomicity}
The effect quantale $\mathcal{A}$ for Flanagan and Qadeer's simpler system~\cite{flanagan2003tldi} can be given as:
\begin{itemize}
\item $E=\{B,L,R,A,\top\}$. \item $a\sqcup b$ is defined according to the lattice given by \citet{flanagan2003atomicity} (Figure \ref{fig:atomicity_lattice}). \item $a\rhd b$ is defined according to Flanagan and Qadeer's $;$ operator (Figure \ref{fig:atomicity_lattice}). \item $I=B$
\end{itemize}
\end{definition}

Flanagan and Qadeer also define an iterator operator on atomicities, used for ascribing effects to loops whose bodies have a particular atomicity.
We defer discussion of iteration until Section \ref{sec:iteration}, and will revisit this example there.

Of course the atomicity effect quantale alone is insufficient to ensure atomicity, because atomicity
depends on correct synchronization.  The choice of effect for each program expression is not
insignificant, but full atomicity checking requires the {product} of the synchronization and
atomicity effect quantales, to track locking and atomicity together.

\subsection{Trace Sets}
\label{sec:trace_sets}
Koskinen and Terauchi study the power of sequential effect systems to both verify safety properties as well as propagate liveness properties from an oracle~\cite{Koskinen14LTR}.  The main idea is to track a pair of possibly-infinite sets: one set of finite traces over an alphabet $\Sigma$ of events (elements of $\Sigma^*$), and a set of infinite traces over $\Sigma$ (elements of $\Sigma^\omega$).  This also carries the structure of an effect quantale:

\begin{definition}[Koskinen-Terauchi Trace Set Effect Quantales]
    \label{def:traceseq}
    A Koskinen-Terauchi \emph{trace set effect quantale} over a set $A$ of events --- written $\mathcal{KT}(A)$ --- is given by the following:
    \begin{itemize}
        \item $E=\mathcal{P}(A^*)\times\mathcal{P}(A^\omega)$
        \item $(a,b)\sqcup(c,d)=(a\cup c, b\cup d)$
        \item $(a,b)\rhd(c,d)=(a\cdot c,b\cup(a\cdot d))$
\item $I=(\{\epsilon\},\emptyset)$
    \end{itemize}
    where the pairwise concatenation of two sets $-\cdot-$ is given by $X\cdot Y=\{xy \mid x\in X\land y\in Y\}$.
    It is easy to verify that this satisfies the effect quantale laws.
\end{definition}
Koskinen and Terauchi also make use of a meet operator for intersection on effects, which is uncommon among effect systems, and only necessary in their system to combine effects derived from core typing rules with effects derived from liveness oracles.

This is an important example for three reasons.  First, it demonstrates that effect quantales admit very powerful effect systems.  
The finite traces are adequate to express any safety property over the events in question.
Liveness properties --- and the infinite trace sets --- are more subtle; a key part of Koskinen and Terauchi's contributions where they propose the effects above~\cite{Koskinen14LTR} is establishing how the effect system above can \emph{propagate} information from a liveness oracle (without some sort of termination analysis, effect systems remain too weak to prove liveness on their own).

Second, it provides a counter-example to a conjecture made in an earlier version of this work~\cite{ecoop17}. That paper conjectured that a slightly different iteration operator was defined for all meaningful effect quantales, but that was incorrect.
Appendix \ref{apdx:comparison} gives more detailed comparison, but briefly the original proposal required properties that were too strong for behavioral effects like traces, which expose \emph{internal} behavior of computations.  Section \ref{sec:iteration} refines the original iteration construction (essentially removing an unnecessary requirement) so it is defined in cases like this as well.

Third, in Section \ref{sec:iter_examples} we will use this to demonstrate the difference between handling only safety or also liveness properties when considering iteration.
Our formal development in Section \ref{sec:soundness} only addresses safety properties enforced by effect quantale, but it is important to note that effect systems for liveness still satisfy the axioms for effect quantales.

\subsection{History Effects}
\label{sec:history_effects}
\citet{skalka2008types,Skalka2008} study \emph{history effects} over a set of events.  The syntax of the effects resembles a process algebra (in a sense the papers on the topic make precise), and the effects have a denotation as sets of finite traces (essentially the first component of the trace set effects above, though proposed much earlier).  These effects also yield an effect quantale:

\begin{definition}[History Effect Quantale]
    \label{def:histeq}
    The \emph{history effect quantale} over a set of events $A$ --- written $\mathcal{H}(A)$ --- is given by the following:
    \begin{itemize}
        \item Effects $E$ are the well-formed (effect variables appear bound by $\mu$) elements generated by the following grammar~\cite{skalka2008types}: \[H ::= \epsilon \mid \textsf{ev}[i] \mid (H;H) \mid (H|H) \mid \mu h\ldotp H \mid h\]
        The history effects include empty effects, the occurrence of event $i$ (written $\textsf{ev}[i]$ assuming $i\in A$), sequential composition, non-deterministic choice, a least fixed point operator, and effect variables.
        \item $a\sqcup b = (a|b)$
        \item $a\rhd b = (a;b)$
\item $I=\epsilon$
    \end{itemize}
    Technically the set $H$ is taken quotiented by an equivalence relation $\approx$ that relates history effects that denote the same sets of finite traces ---, i.e., $a|b\approx b|a$ so they would be considered equivalent in operations.  It is with respect to this quotienting that the operations above satisfy the effect quantale axioms.
\end{definition}
The definition of operators given here does not make use of recursive history effects ($\mu$), but we will make use of them in Section \ref{sec:iteration} when considering iteration.

The primitive events $\textsf{ev}[i]$ are intended as security-related events.  The original intent of this class of systems was to use the history effects to bound the behavior of some code, and then apply a model checker to validate that the code obeyed a certain security policy.
Later versions of the approach have used slightly different sets of events~\cite{Skalka2008,skalka2020types} (such as distinguishing privilege checks from actions requiring privileges), which are also expressible via effect quantales.

Unlike the other examples in this section, effect polymorphism for history effects has been studied, including both a mostly-traditional prenex quantification form to support Hindley-Milner-style inference~\cite{skalka2008types}, as well as more sophisticated forms to support use in object-oriented languages~\cite{Skalka2008} in a way that decouples effect overrides from inheritance mechanisms.  We translate the former into our framework in Section \ref{sec:modeling2} (setting aside the fixpoint combinator), while we do not consider the latter in this work.
Note that like other prior work mixing sequential effects and effect polymorphism, the above effect operators are total.

\subsection{Finite Trace Effects}
\label{sec:finite_traces}
Because trace sets may be harder to build intuition from, and understanding of history effects may be obscured by lack of familiarity with the syntax, we give here a simpler effect quantale over an alphabet $\Sigma$:
\begin{definition}[Finite Trace Effects]
    \label{def:fintraceeq}
    Effects tracking sets of (only) finite event traces can be described by an effect quantale with
    \begin{itemize}
        \item $E=\mathcal{P}(\Sigma^*)$ \item $X\rhd Y = X\cdot Y$ (pairwise concatenation of sets, again $X\cdot Y=\{xy \mid x\in X\land y\in Y\}$)
        \item $X\sqcup Y=X\cup Y$
        \item $I=\{\epsilon\}$
\end{itemize}
    For an alphabet $\Sigma$, we will overload the notation $\mathcal{P}(\Sigma^*)$ to mean not only sets of finite words over $\Sigma$, but also the corresponding effect quantale.
\end{definition}

This is closely related to the previous two behavioral effect quantales.
This is the first projection of the trace set effect quantale $\mathcal{KT}(\Sigma)$.
For a fixed set $A$ of events of interest, we can define an alphabet that can represent groundings of Skalka et al.'s history effects $\mathcal{H}(A)$ as the set of primitive event effects: $\Sigma_{H(A)}=\{\textsf{ev}[i] \mid i\in A\}$.
Skalka et al.~\cite{Skalka2008} actually give a denotation $\llbracket-\rrbracket : H \rightarrow \mathcal{P}(\Sigma_{H(A)}^*)$.  
The details are out of scope for us, but this embedding is an effect quantale homomorphism (defined formally in Definition \ref{def:morphism}).

\subsection{Concurrent ML}
\label{sec:cml}
Setting aside parallel composition (which we do not study), Nielson and Nielson's communication
effect system for Concurrent ML~\cite{nielson1993cml} (later elaborated with Amtoft~\cite{amtoft1999}) is similar to Skalka's, though conceived much earlier.  Their \emph{behaviors}
act as trace set abstractions, with sequencing and non-deterministic choice (union) acting as an
effect quantale's monoid and join operations. 
(They also include a separate parallel composition of behaviors we do not model, discussed in Section \ref{sec:joinoids}.)  
Their subtyping rules for behaviors imply the
required distributivity laws.  The definition is sufficiently similar to history effects that we omit a formal definition to avoid both redundancy and notational confusion.

\subsection{Non-reentrant Critical Sections}
\label{sec:crit}

In proposing a maximally general framework for sequential effect systems, \citet{tate13} gives a motivating example of an effect system that tracks ownership of one global non-reentrant lock. He calls this effect system \textsf{Crit} (for critical section), which we can describe as an effect quantale:
\begin{definition}[Crit Effects]
    \label{def:crit-effects}
    The motivating example \citet{tate13} gives for the generality of productors and effectors (consequently, productoids and effectoids) corresponds to the following effect quantale:
    \begin{itemize}
        \item $E=\{\varepsilon,\mathsf{locking},\mathsf{unlocking},\mathsf{critical},\mathsf{entrant}\}$
        \item $\sqcup$ is given according to the very sparse Hasse diagram in Figure \ref{fig:crit-effects}. The ordering is the reflexive closure of $\varepsilon<\mathsf{critical}$ and $\varepsilon<\mathsf{entrant}$; $\mathsf{locking}$ and $\mathsf{unlocking}$ are incomparable with all other effects.
        \item $\rhd$ is given in Figure \ref{fig:crit-effects}.
        \item $I=\varepsilon$
    \end{itemize}
\end{definition}
\begin{figure}[t]
\begin{center}
\begin{tikzpicture}[node distance=1.5cm]
\node(neutral){\fbox{$\varepsilon$}};
\node(critical)[above left of=neutral]{\fbox{$\mathsf{critical}$}};
\node(locking)[left of=critical]{\fbox{$\mathsf{locking}$}\hspace{2em}};
\node(entrant)[above right of=neutral]{\fbox{$\mathsf{entrant}$}};
\node(unlocking)[right of=entrant]{${}\qquad{}$\fbox{$\mathsf{unlocking}$}};
\draw(neutral)--(critical);
\draw(neutral)--(entrant);
\end{tikzpicture}

$
\begin{array}{|c|ccccc|}
\hline
\rhd & \mathsf{locking} & \mathsf{unlocking} & \mathsf{critical} & \mathsf{entrant} & \varepsilon \\
\hline
\mathsf{locking} & - & \mathsf{entrant} & \mathsf{locking} & - & \mathsf{locking} \\
\mathsf{unlocking} & \mathsf{critical} & - & - & \mathsf{unlocking} & \mathsf{unlocking} \\
\mathsf{critical} & - & \mathsf{unlocking} & \mathsf{critical} & - & \mathsf{critical} \\
\mathsf{entrant} & \mathsf{locking} & - & - & \mathsf{entrant} & \mathsf{entrant} \\
\varepsilon & \mathsf{locking} & \mathsf{unlocking} & \mathsf{critical} & \mathsf{entrant} & \varepsilon \\
\hline
\end{array}
$

\end{center}
\caption{Sequencing for Tate's critical section effects. $-$ represents an undefined result for composition.}
\label{fig:crit-effects}
\end{figure}

Tate's \textsf{Crit} effect structure rejects acquiring a held lock or releasing an un-held lock --- forcing some sequencings to be undefined in Figure \ref{fig:crit-effects}.  Acquiring and releasing \emph{almost} cancel out: their composition is \emph{not} the unit element $\varepsilon$, because this would violate associativity:
$\mathsf{locking}\rhd\mathsf{locking}\rhd\mathsf{unlocking}$
must give the same result regardless of which $\rhd$ is simplified first: if $\mathsf{locking}\rhd\mathsf{unlocking}$ were equal to $\varepsilon$, then reducing the right sequencing first would yield $\mathsf{locking}\rhd\varepsilon=\mathsf{locking}$, while reducing the left first would be undefined.
This is why the effect quantale and Tate's original effectoid requires distinguishing effects for code that uses the lock completely (both acquiring and releasing) from code that does not use the lock at all.

While Tate's framework is more general than ours (we discuss the relationship further in Section \ref{sec:tate}), and \textsc{Crit} does not use the full flexibility of his framework, the fact that we can give his motivating example in our framework suggest that relatively little expressivity useful for concrete effect systems is lost by using effect quantales.
We discuss the differences further in Section \ref{sec:tate}, but note that Tate used this example in part because despite its simplicity, it could not be expressed in prior generic models of effect systems~\cite{atkey2009parameterised,wadler2003marriage,Filinski1999,Filinski2010}.

\subsection{Deadlock Freedom via Lock Levels}
\label{sec:dlf}

Suenaga gives a sequential effect system for ensuring deadlock freedom in a language with unstructured locking primitives~\cite{suenaga2008type}. This is the closest example we know of to our new lockset effect quantale (Definition \ref{def:locking}).  However, Suenaga's lock tracking is structured a bit differently from ours: he tracks the state of a lock as either explicitly present but unowned (by the current thread), or owned by the current thread, thus not reasoning about recursive lock acquisition. 

Suenaga uses a flow-sensitive type judgment $\Gamma\vdash s \Rightarrow \Gamma' \& \mathit{lev}$, which tracks both type-invariant information (e.g., $x$ is a lock with static lock level~\cite{safelocking99} $\alpha$) as well as information corresponding to an effect (whether the lock is held before and after the program executes).  Effectively, his effects track the pre- and post-state lock ownership as well as a lower-bound on the locks acquired by an expression ($\mathit{lev}$ is this lower bound, drawn from a total order $\mathsf{Level}$).  Sequencing an effect that acquires locks at level $\alpha$ and above after an effect that still holds locks at levels $\alpha$ and above is undefined.

Suenaga characterizes only the lower bound on levels of acquired locks as the effect, but as we argue in Section \ref{sec:bg}, some of the flow-sensitive tracking should properly be considered an effect as well.
To actually apply the intuition for pushing flow-sensitive information from typing judgments into effects, we must push a subset of the information from the type environment into effects, possibly duplicating some along the way.
We can view information about which locks are held or not (by the current thread) as effects, removing that information from the types of variables referring to locks, leaving lock types as tracking only levels.  However, as part of controlling which effects are sensible, lock level information beyond the lower bound ($\mathit{lev}$) from the type judgment should also be moved into effects.
Rearranging the flow-sensitive lock tracking into part of an effect as suggested above, we can give an effect quantale capturing the core structure of his approach:

\begin{definition}[Suenaga's Deadlock Freedom Effects $\mathcal{DL}(L)$]
    \label{def:dlf_suenaga}
    The core structure of the deadlock freedom effect system of \citet{suenaga2008type} can be formulated as an effect quantale $\mathcal{DL}(L)$ for some set of locks $L$:
    \begin{itemize}
        \item $E=\{ (X,L,Y)\in(L\rightharpoonup(\mathsf{Lvl}\times\mathsf{Ob}))\times\mathsf{Lvl}\times(L\rightharpoonup(\mathsf{Lvl}\times\mathsf{Ob})) \mid \mathsf{dom}(X)=\mathsf{dom}(Y)\land\mathsf{Levels}(X)=\mathsf{Levels}(Y)\land\mathsf{max}(\mathsf{Levels}_{\mathsf{held}}(X))< L \land\mathsf{UniqueHeldLevels}(X)\land\mathsf{UniqueHeldLevels}(Y)\}$
            where 
            \begin{itemize}
                \item $\mathsf{Lvl}=\mathbb{N}\uplus\{\infty\}$
                \item $\mathsf{Ob}=\mathbb{1}|\mathsf{held}$
                \item $\mathsf{UniqueHeldLevels}(A)=\neg\exists x,y,l\ldotp x\neq y\land A(x)=A(y)=(l,\mathsf{held})$
                \item $\mathsf{Levels}_{\mathsf{held}}(A)=\{l \mid \exists x\ldotp A(x)=(l,\mathsf{held}) \}$
            \end{itemize}
        \item $(X,L,Y)\sqcup(X',L',Y')=(X\cup X', \mathsf{min}(L,L'), Y\cup Y')$ when well-formed and $X\triangle Y=X'\triangle Y'$ (i.e., locks tracked in one side of the join argument are neither acquired nor released), taking $\cup$ as compatible union of partial functions (which must agree where both are defined)
        \item $(X,L,Y)\rhd(X',L',Y')=(X\uplus(X'|_{\mathsf{dom}(X')\setminus\mathsf{dom}(X)}), \mathsf{min}(L,L'), Y|_{\mathsf{dom}(Y)\setminus\mathsf{dom}(X)}\uplus Y')$
            when:
            \begin{itemize}
                \item the result is well-formed,
                \item $Y|_{\mathsf{dom}(Y)\cap\mathsf{dom}(X')}=X'|_{\mathsf{dom}(Y)\cap\mathsf{dom}(X')}$
                \item $\mathsf{max}(\mathsf{Levels}_\mathsf{held}(Y\setminus X'))< L'$
                \item $\mathsf{max}(\mathsf{Levels}_\mathsf{held}(X'\setminus Y)) < L$
            \end{itemize}
        \item $I=(\emptyset,\infty,\emptyset)$
    \end{itemize}
\end{definition}
In prose, Suenaga's effects are equivalent to triples of pre- and post-state information of which locks are held, along with a lower bound on the level of locks acquired during an expression's execution.  These triples are only valid if the pre- and post-state information assumes any simultaneously held locks have different levels, and the lower bound on acquired-lock-levels is larger than any locks assumed held initially.
The join simply unions the state information (each side may mention locks unknown to the other) and takes the minimum lock level, but this is only valid if the result satisfies the general constraints on well-formed effects, and neither effect acquires or releases locks unknown to the other (checked by comparing symmetric difference of the partial functions).
A consequence of this join is that there is a least element $(\emptyset,0,\emptyset)$.
Sequencing performs some of the bi-abduction~\cite{calcagno2011compositional} as in $\mathcal{L}(L)$ (moving unmodified variables forward and backward through effects that do not need them), additionally checking that for variables tracked by both effects, the postcondition $Y$ equals the precondition $X'$, that the second effect acquires locks at levels strictly larger than those locks held after the first effect but unused by the second, and similarly that any locks assumed held by the second effect but unused by the first are safe to hold when acquiring (have lower levels than) locks acquired by the first.

\subsection{`Must' Effects}
\label{sec:must}
\citet{mycroft16} give the following example of an effect system which is commutative (is disregards program order), but still has need for multiple operators.
\begin{definition}[Must Effects]
    \label{def:must_effects}
    For a set $X$ of events of interest, define the effect quantale $\mathsf{Must}(X)$ as:
    \begin{itemize}
        \item $E=\mathcal{P}(X)$
        \item $a\rhd b = a\cup b$
        \item $a\sqcup b = a\cap b$
        \item $I=\emptyset$
    \end{itemize}
\end{definition}
These effects capture the behaviors in $X$ that a program \emph{definitely} performs: union collects actions performed in sequence into a set (ignoring program order) but the join used to combine branches is intersection, retaining only the actions performed by \emph{both} branches. In particular, the notion of subeffecting in this case is reversed from the typical use of sets: $a\sqsubseteq b\Leftrightarrow a\supseteq b$, making $X$ the least element.
The distributivity laws follow from the powerset being a distributive lattice under union and intersection.

\section{Iteration}
\label{sec:iteration}

Many sequential effect systems include a notion of iteration, used for constructs like explicit
loops.  The operator for this, usually written as a postfix superscript ${}^*$, gives the overall effect of any
(finite) number of repetitions of an effect.
Our attention is focused on safety properties, which require only finite iteration, though future work exploring soundness for liveness properties would likely use this as a stepping stone to define infinite iteration, as occurs in Kleene $\omega$-algebras~\cite{esik2015continuous,cohen2000separation,laurence2011omega}.

The iteration construct must follow from some fixed point construction on the semilattice.
However, the most obvious approach --- using a least fixed point theorem on effect quantales with a bottom element --- requires a least element, which not all effect quantales have.  The definition of an effect quantale could be changed, and the examples from Section \ref{sec:quantales} could have synthetic bottom elements added, but this complicates the axiomatization and turns out to be unnecessary.

Instead, we detail an approach based on \emph{closure operators} on partially ordered sets in
Section \ref{sec:closure_operator}, which applies to any effect quantale satisfying some mild
restrictions and coincides with manual iteration definitions for prior work.  First, in Section \ref{sec:iter_props}, we motivate a number of required properties for any derived notion of iterating an effect.

\subsection{Properties Required of an Iteration Operator}
\label{sec:iter_props}
Iteration operators must satisfy a few simple but important properties to be useful.  We first list, then explain these properties.

\[\begin{array}{rl}
    \textrm{Extensive} & \textrm{if $e^*$ is defined, then}\;e \sqsubseteq e^*\\
    \textrm{Idempotent} & \textrm{if $e^*$ is defined, then $(e^*)^*$ is also defined, and}\; (e^*)^*=e^*\\
    \textrm{Monotone} & \textrm{if $e^*$ and $f^*$ are defined, then}\; e\sqsubseteq f \Rightarrow e^*\sqsubseteq f^*\\
    \textrm{Foldable} & \textrm{if $e^*$ is defined, then so is $(e^*)\rhd( e^*)$, and}\; (e^*)\rhd(e^*)\sqsubseteq e^*\\
    \textrm{Possibly-Empty} & \textrm{if $e^*$ is defined, then}\; I \sqsubseteq e^*
\end{array}\]

An \emph{extensive} iteration operator ensures one iteration of a loop body is a lower bound on the effect of multiple iterations.  Similarly, iterating an effect should be an upper bound on 0 iterations (it should be \emph{possibly-empty}) and should be an upper bound on extra iterations before or after (it should be \emph{foldable}).  Iteration must also be \emph{monotone} and \emph{idempotent} in the usual ways.  Note that the foldable requirement extends to the case of individual repetitions before or after: because iteration is also extensive and sequencing is monotone in both arguments, these axioms imply $\forall e\ldotp e\rhd e^* \sqsubseteq e^*$ and $\forall e\ldotp e^*\rhd e\sqsubseteq e^*$ (the original requirements from \citet{ecoop17}).

The requirements that iteration should be possibly-empty, extensive, and foldable, are directly related to the dynamic execution of a loop (which may execute 0, 1, or more times), and play a role in our syntactic soundness proof.
Monotonicity is required to interact appropriately with subtyping.  Being idempotent is not absolutely necessary, but is both assumed in prior sequential  effect systems, and naturally true of the constructions we give below.
Finally, the operator is assumed to be partial: as with the join and sequencing operations, sometimes iteration simply doesn't make sense for particular input. For example, it is unclear what it means to iterate a lock acquisition an indeterminate (finite) number of times --- there is no way to match this with an equal number of unlocks.

These considerations lead us the following class of effect quantales:

\begin{definition}[Iterable Effect Quantales]
    \label{def:iterable}
    An \emph{iterable} effect quantale $(E,\sqcup,\rhd,I,*)$ is an effect quantale equipped with an additional operator $(-)^* : E \rightharpoonup E$ that is extensive, idempotent, monotone, foldable, and possibly-empty.
\end{definition}

Trivially, every effect quantale is also an iterable effect quantale: the iteration operator that is undefined everywhere satisfies the laws above.
Thus when we speak of iterable effect quantales, the key element is a choice of a particular iteration operator of interest.
For an effect quantale to be \emph{usefully} iterable, we require a non-trivial iteration operator.  The remainder of this section describes a class of effect quantales where a non-trivial operator that is as precise as possible (in a manner defined shortly) is uniquely determined, and where that determined iteration operator coincides with hand-designed iteration of sequential effects when applied to examples from Section \ref{sec:modeling}.

\subsection{Iteration via Closure Operators}
\label{sec:closure_operator}
For a general notion of iteration, we will work with the notion of a \emph{closure operator} on a poset:
\begin{definition}[Closure Operator~\cite{birkhoff,blyth2006lattices,saraswat1991}]
A (total) closure operator on a poset $(P,\sqsubseteq)$ is a function $f : P \rightarrow P$ that is
\[\begin{array}{rl}
    \textrm{Extensive}& \forall e, e \sqsubseteq f(e) \\
    \textrm{Idempotent}& \forall e, f(f(e)) = f(e)\\
    \textrm{Monotone}& \forall e,e'\ldotp e\sqsubseteq e'\Rightarrow f(e)\sqsubseteq f(e')
\end{array}\]
\end{definition}
Closure operators have several particularly useful properties~\cite{birkhoff,blyth2006lattices,saraswat1991}. Writing ${x\uparrow}=\{ y \mid y\in P \land x \le y \}$ to denote the principal up-set of $x$:
\begin{itemize}
\item Idempotence implies that the range of a closure operator is also the set of fixed points of the operator.
\item Closure operators $f$ on a poset $(P,\le)$ are in bijective correspondence with their ranges $\{f(x) \mid x\in P\}$.  In particular, from the range we can recover the original closure operator by mapping each element $x$ of the poset to the least element of the range that is above that input: \[x\mapsto \mathsf{min}(\{f(x) \mid x\in P\}\cap \{y \mid y\in P\land x\le y\})\]
The properties of closure operators ensure that this is defined (in particular, that the least element used in the definition above exists and is unique, which may not be the case if $f$ is not a closure operator).
\item A given subset $X$ of a poset $(P,\le)$ is the range of a closure operator --- called a \emph{closure subset} --- if and only if for every element $x\in P$ in the poset, $X\cap({x\uparrow})$ has a least element~\cite[Theorem 1.8]{blyth2006lattices}.  (The left direction of the iff is in fact proven by constructing the closure operator as described above.)
\end{itemize}
This means that if we can identify the desired range of our iteration operation (the results of the iteration operator) and show that it meets the criteria to be a closure subset, the construction above will yield an appropriate closure operator, which we can take directly as our iteration operation.  

We will use \emph{partial closure operators}, where the result of iteration may be undefined, because we will construct it for effect quantales where certain combinations may be undefined, such as Tate's Crit (Definition \ref{def:crit-effects}), or where certain upper bounds on repetition do not exist, such as the locking effect quantale (Definition \ref{def:locking}).  The useful closure operator properties above carry over to the partial case. In particular, the bijection between operators and ranges holds for the elements where the operator is defined, and a subset $X$ is a \emph{partial closure subset} if for every $x\in P$, $X\cap({x\uparrow})$ has a least element or is empty, in which case the operator is undefined on $x$.
\citet{saraswat1991} observe the useful intuition that a partial closure operator is a (total) closure operator on a downward-closed subset: if the operator is undefined for some $x$ it is undefined for all greater elements, while if it is defined for some $x$ then it is also defined for all lesser elements.
The partial analogues of the properties above are exactly the extensive, idempotent, and monotone properties we require for iteration. As in the total case, identifying a suitable partial closure subset will determine an operator for iteration.

The natural choice is the set of elements for which sequential composition is (strictly) idempotent.
\begin{definition}[(Strictly) Idempotent Elements]
The set of (strictly) idempotent elements $\mathsf{Idem}(Q)$ of an effect quantale $Q$ is defined as
$\mathsf{Idem}(Q) = \{ a\in Q \mid a\rhd a=a\}$ (the set of idempotent elements of $Q$).
\end{definition}

This is the set targeted in earlier versions of this work~\cite{ecoop17}.
But it turns out to be useful later to relax this slightly:

\begin{definition}[Subidempotent Elements]
The set of \emph{subidempotent} elements $\mathsf{SubIdem}(Q)$ of an effect quantale $Q$ is defined as $\mathsf{SubIdem}(Q) = \{ a \in Q \mid a\rhd a\sqsubseteq a\}$ (the set of \emph{subidempotent} elements of $Q$).
\end{definition}

Note that there are parts of the literature on ordered semigroups and ordered monoids that use the term \emph{idempotent} to mean what we call here \emph{subidempotent}.  We follow the more established terminology going as far back as to Birkhoff~\cite{birkhoff}, who defines an element $a$ as idempotent when $a\rhd a=a$ and subidempotent when $a\rhd a\sqsubseteq a$.
For the remainder of this paper, however, we mostly avoid the term ``idempotent'' to avoid confusion with the desired property of our iteration operator.

Not all subidempotent elements necessarily satisfy the requirement that iteration be possibly-empty --- that all elements in the range of our closure operator must be greater than $I$. So we require the closure operator range to be a subset of $I\uparrow$.
This suggests the following subset as an idealized target range (partial closure subset) of the partial closure operator:
\newcommand{\iter}[1]{\ensuremath{\mathsf{Iter}(#1)}}
\begin{definition}[Iterable Elements]
    The set of \emph{iterable} elements $\mathsf{Iter}(Q)$ of an effect quantale $Q$ is defined as $\mathsf{Iter}(Q)=\{ a \in Q \mid a\rhd a\sqsubseteq a \land I \sqsubseteq a \}$, the set of subidempotent elements above unit.
    This can also be written as the intersection of $I$'s principal up-set and the subidempotent elements: $({I\uparrow})\cap\mathsf{SubIdem}(Q)$.
\end{definition}
Indeed, if all elements of an effect quantale are iterable, the identity function gives an operator satisfying the 5 requirements for iteration.  However, all of the examples from Section \ref{sec:modeling} include additional elements, and we must define the behavior of an operator on those elements as well.  To do so it will be useful to distinguish another key class of elements:
\newcommand{\subiter}[1]{\ensuremath{\mathsf{SubIter}(#1)}}
\begin{definition}[Subiterable Elements]
    The set of \emph{subiterable} elements \subiter{Q} of an effect quantale $Q$ is defined as $\subiter{Q}=\{a \in Q \mid \exists b\in\iter{Q}\ldotp a\sqsubseteq b\}$ (the set of subeffects of iterable elements).
\end{definition}
This is the key subset of elements where iteration's behavior is not immediately obvious (or more precisely, $\subiter{Q}\setminus\iter{Q}$, since the iterable elements are subiterable).  For elements that are not subiterable, iteration must be undefined: the iteration operator must be extensive but also always return iterable elements where it is defined, and only subiterable elements have iterable elements above them.

These considerations, along with the standard result that closure operators are in bijective correspondence with their ranges, suggests the following operator as the definition for our underlying partial closure operator:
\[ x \mapsto \mathsf{min}({x\uparrow}\cap\iter{Q}) \]
We map x to the minimum subidempotent element greater than both $x$ and $I$ --- the minimum iterable element greater than $x$ --- when the intersection is non-empty and has a unique least element with respect to $\sqsubseteq$.\footnote{$\mathsf{min}$ here is the partial operator which returns the unique least element if it exists, and is otherwise undefined. We do not require the effects to have minima, or for even binary meets to exist: we merely require that $\mathsf{min}$ returns an element $y$ if and only if $\forall z\in{x\uparrow}\cap\iter{Q}\ldotp y\sqsubseteq z \land (z\sqsubseteq y\Rightarrow z=y)$.}
For this to define a valid partial closure operator, we must show that the mapping above is defined (in particular, the minimum exists) or is undefined \emph{because the intersection is empty} for each element. The latter case (the intersection is empty) can occur when there is no join of $x$ and $I$, or there are no subidempotent elements greater than (or equal to) $x\sqcup I$.  
This partial function could fail to be a partial closure operator if there exists some $y$ for which ${y\uparrow}\cap{I\uparrow}\cap\mathsf{SubIdem}(Q)$ is non-empty but does not have a unique least element (i.e., it has two or more minimal elements, or an infinite descending chain of subidempotents).
Attempting to use such a partial function as a closure operator would let the minimal elements be their own iteration result when present (which we desire), but the lack of a minimum for subidempotents above $y$ would mean the function was not everywhere-defined on a downward-closed subset (hence, not a partial closure operator).

The requirement that when the intersection is non-empty there is a least element is not necessarily true in all effect quantales, as it is not implied by the effect quantale axioms.  For this reason we identify the subclass of all effect quantales for which the mapping above is a partial closure operator: the \emph{principally iterable} effect quantales.\footnote{Appendix \ref{apdx:comparison} explains the difference between this definition and the narrower range of effect quantales studied in the early version of this work~\cite{ecoop17}.}

\begin{definition}[Principally Iterable Effect Quantale]
An effect quantale $Q$ is \emph{principally iterable} if for each subiterable element there is a unique least iterable element greater than it.
\end{definition}

For effect quantales with this property, the partial function above is a partial closure operator:

\begin{proposition}[Closure for Principally Iterable Effect Quantales]
\label{prop:closure_iterable}
For any principally iterable effect quantale $Q$, \iter{Q} is a partial closure subset.
\end{proposition}
\begin{proof}
If for every $x$, ${x\uparrow}\cap\mathsf{Iter}(Q)$ is empty or has a least element, $\mathsf{Iter}(Q)$ is a partial closure subset. 
For elements that are not subiterable (those with no iterable elements above them) the intersection is necessarily empty. For an element $x$ that is subiterable, because $Q$ is principally iterable the unique least iterable element greater than $x$ will be the least element of the intersection.
\end{proof}

\begin{proposition}[Principal Iteration]
\label{prop:free_closure_operator}
For every principally iterable effect quantale $Q$, the partial function
$ x \mapsto \mathsf{min}({x\uparrow}\cap\mathsf{Iter}(Q)) $
is a partial closure operator satisfying our desired properties.
\end{proposition}
\begin{proof}
That this is well-defined (i.e., the \textsf{min} in the definition exists when its argument is non-empty) and that it is a partial closure operator follows immediately from Proposition \ref{prop:closure_iterable} and the partial closure operator analogue of Blyth's construction~\cite[Theorem 1.8]{blyth2006lattices} of a closure operator from its range.
The mapping above is extensive, idempotent, and monotone because it is a partial closure operator~\cite{blyth2006lattices,birkhoff,saraswat1991}.
It is foldable because the range of the closure operator consists of only subidempotent elements, and foldability is exactly subidempotence.
It is possibly-empty because the closure subset is constructed using only elements of $I$'s principal up-set ($I\uparrow$).
\end{proof}
To add a bit of additional intuition surrounding the partial nature of this operator, we can appeal again to \citet{saraswat1991}'s observation that partial closure operators are (total) closure operators on downward closed subsets.  Let $D=\{ x \mid {x\uparrow}\cap\mathsf{Iter}(Q) \neq \emptyset \}$, and notice that $D$ is downward-closed: if $a\sqsubseteq b$ and $b\in D$, then since ${b\uparrow}\subseteq {a\uparrow}$ and there exist subidempotent elements above $b\sqcup I$, there also exist subidempotent elements greater than $a\sqcup I$, and $a\in D$. Moreover, notice that principal iterability implies a minimum exists for each such set, so $(-)^*$ defines a total closure operator on $D$, which is a downward-closed subset of $Q$: $D$, which is the set of subiterable elements, is the domain of definition for $(-)^*$ in a principally iterable effect quantale.
The characterization in terms of total closure on a downward-closed subset implies that the iteration operators satisfy ``downward-closed definedness'' and ``upward-closed undefinedness'' analogues of the Corollaries \ref{coro:undef_mono}, \ref{coro:def_anti}, \ref{coro:join_undef_mono}, and \ref{coro:join_def_anti} for $\rhd$ and $\sqcup$.

Not all effect quantales are principally iterable (see Section \ref{sec:nonprincipal}), but all principally iterable effect quantales induce an iteration operator via Proposition \ref{prop:free_closure_operator}.
We have yet to encounter an effect quantale with a known interpretation (i.e., which is meaningful for programs) that is not principally iterable; all concrete sequential effect systems in the literature we are aware of correspond to principally iterable effect quantales, so this property appears sensible.  
Later (Section \ref{sec:precise_iteration}) we show that the definition above gives the most precise possible notion of iteration for principally iterable effect quantales (though other less precise iteration operators may be useful for computational reasons like efficiency or inference, and useful iteration operators may exist for non-principally-iterable effect quantales).
  So in
practice these requirements to induce a closure-operator-based iteration appear
unproblematic.

In Section \ref{sec:iter_examples} we show that a number of specific sequential effect systems from the literature have principally iterable effect quantales.  But while we cover a range of examples, it would be preferable to have at least some guarantees about certain \emph{classes} of effect quantales being principally iterable.
Here we give two broad classes of effect quantales which are all principally iterable and together contain all of our examples: finite effect quantales, and effect quantales where the elements above $I$ ($I\uparrow$) have all non-empty meets.

For the first, the following lemma will be useful:
\begin{lemma}[Finite Powers of Subidempotent Elements above Unit Exist]
    \label{lem:finpow_subidem_exist}
    For any subidempotent element $x$ greater than $I$ and natural number $n$, $x^n$ is defined and $I\sqsubseteq x^n\sqsubseteq x$.
\end{lemma}
\begin{proof}
    By induction on $n$. For $n=0$, $x^0=I=x$, and by assumption $I\sqsubseteq x$.
    For $n>0$: by the inductive hypothesis $x^{n-1}$ is defined and $I\sqsubseteq x^{n-1}\sqsubseteq x$.
    Then since $x\rhd x\sqsubseteq x$ (it is subidempotent), by Corollary \ref{coro:def_anti}, since $x\rhd x$ is defined, $x^{n-1}\rhd x = x^n$ is defined, and by Proposition \ref{prop:isotonicity} 
    $I\rhd I\sqsubseteq x^{n-1}\rhd x\sqsubseteq x\rhd x\sqsubseteq x$.
\end{proof}
Note that this lemma applies to all effect quantales, without assuming principal iterability.

\begin{proposition}[Finite Effect Quantales are Principally Iterable]
    \label{prop:finite_iterable}
    If $Q$ is finite, then for any $x\in Q$, ${x\uparrow}\cap{\mathsf{Iter}(Q)}$ has a unique least element or is empty.
\end{proposition}
\begin{proof}

    The basic idea is to adapt the well-known fact that in finite semigroups every element has a strictly idempotent power~\cite{almeida2009representation} to the partial monoid of an effect quantale, and then show that taking the least $n$ such that the $n$th power of a specific element is idempotent gives us the required partial closure subset.

    If $x\sqcup I$ is undefined, the set is empty. Otherwise, assume $x\sqcup I$ is defined.

    Next we split depending on whether or not all powers of $x\sqcup I$ are defined.\footnote{Note this can be determined. Since $Q$ is finite, consider the sequence $(x\sqcup I)^1,\ldots,(x\sqcup I)^{|Q|+1}$. If any element of this $|Q|+1$-length sequence is undefined, not all powers are defined. If all $|Q|+1$ elements are defined, then it reveals a cycle in repeated iterations of $(x\sqcup I)$ (since the sequence is longer than the number of elements in $Q$), and further sequencing with $(x\sqcup I)$ will repeat earlier elements of the sequence.}
    If there is some $n$ for which $(x\sqcup I)^n$ is undefined, then then $(x\sqcup I)$ is not subiterable and therefore the set is empty --- if $(x\sqcup I)$ were less than any iterable element $q\in\mathsf{Iter}(Q)$, by Proposition \ref{prop:finite_iterable} $q^{n}$ is defined, and since definedness is antitone (Corollary \ref{coro:def_anti}), that would imply $(x\sqcup I)^{n}$ should be defined, yielding a contradiction.

    Otherwise all finite powers of $x\sqcup I$ are defined. In this case, we show there is a subidempotent element obtained as a power of $x\sqcup I$.

Since the $Q$ itself is finite (and all powers of $(x\sqcup I)$ are defined) the powers of $(x\sqcup I)$ form a finite semigroup, so by standard results on finite semigroups~\cite{almeida2009representation,moore1902definition,clifford1961algebraic}, there exists some least $n>0$ such that $(x\sqcup I)^{2n}=(x\sqcup I)^n$ (since only finitely many elements can be generated from $x\sqcup I$, there must be some power of $x\sqcup I$ which is strictly idempotent, and since the naturals are totally ordered there is a least such power).  We show this is the least subidempotent element above $x\sqcup I$. For any subidempotent $q$ greater than both $x$ and $I$ --- i.e., $(x\sqcup I)\sqsubseteq q$ --- monotonicity of $\rhd$ establishes that $(x\sqcup I)^n\sqsubseteq q^n$.  Because $q$ is subidempotent and greater than $I$, $q^n\sqsubseteq q$ (Lemma \ref{lem:finpow_subidem_exist}), so transitively $(x\sqcup I)^n\sqsubseteq q$.  
    Since this is true for all subidempotents greater than both $x$ and $I$, and $\sqsubseteq$ is a partial order, $(x\sqcup I)^n$ is necessarily the unique least subidempotent greater than both. This means every subiterable element has a unique least iterable element greater than it.
\end{proof}

\begin{proposition}[Effect Quantales with Non-empty Meets Above Unit are Principally Iterable]
    \label{prop:complete_iterable}
    If all non-empty meets of elements in ${I\uparrow}$ are defined, then for any $x\in Q$, ${x\uparrow}\cap{\iter{Q}}$ is empty or has a unique least element.
\end{proposition}
\begin{proof}
    Consider some element $x$.
    In the case that the intersection is empty (i.e., $x$ is not subiterable), the result is direct.

    Otherwise $x$ is subiterable, so the intersection ${x\uparrow}\cap{\iter{Q}}$ is non-empty. Because all non-empty meets exist, 
    define $q=\bigsqcap({x\uparrow}\cap\iter{Q})$. 
    Because the infinimum was taken over a non-empty subset of ${x\uparrow}\cap{I\uparrow}=(x\sqcup I)\uparrow$, we have that $x\sqcup I$ is a lower bound on the infinimum: $x\sqcup I \sqsubseteq q$.
    If $q$ is itself subidempotent, then $q$ is the unique least element of ${x\uparrow}\cap{\iter{Q}}$.
    Otherwise, this means $q \sqsubset q\rhd q$ strictly.  In this case, consider
    for $r\in {x\uparrow}\cap\mathsf{Iter}(Q)$ (which exist because $x$ is subiterable):
        \[I\sqcup x \sqsubseteq q \sqsubseteq r\]
    and then because $r\rhd r$ is defined (it is subidempotent), sequential composition of any lesser elements is also defined (Corollary \ref{coro:def_anti}), giving:
        \[(I\sqcup x)\rhd(I\sqcup x) \sqsubseteq q \rhd q\sqsubseteq r\rhd r \sqsubseteq r\]
    Thus $q\rhd q$ is a lower bound of any $r\in {x\uparrow}\cap\mathsf{Iter}(Q)$, but \emph{strictly} greater than the purported greatest lower bound $q$, yielding a contradiction, implying that $q$ must be subidempotent.

    Because every element is either not subiterable, or subiterable with a unique least iterable element above it, $Q$ is principally iterable.
\end{proof}
Note that the proof above relies critically on all possibly-infinite (non-empty) meets above $I$ being defined, because one way to fail to have a minimum of the intersection in the closure operator definition (Proposition \ref{prop:free_closure_operator}) is for it to have an infinite descending chain above $I$\@.  In the case of an effect quantale satisfying the descending chain condition (i.e., no infinite descending chain exists), binary (hence finite non-empty) meets would be adequate (since the condition would imply only the existence of multiple incomparable minimal elements would need ruling out), for which a similar argument to that above could be made.

Proposition \ref{prop:finite_iterable} implies that the atomicity (Definition \ref{def:atomicity}) and non-reentrant critical section (Definition \ref{def:crit-effects}) effects are principally iterable.
Proposition \ref{prop:complete_iterable} implies that the history effect (Definition \ref{def:histeq}), trace set (Definition \ref{def:traceseq}), finite trace effect (Definition \ref{def:fintraceeq}), Concurrent ML (Section \ref{sec:cml}), the lock-multiset (Definition \ref{def:locking})\footnote{Note that because there is no top element in $\mathcal{L}$, there is no meet of the empty set, so requiring only non-empty meets is key to applying Proposition \ref{prop:complete_iterable} to $\mathcal{L}$.}, deadlock freedom (Definition \ref{def:dlf_suenaga})\footnote{The meets are slightly non-obvious, but take the maximum lock level and intersection of the loop-invariant lock tracking sets (which therefore discards inconsistent assumptions about particular locks' levels).}, and must analysis (Definition \ref{def:must_effects}) effect quantales are all principally iterable.

These proofs tell us how to find $x^*$ in an effect quantale covered by the respective lemmas --- either trying iterated powers of $x\sqcup I$ until an idempotent power is found or taking the infinimum of ${x\uparrow}\cap{\iter{Q}}$.
In each effect quantale we are aware of, the specific behavior is fairly straightforward: see Section \ref{sec:iter_examples}.
The value of these proofs is in showing the iteration operators we give below (Section \ref{sec:iter_examples}) are not just serendipitously defined, but that our iteration construction has some claim to generality, and can tell us the construction works for certain effect quantales without immediately getting into the details.  

Before continuing with demonstrations of the derived iteration applied to the effect quantales from Section \ref{sec:modeling}, we first prove one additional property about iteration operators in effect quantales:
\begin{lemma}[Iterable Elements are Strictly Idempotent]
    \label{lem:iter_strict_idem}
    In any iterable effect quantale $Q$, every iterable element $x\in\iter{Q}$ is strictly idempotent: $x\rhd x = x$.
\end{lemma}
\begin{proof}
    Every iterable element $x\in\iter{Q}$ is both greater than unit ($I\sqsubseteq x$) and subidempotent ($x\rhd x\sqsubseteq x$). 
    However, applying a unit law and monotonicity of $\rhd$:
    $ x = x \rhd I \sqsubseteq x\rhd x $.
    Since $x\sqsubseteq x\rhd x$ (above) and $x\rhd x \sqsubseteq x$ (since $x$ is iterable) in a partial order, $x\rhd x = x$.
\end{proof}

This means that the effect quantale iteration axioms imply a strengthening of the foldable property, to $x^*\rhd x^*=x^*$ whenever $x^*$ is defined.

\subsection{Iterating Concrete Effects}
\label{sec:iter_examples}
We briefly compare our derived iteration operation to those previously proposed for specific effect systems in the literature. Generally the induced iteration either exactly matches that from prior concrete effect systems or matches intuition for what an iteration operator \emph{should} do for a given system.

\begin{example}[Iteration for Atomicity]
The atomicity quantale $\mathcal{A}$ (Definition \ref{def:atomicity}) is principally iterable, so principal iteration (Proposition \ref{prop:free_closure_operator}) models iteration in that quantale.  The result is an operator that is the identity everywhere except for the atomic effect $A$, which is lifted to $\top$ (i.e., no longer atomic) when repeated.
This is precisely the manual definition Flanagan and Qadeer gave for iteration.  In Section \ref{sec:atomicity_quantale}, we claimed any trace fragment matching a particular regular expression evaluated as if it were physically atomic --- a property proven by Flanagan and Qadeer.  In terms of effect quantales, this is roughly equivalent to the claim that $(R^*\rhd B^*)^*\rhd A\rhd(B^*\rhd L^*)^*\sqsubseteq A$.  With our induced iteration operator matching Flanagan and Qadeer's original, this has a straightforward proof:
\[
\begin{array}{rll}
(R^*\rhd B^*)^*\rhd A\rhd(B^*\rhd L^*)^* & =(R\rhd B)^*\rhd A\rhd(B\rhd L)^* & \textrm{since $R^*=R$, $B^*=B$, $L^*=L$}\\
                                         & =R^*\rhd A\rhd L^* & \textrm{$B$ is unit for $\rhd$}\\
                                         & =R\rhd A\rhd L  & \textrm{since $R^*=R$, $B^*=B$}\\
                                         & = A & \textrm{by definition of $\rhd$}
\end{array}
\]
\end{example}

\begin{example}[Iteration for Traditional Commutative Effect Quantales]
For any bounded join semilattice with top, we have by Lemma \ref{lem:subsume} a corresponding effect quantale that reuses join for sequencing (and thus, $\bot$ for unit), making the sequencing operation commutative.  For purposes of iteration, this immediately makes all instances of this effect quantale principally iterable, as idempotency of join ($x\sqcup x=x$) makes all effects subidempotent.  The resulting iteration operator then simply joins its effect with unit --- and because the unit in these systems is a bottom element, this makes iteration the identity function. This exactly models the standard type rule for imperative loops in traditional commutative effect systems, where unit is a bottom element, so they can reuse the effect of the body as the effect of the loop:
\[
\inferrule{\Gamma\vdash e_1 : \mathsf{bool} \mid \chi_1\\\Gamma\vdash e_2 : \mathsf{unit} \mid \chi_2}{\Gamma\vdash \mathsf{while}(e_1)\{ e_2 \} : \mathsf{unit} \mid \chi_1\sqcup\chi_2 }
\]
For quantales where sequencing is merely the join operation on the semilattice, the above standard rule can be derived from our rule in Section \ref{sec:soundness} by simplifying the result effect:
\[
\chi_1\rhd(\chi_2\rhd\chi_1)^*=\chi_1\rhd(\chi_2\rhd\chi_1)=\chi_1\sqcup(\chi_2\sqcup\chi_1)=\chi_1\sqcup\chi_2
\]
\end{example}

\begin{example}[Loop Invariant Locksets]
For the lockset effect quantale $\mathcal{L}$, the subidempotent elements are all actions that do not acquire or release any locks --- those of the form $(a,a)$ for some multiset $a$, or the diagonal relation on the multisets.
Since $I$ is $(\emptyset,\emptyset)$ (which has no elements below it),
$ I\uparrow\cap\mathsf{SubIdem}(\mathcal{L})=I\uparrow\cap\{ (a,b) \mid a=b \}=\{ (a,b) \mid a=b \}$.
The subidempotent elements above unit have all non-empty meets defined (the minimum multiplicity in the input set for each lock), so by Proposition \ref{prop:complete_iterable} $\mathcal{L}$ is principally iterable.
The resulting closure operator is the identity on the subidempotent elements, and undefined elsewhere.  This is exactly what intuition suggests as correct --- the iterable elements are those that hold the same locks before and after each loop iteration, and attempts to repeat other actions an indeterminate number of times should be invalid.
\end{example}

\begin{example}[Iterating Trace Sets]
    \label{ex:iter_trace_sets}
    In trace sets, the induced iteration operator has the behavior
    \[(A,B)^*=\left((\bigcup_{i\in\mathbb{N}} A^i), \bigcup_{i\in\mathbb{N}} A^i\cdot B\right)\]
    That is, iteration takes effects with finite traces $A$ and infinite traces $B$ to an effect whose finite traces are the concatenation of any sequence of traces in $A$, and whose infinite effects are any repetition of finite traces followed by an infinite trace from $B$.  
    This is clearly the least idempotent element greater than $(A,B)$.
    Since $A^0=\{\epsilon\}$, $I\sqsubseteq(A,B)^*$ as well, making this effect quantale principally iterable.\footnote{Readers of the original paper~\cite{ecoop17} on effect quantales should note this example does \emph{not} satisfy the more restrictive criteria for iteration in that paper. Appendix \ref{apdx:comparison} articulates how this paper's notion of principally iterable generalizes the notion in the original paper.}
\end{example}

An attentive reader may be wondering why the infinite component of Example \ref{ex:iter_trace_sets} does not include $A^\omega$ --- an infinite number of executions of the finite behaviors, corresponding to the loop repeating for an infinite number of iterations, with each individual iteration having finite execution.  This example highlights that additional work is required to treat liveness properties of sequential effect systems in a general manner.  Our results are valid for safety properties, and therefore apply to \emph{finite prefixes} of executions (which may include finite prefixes of infinite executions).  An operator for infinite execution would require a separate notion of infinite iteration --- $\omega$-iteration, as appears in Kleene $\omega$-Algebras~\cite{cohen2000separation,esik2015continuous}. 
While a corresponding extension to the theory of sequential effect systems has clear semantic appeal, from the perspective of applying such effect systems to programs in a conventional programming language like Java, it becomes unclear which operator (finite or infinite) to use for the effect of iteration constructs; all static sequential effect systems from the literature that contain iteration (rather than fixed point) constructs treat only safety properties and use a finite notion of iteration for this reason.
Even those considering liveness properties in general require information outside the type system for this: \citet{Koskinen14LTR} assume an oracle for assigning liveness aspects of effects.

\begin{example}[Iterating History Effects]
    In history effects, the induced iteration operator has the behavior $h^*=\mu s\ldotp (\epsilon|(s;h)|(h;s))$.  Essentially, iterating an effect $h$ may proceed by doing nothing ($\epsilon=I$), or performing $h$ before or after repetition.  Technically only one of the uses of sequencing with $h$ is required as $\mu s\ldotp \epsilon|(s;h) \approx \mu s\ldotp \epsilon|(h;s)$, but this form makes it easier to see that folding holds.
\end{example}
\begin{example}[Iterating Tate's Crit Effects]
    In the \textsf{Crit} effect system given by \citet{tate13} (recalled in Definition \ref{def:crit-effects}), the subidempotent elements are \textsf{critical}, \textsf{entrant}, and $\varepsilon$.  Iterating any of these effects gives the same effect, while iterating \textsf{locking} or \textsf{unlocking} is undefined.
\end{example}
\begin{example}[Iterating Suenaga's Deadlock Effects]
    In Suenaga's $\mathcal{DL}(L)$ (Definition \ref{def:dlf_suenaga}), the subidempotent elements correspond to loop-invariant locks as in our locking effect quantale.  Here, with the extra lock level information, these are all effects of the form $(X,L,X)$.  Since the only valid effects of this form are those where the maximum level of lock held in $X$ is strictly less than $L$, this corresponds to effects describing loop bodies with loop-invariant lock sets, where if the body acquires locks, the new locks are at levels $L$ or greater (in increasing level order), \emph{and those locks are then released} before the end of the body (since the locks held after are exactly those held initially).
\end{example}

\subsection{The Precision of Principal Iteration}
\label{sec:precise_iteration}
The iteration operator given by Proposition \ref{prop:free_closure_operator} is available for interesting classes of effect quantales (Propositions \ref{prop:finite_iterable} and \ref{prop:complete_iterable}), and gives intuitive results for known examples (Section \ref{sec:iter_examples}), including recovering previous manually-designed iteration operators.  
Here we characterize part of why this operator seems so well-behaved: when defined, it is the most precise closure operator satisfying the iteration properties.

\begin{proposition}[Principal Iteration on Principally Iterable Effect Quantales is Maximally Precise]
    \label{prop:maxprecise}
For a principally iterable effect quantale $Q$, the principal iteration operator $(-)^*$ on $Q$ (Proposition \ref{prop:free_closure_operator}) is the most precise operator satisfying the five iteration axioms, in the sense that any other iteration operator $f$ on $Q$ must return coarser results: $\forall x\ldotp f(x)~\textrm{defined}\Rightarrow x^*~\textrm{defined}\land x^*\sqsubseteq f(x)$
\end{proposition}
\begin{proof}
    Assume an additional closure operator $f$ on $Q$ that is extensive, idempotent, monotone, foldable, and possibly-empty.
    We will separately prove that $(-)^*$ is defined everywhere $f$ is, and then that where defined $f$ is less precise than $(-)^*$. 

    All valid iteration operators must be undefined on elements which are not subiterable.  $(-)^*$ is defined on all subiterable elements, which is the largest possible domain of definition for iteration operators, so if $f$ is defined for some $x\in Q$ then $(-)^*$ must be defined there as well.

    Consider some $x$ for which $f(x)$ is defined. The axioms for iteration imply that $x\sqcup I\sqsubseteq f(x)$, and that $f(x)$ is subidempotent --- i.e., an iterable element greater than $x$. By the reasoning above, $x^*$ is defined, and by construction is the unique \emph{least} iterable element greater than $x$, so $x^*\sqsubseteq f(x)$.
\end{proof}

Proposition \ref{prop:maxprecise}, along with the fact that principal iteration recovers the manual constructions of prior work, implies that those systems used the most precise iteration operators possible.

\subsection{A Non-Principally Iterable Effect Quantale}
\label{sec:nonprincipal}
We know that not all effect quantales are principally iterable because the axioms do not imply it, but it is illustrative (and later useful) to exhibit an effect quantale that is not principally iterable.

\begin{example}[A Non-Principally Iterable Effect Quantale]
    \label{ex:non-principal}
    Define $Q_{NP}=(Q,\sqcup,\rhd,I)$ by:\\
    \begin{minipage}{0.7\textwidth}
        \begin{itemize}
        \item $Q = \{ a_i \mid i\in\mathbb{N}\}\uplus\{ b_i \mid i\in\mathbb{N}\}$
        \item $a_i\sqcup a_j = a_{\mathsf{max}(i,j)}$, $b_i\sqcup b_j=b_{\mathsf{min}(i,j)}$, $a_i\sqcup b_j = b_j$. Put differently, the partial order is the total order consisting of two copies of the natural numbers, with one ordered by $\le$ (the $a_i$), the other ordered in reverse ($b_i$), with the reverse order copy greater than the standard order, per the abbreviated Hasse diagram to the right.
        \item $a_i\rhd a_j = a_{i+j}$, $a_i\rhd b_j = b_j = b_j \rhd a_i$, and $b_i\rhd b_j = b_{\mathsf{min}(i,j)}$
        \item $I=a_0$
        \end{itemize}
    \end{minipage}
    \begin{minipage}{0.2\textwidth}
        \[
            \begin{array}{c}
                b_0 \\ \mid \\ b_1 \\ \mid \\ \vdots \\ \mid \\ a_1 \\ \mid \\ a_0
            \end{array}
        \]
    \end{minipage}
\end{example}
In $Q_{NP}$, the subidempotent elements are $\{ b_i \mid i\in\mathbb{N}\}\uplus\{a_0\}$.  However, $Q_{NP}$ is not principally iterable.  \emph{All} of $Q_{NP}$ is subiterable, but elements of the subset $P=\{ a_i \mid i > 0 \}$ lack \emph{least} subiterable elements above them. The reason is that any iterable element above an element of $P$ must be $b_i$ for some $i$, however for any choice of $i$, $b_{i+1}$ is a \emph{lesser} iterable element still greater than all of $P$. 
In particular, $Q_{NP}$ has non-iterable elements ($P$) below an infinite strictly descending chain in \iter{Q_{NP}}. Any effect quantale with a similar structure embedded within it will also fail to be principally iterable.

However, $Q_{NP}$ has an infinite number of valid effect quantale iteration operators:
\begin{proposition}[Infinite Set of Iteration Operators for $Q_{NP}$]
    \label{prop:desc_chain_iter}
    For any $n\in\mathbb{N}$, the following definition of $(-)^{*n}$ on $Q_{NP}$ satisfies the required iteration axioms:
    \[
        \begin{array}{r@{\;=\;}l@{\quad}l}
            a_0^{*n} & a_0 &\\
            a_i^{*n} & b_n & \textrm{for $i>0$}\\
            b_i^{*n} & b_n & \textrm{for $i > n$}\\
            b_i^{*n} & b_i & \textrm{for $i\le n$}
        \end{array}
    \]
\end{proposition}
\begin{proof}
    The operator essentially chooses some $b_n$, and maps all elements below $b_n$ to $b_n$ (except $a_0$), and is the identity mapping above that point.
    This is clearly extensive, monotone, possibly-empty, idempotent, and since every element in the domain is strictly idempotent, also foldable.
\end{proof}
Notice that this describes an infinite descending chain of increasingly precise iteration operators: $\forall n\ldotp (-)^{*n+1}\le (-)^{*n}$.
These are not the only valid iteration operators: any subset of $\{ b_i \mid i\in\mathbb{N}\}$ \emph{possessing a least element} yields a partial closure subset, to which the standard closure operator construction applies: $(-)^{*n}$ is simply the result of applying this construction to the partial closure subset $\{a_0\}\cup\{b_i \mid i \le n\}$.

\section{Effect Polymorphism with Partial Effect Composition}
\label{sec:effpoly}
In our upcoming type safety proof, we include parametric polymorphism over types and effects, using a variant of System F extended with separate kinds for datatypes and effects --- an idea reaching back to the earliest polymorphic effect system~\cite{lucassen88}. Here we recall why we care about effect polymorphism, and highlight some subtleties in integrating it into a calculus with partial effect operators, before developing some machinery to add variables to any given effect quantale.

Effect polymorphism is an essential aspect of code reuse in static
effect systems~\cite{lucassen88,talpin1992polymorphic,rytz12,ecoop13}.  It permits writing functions
whose effects depend on the effect of higher-order arguments.  For example, consider the atomicity
of fully applying the annotated term
\[
    \mathcal{T} = \lambda\ell:\mathsf{lock}\ldotp\Lambda\gamma::\mathcal{E}\ldotp\lambda
    f:\mathsf{unit}\overset{\gamma}{\rightarrow}\mathsf{unit}\ldotp 
        \left(\mathsf{acquire}\;\ell; f(); \mathsf{release}\;\ell \right)
\]
(where $\mathcal{E}$ is the kind of effects).
The atomicity of a full application of term $\mathcal{T}$ (i.e., application to a choice of effect
and appropriately typed function term) depends on the (latent) atomicity of $f$.  For the moment,
assume we track only atomicities (not lock ownership).  
The type of $\mathcal{T}$ is intuitively
\[
\Pi\ell:\mathsf{lock}\overset{B}{\rightarrow}
    \forall\gamma::\mathcal{E}\overset{B}{\rightarrow}
        (\mathsf{unit}\overset{\gamma}{\rightarrow}\mathsf{unit})
            \xrightarrow{R\rhd\gamma\rhd L} \mathsf{unit}
\]
though we have not yet shown how effect variables (i.e., $\gamma$) can be used in effect expressions.
If $f_1$ performs
only local computation, its latent effect can have static atomicity $B$, making the atomicity of
$\mathcal{T}~\ell~[B]\;f_1$ atomic ($A$).  If $f_2$ acquires \emph{and releases} locks multiple times, its static effect
must be $\top$ (valid but non-atomic), making the atomicity of $\mathcal{T}~\ell~[\top]\;f_2$
also valid but non-atomic.

This sort of extension is straightforward when operators on effects are total, as in prior systems with effect polymorphism.  However, effect quantale operators are \emph{partial} in general, so not all substitutions of an effect for an effect variable will be meaningful, because some instantiations of a variable will lead to sequencing or joining two effects whose result is undefined. 

To support effects mentioning effect variables but otherwise drawing on base effects from a given effect quantale $Q$, we must provide a way to extend any effect quantale with a given set of effect variables.  We write $\Xi$ for a set of distinct effect variables $\alpha$, $\beta$, etc., which we can generally take to be the set of effect variables used in a program of interest.

\newcommand{\unsqcup}{\mathbin{\underline{\sqcup}}}

To this end, we take the syntactic language of effect quantale expressions over iterable $Q$ and $\Xi$:
\[
    S ::= Q \mid \Xi \mid S \rhd S \mid S \sqcup S \mid S^*
\]

This provides us with a syntax for writing effects combining elements of $Q$ with undetermined effect variables.  
Note the star superscript here is taken as concrete syntax, not an abuse of the Kleene star in defining the grammar. Recall that all effect quantales are iterable if iteration is taken to always be undefined, so requiring $Q$ to be iterable is not a limitation.  
We use $\chi$ to range over elements of $S$, and specify the grammar above for a particular $Q$ as $S_Q$.
We consider $\Xi$ fixed for any given development.
Lacking binders within syntactic effects, na\"ive syntactic substitution of effect variables is adequate for our needs.  We write the substitution of an element $\chi'$ of $S_Q$ for an effect variable $\alpha\in\Xi$ into some effect $\chi$ as $\chi[\chi'/\alpha]$.
On occasion we will need to distinguish the syntactic operators from those in $Q$, in which case the operators from $Q$ will carry subscripts, or an expression using (only) $Q$'s operators will be postfixed with $\in Q$.

Our formalization will use these \emph{syntactic} effects with variables in the type-and-effect judgment. 
Sequencing and join of effects in the type system will then simply be \emph{syntactic} join and sequencing --- building terms of $S_Q$.  However, this alone will yield large and unwieldy effects when type-checking programs, and loses simplifications from $Q$, so we permit effects to be simplified according to the equivalence in Figure \ref{fig:approx} (or more precisely, the type system will include effect subsumption, which includes equivalence). The rules encode the base effect quantale axioms, plus the iteration axioms (though using Lemma \ref{lem:iter_strict_idem}'s result instead of the foldable axiom), and also reflects simplifications from $Q$ into equivalences on syntactic effects.
From these, we can derive subeffecting on syntactic effects as $\chi\sqsubseteq\chi'\Leftrightarrow\chi\sqcup\chi'\equiv\chi'$.

\begin{figure}
    \begin{mathpar}
    \fbox{$ \chi \equiv \chi$}
    \and
    \inferrule[Eq-Refl]{  }{ \chi \equiv \chi}
    \and
    \inferrule[Eq-Sym]{\chi\equiv\gamma}{\gamma\equiv\chi}
    \and
    \inferrule[Eq-$\rhd$-Cong]{ \chi \equiv \chi' \\  \gamma \equiv \gamma' }{ \chi\rhd\gamma \equiv \chi'\rhd\gamma'}
    \and
    \inferrule[Eq-$*$-Cong]{ \chi \equiv \chi'}{ \chi^* \equiv \chi'^*}
    \and
    \inferrule[Eq-$\sqcup$-Cong]{ \chi \equiv \chi' \\  \gamma \equiv \gamma' }{ \chi\sqcup\gamma \equiv \chi'\sqcup\gamma'}
    \and
    \inferrule[Eq-$\sqcup$-Idem]{ }{ \chi\sqcup\chi\equiv\chi}
    \and
    \inferrule[Eq-Comm]{ }{ \gamma\sqcup\chi \equiv \chi\sqcup\gamma}
    \and
    \inferrule[Eq-$\rhd$-Assoc]{ }{ \chi\rhd(\chi'\rhd\chi'') \equiv (\chi\rhd\chi')\rhd\chi''}
    \and
    \inferrule[Eq-$\sqcup$-Assoc]{ }{ \chi\sqcup(\chi'\sqcup\chi'') \equiv (\chi\sqcup\chi')\sqcup\chi''}
    \and
    \inferrule[Eq-UnitL]{ }{ I\rhd\chi\equiv\chi}
    \and
    \inferrule[Eq-UnitR]{ }{ \chi\rhd I\equiv\chi}
    \and
    \inferrule[Eq-$\rhd$-Simp]{ q\rhd_Q q'= q''}{ q\rhd q'\equiv q''}
    \and
    \inferrule[Eq-$\sqcup$-Simp]{ q\sqcup_Q q'= q''}{ q\sqcup q'\equiv q''}
    \and
    \inferrule[Eq-$*$-Simp]{ q^*= q'\in Q}{  q^*\equiv q'}
    \and
    \inferrule[Eq-DistL]{ }{ \chi\rhd(\chi'\sqcup\chi'') \equiv (\chi\rhd\chi')\sqcup(\chi\rhd\chi'')}
    \and
    \inferrule[Eq-DistR]{ }{ (\chi\sqcup\chi')\rhd\chi'' \equiv (\chi\rhd\chi'')\sqcup(\chi'\rhd\chi'')}
    \and
    \inferrule[Eq-$*$-Idemp]{ }{(\chi^*)^*\equiv\chi^*}
    \and
    \inferrule[Eq-$*$-Extensive]{ }{ \chi\sqcup(\chi^*)\equiv\chi^*}
    \and
    \inferrule[Eq-$*$-Mono]{ \chi\sqcup\gamma\equiv\gamma}{\chi^*\sqcup\gamma^*\equiv\gamma^*}
    \and
    \inferrule[Eq-$*$-Empty]{ }{ I\sqcup\chi^*\equiv\chi^*}
    \and
    \inferrule[Eq-$*$-Fold]{ }{ (\chi^*\rhd\chi^*)\sqcup\chi^*\equiv\chi^*}
    \and
    \inferrule[Eq-Trans]{\chi\equiv\chi' \\ \chi'\equiv\chi''}{\chi\equiv\chi''}
    \end{mathpar}
    \caption{The equivalence relation $\equiv$ on elements of $S_Q$.}
    \label{fig:approx}
\end{figure}

\paragraph{Invalid Effects}
This syntax and equivalence, however, risks losing one advantage of effect quantales, which is that they directly model the cases where effect combinations or operations are invalid. 
Using syntactic effects without checking for validity in $Q$, and using the \emph{syntactic} join and sequencing operators hides partiality: even if $a\rhd_Q b$ is undefined in $Q$, $a\rhd b$ is an element of $S_Q$.
It may also be that for some $a$, $\alpha$, and $b$, there is \emph{no} substitution for $\alpha$ into $a\rhd\alpha\rhd b$ that will yield a valid combination in $Q$.
The syntactic addition of effect variables leads to four distinct classes of syntactic effects:
\begin{itemize}
    \item Trivially valid effects, which contain no effect variables and only perform syntactic operations corresponding to operations defined in the effect quantale --- where all operations are defined for their inputs.  Each of these is equivalent (via $\equiv$) to an element of the original effect quantale. We prove this in Lemma \ref{lem:closed_concrete_or_invalid}.
    \item Effects that are clearly invalid (undefined) in a given effect quantale (such as joining lock-acquiring and lock-releasing effects).  We call these effects \emph{trivially} invalid not because they are necessarily trivial, but because they are invalid according to the rules of the underlying effect quantale itself in a way apparent from the effect quantale alone --- i.e., those which contain a sequencing, join, or iteration on concrete effect arguments for which the corresponding semantic operation is undefined.
    \item Effects that are not trivially invalid, but would be invalid for any possible substitution of effect variables, but which are not equivalent to any trivially-invalid effect.
    Suenaga's deadlock freedom system (Definition \ref{def:dlf_suenaga}) admits this possibility: consider sequencing two effects acquiring a lock $x$ on either side of an effect variable, as in $(x\mapsto(3,\mathbb{1}), 3, x\mapsto(3,\mathsf{held}))\rhd\alpha\rhd(x\mapsto(1,\mathbb{1}),1,x\mapsto(1,\mathsf{held}))$. This effect is invalid, because $x$'s lock level changes, and effects are required to preserve levels.
    We call these \emph{opaquely} invalid, because they are definitely not valid, but checking this requires substantial knowledge of the effect quantale's internals (unlike trivially invalid effects, which require only applying the operators to see if they are defined).
    \item Effects that may be invalid, or not, depending on the instantiation of an effect variable, which we call \emph{possibly-valid}.  For example, the effect $(\emptyset,\{\ell\})\sqcup\alpha$, after substituting a concrete choice of effect for $\alpha$, may become invalid (e.g., substituting an effect that releases a lock), or may be valid (e.g., substituting another effect that acquires only $\ell$ exactly once). These are distinguished from opaquely invalid effects by the fact that there are some substitutions where the result would be valid. 
Trivially valid effects are a subset of possibly-valid effects.  
\end{itemize}

While we do not wish to ascribe programs invalid effects, removing invalid effects from this syntax by construction is complex.  Consider an alternative that makes syntactic effects that are trivially invalid impossible to represent: this makes substitution of effect variables a partial function, which in turn makes substitution of effects into terms partial, significantly complicating syntactic type safety.

Instead, we can make guarantees about syntactic effects relating to validity, particularly the trivially invalid effects.
It is clear they are invalid, but checking for them requires no specific knowledge of the specific effect quantale (unlike opaquely invalid effects) and the existence of possibly-valid effects already requires additional checks after instantiating effect variables.  Because we permit the formation of polymorphic types that cannot be properly instantiated with any argument (i.e., opaquely invalid effects), we informally refer to this approach as \emph{lazy} effect validation.  Because deriving a top-level effect for a closed program requires eventually instantiating all effect variables (our typing judgment will enforce this), any effect of a closed program is either non-trivial and closed (and therefore valid) or trivially invalid. We prove this after first giving a precise definition.

We can define trivially invalid effects formally using $\equiv$:
\begin{definition}[Trivially Invalid Effects]
    An element $\chi\in S_Q$ is \emph{trivially invalid} if there exists some $\chi'\in S_Q$ such that $\chi\equiv \chi'$ and:
    \begin{itemize}
        \item $\chi'$ contains as a subexpression $q \rhd q'$ where $q\rhd_Q q'$ is undefined in $Q$, or
        \item $\chi'$ contains as a subexpression $q \sqcup q'$ where $q\sqcup_Q q'$ is undefined in $Q$, or
        \item $\chi'$ contains as a subexpression $q^*$, where $q^*$ is undefined in $Q$.
    \end{itemize}
    The predicate that an effect is trivially invalid is $\mathsf{TriviallyInvalid}(-)$.
    For convenience, we define $\mathsf{NonTrivial}(X)=\neg\mathsf{TriviallyInvalid}(X)$.
\end{definition}
\textsf{TriviallyInvalid} and \textsf{NonTrivial} are by construction invariant under $\equiv$.

This is a decidable property of syntactic effects: by avoiding the use of possibly-expansive equivalence rules\footnote{\textsc{Eq-UnitL}, \textsc{Eq-UnitR}, \textsc{Eq-$*$-Idemp}, \textsc{Eq-$*$-Extensive}, \textsc{Eq-$*$-Fold}, \textsc{Eq-$*$-Empty}, \textsc{Eq-$\sqcup$-Idem}, \textsc{Eq-$\rhd$-Simp}, \textsc{Eq-$\sqcup$-Simp}, and \textsc{Eq-$*$-Simp}; the other rules either preserve the number of operators, or can only be applied a certain number of times depending on the shape of the expression (for the distributivity rules).} in a way that introduces additional operators, there are only finitely many syntactically equivalent effects.  If any of them contains a join, sequence, or iteration operation applied to only concrete inputs, where the corresponding operation in the effect quantale is undefined, it is trivially invalid, and therefore fails to be nontrivial.

Note that sequencing with, joining with, or iterating trivially invalid effects always produces a trivially invalid effect.

We will later use the kinding relation to restrict which effect variables are valid, which ensures that well-typed top-level programs contain no effect variables in their effects, and are therefore closed. 
Such variable-free effects either simplify to the equivalence class of a ground effect from $Q$, or are trivially invalid:
\begin{lemma}[Closed Syntactic Effects are Concrete or Trivially Invalid]
    \label{lem:closed_concrete_or_invalid}
    For any syntactic effect $\chi\in S_Q$ which is closed ($\mathsf{FV}(\chi)=\emptyset$), then either $\mathsf{TriviallyInvalid}(\chi)$ or there exists a unique $q\in Q$ such that $\chi\equiv q$.
\end{lemma}
\begin{proof}
    By induction on the structure of $\chi$.
    It is direct when $\chi$ is an element of $Q$, vacuous when $\chi$ is an effect variable (since that is not closed). We show the sequencing case, with other cases proceeding similarly.  In this case, $\chi=\chi'\rhd\chi''$. By the inductive hypothesis, both $\chi'$ and $\chi''$ are trivially invalid or are equivalent to a unique element each of $Q$. If either or both are trivially invalid, then so is $\chi$. When $\chi'\equiv q'$ and $\chi''\equiv q''$, consider $q'\rhd_Q q''$.  If the result is undefined, then $\chi$ is trivially invalid. Otherwise there exists a $q\in Q$ such that $q'\rhd_Q q''=q$, so by \textsc{Eq-$\rhd$-Simp} $\chi\equiv q$.
\end{proof}

\subsection{Alternative Approaches to Polymorphism and Partial Effect Operators}
An alternative approach would be to use \emph{bounded} parametric polymorphism.  Since definedness of operators is downward-closed,
checking that an effect is valid assuming the variable is instantiated to its upper bound is sufficient to ensure any lesser instantiation of the variable is also defined.  Bounded effect polymorphism generally increases the expressivity of a polymorphic type system, so this is appealing.  Unfortunately it introduces new complications we view as better suited to more focused exploration in future work.  First, na\"ively introducing bounded effect polymorphism rules out cases like the following type-annotated example because effect quantales generally lack a valid greatest effect:
\[\Lambda\chi::\mathcal{E}\ldotp\Lambda\chi'::\mathcal{E}\ldotp\Lambda\chi''::\mathcal{E}\ldotp\lambda f:\mathsf{unit}\xrightarrow{\chi}\mathsf{unit}\ldotp\lambda g:\mathsf{unit}\xrightarrow{\chi'}\mathsf{unit}\ldotp\lambda h:\mathsf{unit}\xrightarrow{\chi''}\mathsf{unit}\ldotp f(); g(); h()\]
While sequencing three arbitrary unit-to-unit functions of different effects may seem contrived, it is not generally uncommon for code to be parameterized by multiple external pieces of code whose effects interact.

Even if we were willing to discard highly-abstracted code, it presents a practical problem for even simple uses of effect polymorphism in some effect quantales: $\mathcal{L}$ has infinite height for multiple reasons: it has no effect that bounds all lock-preserving effects from above, no effect that bounds all lock-releasing effects (for a particular set of releases) from above, and no effect that bounds all lock-acquiring effects (for a particular set of acquisitions) from above, so any choice of bound would impose arbitrary restrictions on the effect with which a polymorphic function could be used.
So with partial effect operators, having (only) bounded effect polymorphism does not strictly increase expressive power, as it also rules out some programs accepted by our system.
Taking things a step further and using parametric polymorphism with HM(X)-style constraints~\cite{pottier-icfp-98,odersky1999type} related to definedness of certain operations might alleviate this second issue, but brings in the challenges of reasoning about opaquely invalid effects.

The design space for polymorphic effect systems in general is enormous and under-explored, and it is not our goal to thoroughly explore those choices here, but only highlight the challenge, and to demonstrate one compatible and reasonably permissive approach and some of its nuances.
This is why we have only addressed parametric effect polymorphism here, despite examples in the literature of effect polymorphism using bounding, constraints~\cite{Grossman2002Cyclone}, relative effect declarations~\cite{vanDooren2005,rytz12}, qualifier-based effects~\cite{ecoop13}, and other richer forms of effect polymorphism~\cite{Skalka2008}.  Our focus is demonstrating compatibility of effect quantales with effect polymorphism, rather than to produce the final word on the topic.  Such an endeavor would require further work unifying these various forms of polymorphism: while constraint-based non-prenex quantification subsumes most classic forms of parametric polymorphism, relative effect polymorphism and related techniques~\cite{vanDooren2005,rytz12,ecoop13,Skalka2008} have similarity to path-dependent types~\cite{amin2012dependent,rompf2016type} that requires further work to precisely characterize.  

\section{Value-Dependent Effect Quantales}
\label{sec:valdep}
To support value-dependent effects, we must develop machinery to allow effect quantales to be indexed by an arbitrary set of program values, and show that it interacts sensibly with the machinery developed for polymorphic effects.
For this purpose this section develops the theory behind \emph{indexed} effect quantales, for which we require the corresponding notion of an effect quantale homomorphism.

\begin{definition}[Effect Quantale Homomorphism]
    \label{def:morphism}
An \emph{effect quantale homomorphism} $m$ between two effect quantales $Q$ and $R$ is a total function between the carrier sets which:
\begin{itemize}
    \item refines partial monoid composition: $\forall x,y,z\in Q\ldotp x\rhd y = z\Rightarrow {m(x)}\rhd {m(y)}\sqsubseteq m(z)$
    \item refines partial join: $\forall x,y,z\in Q\ldotp x\sqcup y = z\Rightarrow {m(x)}\sqcup {m(y)}\sqsubseteq m(z)$
    \item coarsens unit: $I_R\sqsubseteq m(I_Q)$
\end{itemize}
\end{definition}

Effect quantale homomorphisms essentially model the idea of embedding one effect system into another, which may have a more granular partial order than the original.
The requirement that morphisms map unit to something possibly greater than the unit in the target effect quantale may be surprising when compared to their refinement of join and sequencing. This is required for many inclusions to be morphisms: consider an extension $\mathcal{A}_N$ of the atomicity effect quantale $\mathcal{A}$ (Definition \ref{def:atomicity}) that adds a new least element $N$ for actions that are both-movers by means of being pure (e.g., reducing a function application) as opposed to those that are both movers because they are well-synchronized.  $N$ becomes the new unit.  In this case, we would expect the inclusion $\mathcal{A}\hookrightarrow\mathcal{A}_N$ to be a morphism --- but this requires morphisms to allow mapping the unit of $\mathcal{A}$ to something above the unit of $\mathcal{A}_N$.

Some of our definitions are given more concisely in terms of a category:
\begin{definition}[Category of Effect Quantales]
    The category \textbf{EQ} has as objects effect quantales, and as morphisms effect quantale homomorphisms.
\end{definition}
Because we aim to make our work comprehensible to those familiar with only syntactic methods, we will give further definitions both with and without categories.

\begin{definition}[Indexed Effect Quantale]
    \label{def:indexed_eq}
An \emph{indexed effect quantale} is an assignment $Q$ of an effect quantale $Q(X)$ to each set $X$, which functorially assigns an effect quantale morphism $Q(f) : Q(X) \rightarrow Q(Y)$ to any function $f:X\rightarrow Y$ between index sets, satisfying:
\[
    Q(\mathit{id}_X)=id_{Q(X)}
    \quad
    \forall g\in X\rightarrow Y, f\in Y\rightarrow Z\ldotp Q(f\circ g)=Q(f)\circ Q(g)
\]
Alternatively, an indexed effect quantale is a covariant functor $Q : \textbf{Set}\rightarrow\textbf{EQ}$.
\end{definition}

The lock set effect quantale $\mathcal{L}$ we described earlier is in fact an indexed effect quantale, parameterized by the set of lock names to consider.  Likewise, $\mathcal{KT}(A)$ and $\mathcal{H}(A)$ are indexed effect quantales, indexed by a set of events.
Their functorial behavior on functions between index sets is given by applying the function to the index elements mentioned in each effect.  For example, for a function $f:L\rightarrow L'$ between lock sets, $\mathcal{L}(L)(f)(\{\ell\},\{\ell\})=(\{f(\ell)\},\{f(\ell)\})\in\mathcal{L}(L')$ (note that because these are really multisets, this is well defined for any $f$).

Because we are typically interested in effect quantales indexed by sets of runtime values (i.e., singletons), and because the set of well-typed values changes during program execution, we will need to transport terms well-typed under one use of the quantale into another use of the quantale, under certain conditions.  The first is the introduction of new well-typed values (e.g., from allocating a new heap cell), requiring a form of inclusion between indexed effect quantales.  The second is due to substitution: our call-by-value language considers variables to be values, but during substitution some variable may be replaced by another value that was already present in the set.  This essentially collapses what statically appears as two (or more) values into a single value, thus \emph{shrinking} the set of values distinguished inside the quantale.  
We must ensure the effect quantale homomorphisms we use support these cases.
The latter is already ensured by the functoriality condition on indexed effect quantales.

To support inclusion, we require one natural refinement of indexed effect quantales and their induced homomorphisms:
\begin{definition}[Monotone Indexed Effect Quantale]
    \label{def:monotone_ieq}
An indexed effect quantale $Q$ is called \emph{monotone} when for two sets $S$ and $T$ where $S\subseteq T$, the homomorphism $Q(\iota)$ resulting from the inclusion function $\iota:S\hookrightarrow T$ is itself an inclusion $Q(\iota):Q(X)\hookrightarrow Q(Y)$.
Alternatively, an indexed effect quantale $Q$ is monotone if it restricts to a covariant functor $\textbf{Set}_{\mathit{incl}}\hookrightarrow\textbf{EQ}_{\mathit{incl}}$ between the inclusion-only subcategories of sets and effect quantales.
\end{definition}

To support value-dependent effects, our soundness framework will work with monotone indexed effect quantales --- in particular, indexing by syntactic values.

\section{Syntactic Type Safety for Generic Sequential Effects}
\label{sec:soundness}
In this section we give a purely syntactic proof that effect quantales are adequate for syntactic type safety proofs of sequential type-and-effect systems.  We treat only safety, and as mentioned earlier (Section \ref{sec:iteration}) we do not treat liveness in this paper. For the growing family of algebraic characterizations of sequential effects, this is the first soundness proof we know of that is (1) purely syntactic, (2) explicitly treats the indexed versions of the algebra required for singleton effects, (3) explicitly addresses effect polymorphism (which is complicated by partial effect operators), and (4) includes direct iteration constructs.  This development both more closely mirrors common type soundness developments for applied effect systems than the category theoretic approaches discussed in Section \ref{sec:semantics}, and demonstrates machinery which would need to be developed in an analogous way for syntactic proofs using those concepts.  We select a syntactic type safety technique in hopes of making the proof, and use of abstract effect systems, more broadly approachable (most researchers using semantic techniques can read syntactic proofs, but many researchers use only syntactic techniques in work on effect systems).

We give this type safety proof for an \emph{abstract} effect system --- primitive operations, the notion of state, and the overall effect systems are all abstracted by a set of parameters (operational semantics for primitives that may affect the chosen state).
This alone requires relatively little mechanism at the type level, but we wish to not only
demonstrate that effect quantales are sound, but also that they are adequate for non-trivial
existing sequential effect systems.  In order to support such modeling (Section \ref{sec:modeling2}), the type system includes:
\begin{itemize}
    \item parametric polymorphism over types and effects as different kinds
    \item singleton types~\cite{aspinall1994subtyping} --- also known as value-dependent types~\cite{swamy2011secure} --- commonly used for reference types with region tags or lock names, to allow type dependency in the presence of effectful computation~\cite{pedrot2019fire}
    \item (assumed) effect constructors, for constructing effects (e.g., effects mentioning particular locks).
\end{itemize}
The language we study includes no built-in means to introduce a non-trivial (non-unit) effect, relying instead on the supplied primitives.
The language also includes only purely parametric effect polymorphism, as discussed in Section \ref{sec:effpoly}.

We stage the presentation to first focus on core constructs related to effect quantales, then briefly recap machinery from Systems F and F$\omega$ (and small modifications beyond what is standard), before proving type safety.  Section \ref{sec:modeling2} demonstrates how to instantiate the framework to model prior effect systems.

\paragraph{A Note on Abstract Soundness}
\label{para:abstract_soundness}
Before we proceed, it is worth emphasizing that there are \emph{two} commonly used notions of soundness for \emph{abstract} effect systems.
Filinski~\cite{Filinski2010} neatly summarizes a distinction between two branches of work on effects. \emph{Denotational} approaches (e.g., parameterized monads~\cite{atkey2009parameterised}, productors~\cite{tate13}, graded monads~\cite{katsumata14,fujii2016towards}, and graded joinads~\cite{mycroft16}) describe the semantics of effectful computation. \emph{Restrictive} approaches take ambient computational effects for granted, and focus on using effect systems to reason about where subsets of the ambient effects may occur (most applications of effect systems to Java- or ML-like languages~\cite{tldi12,ecoop13,tofte1997region,bocchino09,nielson1993cml,amtoft1999,flanagan2003atomicity,Abadi2006,objtyrace99}).  One could view the restrictive approach as a means of approximating, for a program written in a general monad with many effects, a way of modeling the program in a more nuanced (multi-)monadic semantics.  In this latter branch --- where we would place this work --- the primary concern is not with exact characterization, but with sound \emph{bounding} of possible behaviors.  

The restrictive approach assumes all effects are possible, so soundness results there almost always center on a syntactic consistency criterion: that if an expression with a given effect reduces to another expression, there is a relationship between the dynamically-invoked effects of the reduction and the static effect of the remaining expression that justifies the original static effect as a reasonable over-approximation. (Each system makes this appropriately precise.)  This class of soundness proofs does not necessarily entail that an effect system enforces the \emph{semantic} properties it intends to (e.g., that the atomicity effects accurately characterize atomicity, or that the locking effects accurately characterize the locking behavior of the program).  For concrete sequential effect systems in this group, there are typically additional proofs relating the guarantees about reduction sequences to the actual semantic property intended.  For abstract approaches of this sort~\cite{marino09,ecoop17}, the semantic intent is ignored beyond what is implied directly by connection to instrumented operational semantics. Moreover, the soundness proofs are about soundness in the sense of Wright and Felleisen's ``Syntactic Approach to Type Soundness''~\cite{wright1994syntactic}, which really proves type \emph{safety}, and thus does not address liveness properties (hence our emphasis on type safety rather than soundness throughout the paper).  We present such a proof here in Section \ref{sec:syntactic_safety}.

On the other hand, denotational approaches~\cite{atkey2009parameterised,tate13,katsumata14,mycroft16} inherently capture the semantics: the semantic property is enforced exactly when the denotational semantics guarantee it.  We can compare classes of algebraic structures like join semilattices or effect quantales to these systems (as we do in Section \ref{sec:semantics}) because these are typically given as \emph{pairs} of constructs: an algebraic structure (akin to effect quantales) describing what sorts of equations and coercions should exist between semantic components, \emph{and} the semantics themselves that must behave accordingly (and which intrinsically address matters of safety or liveness).

We prove type safety syntactically, which is adequate to cover what most work on effect systems applied to mainstream languages like Java would consider their soundness concerns.  But in Section \ref{sec:interp} we show how to extend this to certain classes of semantic soundness (for safety properties), without leaving the syntactic framework most familiar to many type system designers.

We describe the parameters to our setup prior to describing the primary type system, but because some parameters depend on parts of the type system,
Figure \ref{fig:judgments} summarizes the various judgments involved in our system, and where they are defined.

\begin{figure}
    \begin{tabular}{|l|l|l|}
        \hline
        Judgment Form & Purpose & Defined in\\
        \hline \hline
        $\Sigma\vdash\Gamma$ & Well-Formed Type Environments & Figure \ref{fig:syntax}\\ \hline
        $\Sigma;\Gamma\vdash e : \tau \mid \chi$ & Expression Typing (assuming $\Sigma$ and $\Gamma$, $e$ has type $\tau$ and effect $\chi$) & Figure \ref{fig:syntax}\\\hline
        $\Sigma;\Gamma\vdash \tau :: \kappa$ & Kinding (of types, and effects as a subset of types) & Figure \ref{fig:kinding-equiv}\\\hline
        $\Sigma;\Gamma\vdash \chi\equiv\chi$ & Syntactic Effect Equivalence & Figure \ref{fig:kinding-equiv}\\\hline
        $\Sigma;\Gamma\vdash \chi\sqsubseteq\chi$ & Subeffecting & Figure \ref{fig:kinding-equiv}\\\hline
    \end{tabular}
    \caption{Judgments in the type system formalization.}
    \label{fig:judgments}
\end{figure}

\subsection{Parameters to the Language}
\label{sec:params}
We parameterize our core language by a number of external features.
First among these is a monotone indexed effect quantale $Q$.  
Since we will index our effect quantales by syntactic values, monotonicity ensures effects remain valid under the creation of new syntactic values (e.g., allocation of new heap cells).
Stability under substitution of variables (syntactic values in a call-by-value calculus) is ensured by the fact that all effect quantale homomorphisms map defined compositions to defined compositions.
Any constant (i.e., non-indexed) effect quantale trivially lifts to a monotone indexed effect quantale that ignores its arguments.  The product construction $\otimes$ lifts in the expected way for two indexed effect quantales sharing a single index set.

We must also abstract over notation for elements of the indexed effect quantale.  We assume a (possibly infinite) family of effect constructors of the form $E(\overline{x}^n)$ where each constructor takes a specified number of elements of the indexing set (which in our framework will always be a set of syntactic values).  Syntactic substitution into these effects is then a matter of mapping substitution across the arguments to the constructor --- in the case of our core language, mapping syntactic substitution of values across the values given as arguments.
For example, for the locking effect quantale (Definition \ref{def:locking}) there may be a range of effect constructors $\mathsf{Locking}^n_m(\ldots)$ that accept $n+m$ arguments, where the first $n$ arguments specify the locks held in the precondition, and the last $m$ specify the locks held in the postcondition (possibly including multiples), such as $\mathsf{Locking}^1_0(x)$ to indicate the effect of code that requires lock $x$ held once and if terminating finishes having released the lock.  Substituting $x$ by a concrete lock location $\ell$ would yield $\mathsf{Locking}^1_0(\ell)$.
Our proof assumes all effects of the given indexed effect quantale can be written in this way (including $I$), but does not require case analysis or (co)induction on effects specified this way.

The remaining \emph{language parameters} beyond $Q$ include choices of program state, primitives (including new values), their types and semantics, and various properties of those parameters.  This added complexity allows proving type safety for a large class of languages with sequential effect systems with a single proof.
Because the syntax of terms and types, dynamic semantics, and static type judgments are all extended by these language parameters, we must carefully manage the dependencies between the core definitions and parameters to avoid circularity.
Such circularity is manageable with sophisticated tools in the ambient logic~\cite{Delaware2013MLC,birkedal2013intensional,atkey2013productive,mogelberg2014type}, but we prefer to avoid them for now.

The language parameters include (ordered such that later items depend only on strictly earlier items):
\begin{itemize}
\item A set $\mathsf{State}$ representing an abstract notion of state, usually denoted by $\sigma\in\mathsf{State}$.  For a pure calculus \textsf{State} might be a singleton set, while other languages might instantiate it to a set representing heaps, accumulators of program history, etc.
\item A set of primitives $P$, with specified arities.  $p\in P$ refers only to names (which extend the term language). 
This includes both operations whose semantics (below) will operate on terms and \textsf{State}s, as well as additional values (next item) that do not interact directly with general terms (e.g., references).
We refer to the arity of a primitive $p$ by $\mathsf{Arity}(p)$
\item A subset $C\subset P$ of additional constants (primitives with arity 0), including run-time values such as heap locations. We will refer to these with the metavariable $c$.
\item A set of type families (type constructors) $T$ for describing the types of primitives.
\item A function $K : T \rightarrow \mathsf{Kind}$, for ascribing a kind to each type in $T$, restricted so the final result is always the kind for a type (not an effect). 
      (Thus, reference types may be modeled by ascribing an appropriate kind to type constructor \textsf{ref}.)
\item A set $\mathsf{StateEnv}$ of partial functions from primitives to types ($\mathsf{StateEnv}\subseteq P\rightharpoonup \mathsf{Type}$), which plays the role of giving types to states in $\mathsf{State}$. $\mathsf{StateEnv}$ must be a partial order $(\mathsf{StateEnv},\le)$ with respect to the natural ordering on partial functions ($f \le g$ if the functions agree where both are defined, and $g$ may be defined on additional points).
The functions are partial because there may be elements of $P$ that are invalid at certain program points, such as unallocated heap locations.
Elements of $\mathsf{StateEnv}$ are constrained such that for primitives whose types are function or quantified types,
the latent effects prior to the full set of arguments being supplied must all be $I$, and only primitives with arity greater than 0 are given functional or quantified types.
To state this precisely, we adopt the notation of $B(x)$ to represent any binder of $x$ ($\forall x::\kappa$ or $\Pi x:\tau$). Then
if a primitive $p$ has a type $\delta(p)$ of the form $B_0(x_0)\xrightarrow{\chi_0}\ldots B_n(x_n)\xrightarrow{\chi_n} \tau$ with $\mathsf{Arity}(p)=n$, then
$\forall i<n\ldotp \chi_i=I$.  This corresponds to ensuring that only when a primitive is fully-applied does any non-trivial computational effect occur (specifically, $\chi_n$ may be non-trivial), and the result must have type $\tau$ (which may itself be the type of a closure).
We use $\Sigma$ to range over elements of \textsf{StateEnv}, following common use of $\Sigma$ for store types in syntactic type safety proofs.
We later refer to these penultimate effects and types: we write $\mathsf{LastEffect}(\Sigma(p))$ for $\chi_n$, and $\mathsf{LastResult}(\Sigma(p))$ for $\tau$.
\item A least element $\delta\in\mathsf{StateEnv}$ for ascribing a type to some primitive that is independent of the state --- i.e., source-level primitive operations and values (but not store references).
\item A \emph{partial} primitive semantics $\llbracket-\rrbracket : \mathsf{Term}\rightarrow\mathsf{State}\rightharpoonup \mathsf{Term}\times Q\times\mathsf{State}$ specifying for some terms and states, the resulting (i) term, (2) runtime effect, and (3) resulting state of reducing that term (intended for and restricted to reducing full applications of primitives). As with elements of $\mathsf{StateEnv}$, the semantics is partial because some syntactic values may be invalid in certain program states (such as dereferencing an unallocated heap location).
\end{itemize}
For type safety, we will also need an additional language parameter and additional requirements on those already mentioned:
\begin{itemize}
\item A set $B\subseteq P$ of primitives which may block, called blocking primitives
\item Primitives of non-zero arity always appear fully-applied.  If a primitive $p\in P$ has specified $\mathsf{Arity}(p)=n$, then every syntactic occurrence of $p$ in the program occurs inside $n$ nested (term or type) applications. (If currying is desired, the primitive's full application can always be wrapped in $\lambda$s and $\Lambda$s.)
\item There is a relation $\vdash \sigma : \Sigma$ for well-typed states.
\item $\llbracket-\rrbracket$ is defined only on full applications of primitive operations, judged according to the current state environment. The term $p~\overline{v}$ is fully-applied in the current state environment $\Sigma$ if $\epsilon;\Sigma\vdash p~\overline{v} : \tau \mid \chi$, the arity of $p$ is equal to $|\overline{v}|$, and $\epsilon;\Sigma\vdash \tau :: \star$.  
In contrast, if the arity of $p$ is \emph{greater} than $|\overline{v}|$, we call the application \emph{incomplete}, and require $\llbracket p~\overline{v}\rrbracket$ be undefined. Together with the term restriction that primitives are always fully applied in source programs, this ensures there is no ambiguity about when to reduce a primitive application.
If the primitive application $p~\overline{v}$ is fully applied under $\Sigma$, then for any $\sigma$ such that $\vdash\sigma:\Sigma$, $\llbracket p~\overline{v}\rrbracket(\sigma)$ must be defined or $p\in B$ ($p$ is a designated blocking operation).  We call this property \emph{primitive progress}.
\item Types produced by $\delta$ must be well-formed in the empty environment, and must not be closed base types (e.g., the primitives cannot add a third boolean, which would break the canonical forms lemma).
\item Effects produced by $\llbracket-\rrbracket$ are valid for the indexed effect quantale parameterized by the values at the call site (i.e., the dynamic effects depend only on the values at the call).
\item When the primitive semantics are applied to well-typed primitive applications and a well-typed state, the resulting term is well-typed (in the empty environment) with argument substitutions applied, and the resulting state is well-typed under some ``larger'' state type:\\
$\begin{array}{l}
\epsilon;\Sigma\vdash p_i\;\overline{v} : \tau \mid \gamma \land \vdash \sigma:\Sigma \land \llbracket p_i\;\overline{v}\rrbracket(\sigma)=(v',\gamma',\sigma') \\
\qquad\Rightarrow \exists \Sigma'\ldotp \Sigma\le\Sigma' \land \epsilon;\Sigma'\vdash v' : \tau[\overline{v}/\mathsf{args}(\delta(p_i))] \mid I \land \vdash\sigma':\Sigma' \land \epsilon;\Sigma\vdash\gamma'\sqsubseteq\gamma
\end{array}$\\
We call this property \emph{primitive preservation}.  Intuitively, it says that under reasonable assumptions, the result value, effect, and new state produced by the primitive semantics agree with the the result type and effect predicted by $\Sigma$, and the new state is updated appropriately.
\end{itemize}

The language parameters and language components are stratified as follows:
\begin{itemize}
    \item The syntax of kinds is closed.
    \item The core language's syntax for terms and types is mutually defined (the language contains explicit type application and singleton types), parameterized by $T$ and $P$.  The latter parameters are closed sets, so the mutual definition is confined to the core.
    \item The type judgment depends on (beyond terms, types, and kinds) $\delta$, $K$, and $\mathsf{StateEnv}$, which in turn depend on the now-defined syntax.
    \item $\mathsf{State}$ may depend on terms, types, and kinds.
    \item The dynamic semantics will depend on terms, types, kinds, \textsf{State}, and $\llbracket-\rrbracket$ (which cannot refer back to the main dynamic semantics).
    \item Primitive preservation and primitive progress depend on the typing relation and state typing.
    \item The type soundness proof will rely on all core typing relations, state typing (which may be defined in terms of source typing), and the primitive preservation and progress properties.
\end{itemize}
Ultimately this leads to a well-founded set of dependencies for the soundness proof, as shown in Figure \ref{fig:deps}.  Note that there are no circularities between the framework and language parameters.

\begin{figure}
    \includegraphics[width=0.75\textwidth]{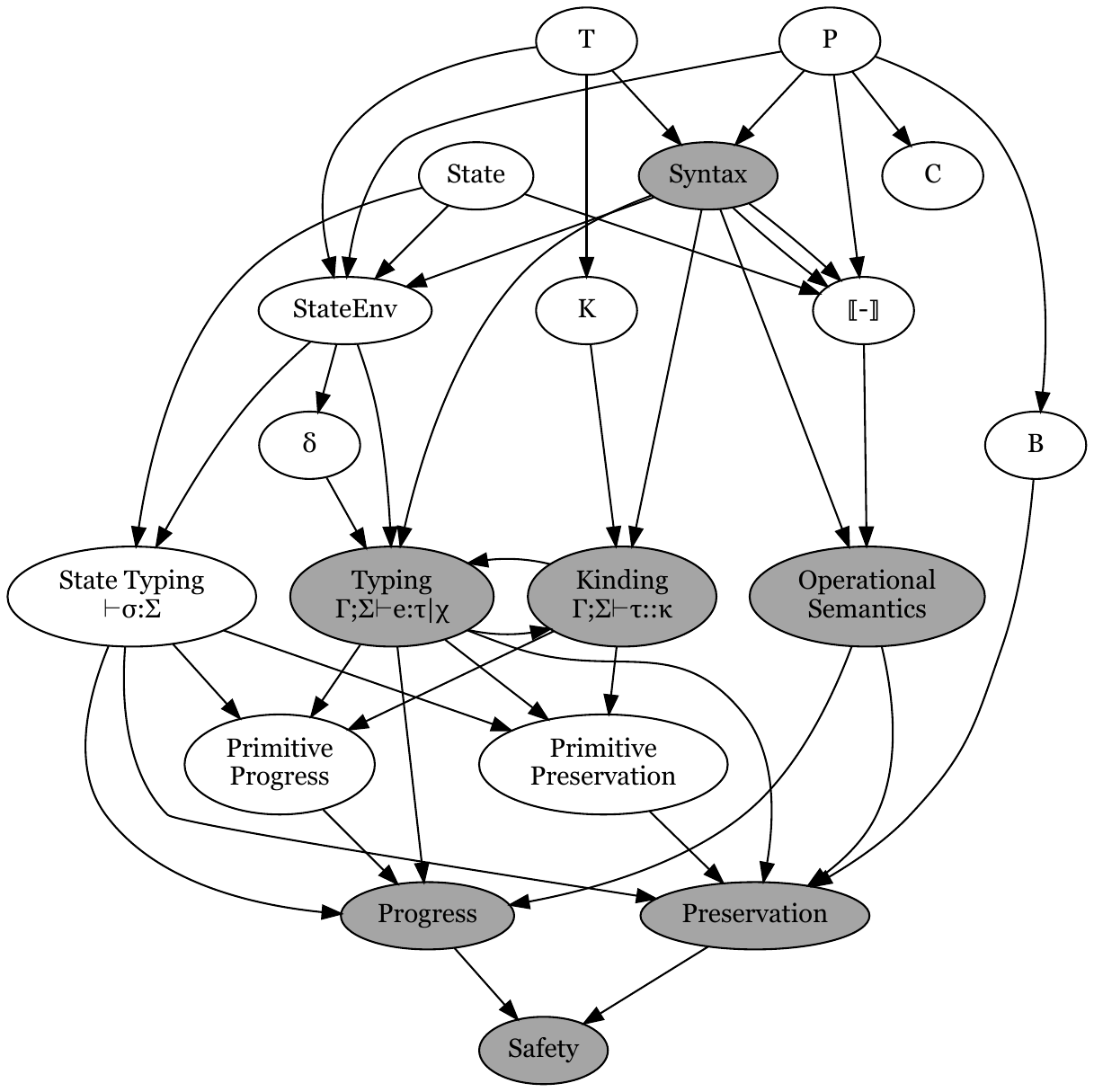}
    \caption{Dependencies between the framework and its parameters, omitting the indexed effect quantale $Q$. Shaded nodes are provided by the framework. Unshaded nodes are parameters.}
    \label{fig:deps}
\end{figure}

\subsection{The Core Language, Formally}

\begin{figure}[t!]
\small
\[\begin{array}{lrcl}
\textsf{Kinds} & \kappa & ::= & \star \mid \mathcal{E} \mid \kappa\Rightarrow\kappa\\
\textsf{Effects} & \gamma,\chi & =& S_{Q(\textsf{Values})} \\ \textsf{Types} & \tau & ::= & T \mid \tau\;\tau \mid \chi \mid \Pi x:\tau\overset{\chi}{\rightarrow}\tau \mid \alpha \mid \mathsf{bool} \mid \forall\alpha::\kappa\overset{\chi}{\rightarrow}\tau \mid \mathsf{unit} \mid \mathcal{S}(v) \\
\mathsf{Terms} & e & ::= & p \mid (\lambda x\ldotp e) \mid e\;e\mid x \mid \mathsf{true} \mid \mathsf{false} \mid \mathsf{if}\;e\;e\;e \mid \mathsf{while}\;e\;e \mid (\Lambda \alpha::\kappa\ldotp e) \mid e[\tau] \mid ()\\
\mathsf{TypeEnv} & \Gamma & ::= & \epsilon \mid \Gamma,x:\tau \mid \Gamma,\alpha::\kappa\\
\mathsf{Values} & v & ::= & c \mid (\lambda x\ldotp e) \mid (\Lambda \alpha::\kappa\ldotp e) \mid x \mid \mathsf{true} \mid \mathsf{false}
\end{array}\]
\begin{mathpar}
\fbox{$\Sigma\vdash\Gamma$}
\and
\inferrule*[left=WF-Emp]{ }{\Sigma\vdash\epsilon}
\and
\inferrule*[left=WF-Val]{\Sigma\vdash\Gamma\\\Gamma;\Sigma\vdash\tau::\star\\x\not\in\Gamma}{\Sigma\vdash\Gamma,x:\tau}
\and
\inferrule*[left=WF-Ty]{\Sigma\vdash\Gamma\\\alpha\not\in\Gamma}{\Sigma\vdash\Gamma,\alpha::\kappa}
\\
\fbox{$\Gamma;\Sigma\vdash e : \tau \mid \gamma$}
\and
\inferrule[T-Prim]{ }{\Gamma;\Sigma\vdash p_i : \Sigma(p_i) \mid I}
\and
\inferrule[T-Var]{\Gamma(x)=\tau}{\Gamma;\Sigma\vdash x : \tau \mid I}
\and
\inferrule[T-Lam]{
    \Gamma;\Sigma\vdash\tau::\star\\
    \Gamma,x:\tau;\Sigma\vdash e : \tau' \mid \gamma
}{
    \Gamma;\Sigma\vdash(\lambda x\ldotp e) : \Pi x:\tau\overset{\gamma}{\rightarrow}\tau' \mid I
}
\and
\inferrule*[left=T-App]{
    \Gamma;\Sigma\vdash e_1 : \Pi x:\tau\overset{\gamma}{\rightarrow}\tau' \mid \gamma_1\\
    \Gamma;\Sigma\vdash e_2 : \tau \mid \gamma_2\\
    x\not\in\mathsf{FV}(\gamma,\tau')\vee\mathsf{Value}(e_2)\\
}{
    \Gamma;\Sigma\vdash e_1\;e_2 : \tau'[e_2/x] \mid \gamma_1\rhd\gamma_2\rhd\gamma[e_2/x]
}
\and
\inferrule*[left=T-If]{
    \Gamma;\Sigma\vdash c : \mathsf{bool} \mid \gamma_c\\
    \Gamma;\Sigma\vdash e_1 : \tau \mid \gamma_1\\
    \Gamma;\Sigma\vdash e_2 : \tau \mid \gamma_2\\
}{
    \Gamma;\Sigma\vdash\mathsf{if}\;c\;e_1\;e_2 : \tau \mid \gamma_c\rhd(\gamma_1\sqcup\gamma_2)
}
\and
\inferrule*[left=T-While]{
    \Gamma;\Sigma\vdash c : \mathsf{bool} \mid \gamma_c\\
    \Gamma;\Sigma\vdash e : \tau \mid \gamma_b\\
}{
    \Gamma;\Sigma\vdash\mathsf{while}\;c\;e : \mathsf{unit} \mid \gamma_c\rhd(\gamma_b\rhd\gamma_c)^*
}
\and
\inferrule*[left=T-Forall]{
    \Gamma,\alpha::\kappa;\Sigma\vdash e : \tau \mid \gamma
}{
    \Gamma;\Sigma\vdash(\Lambda\alpha::\kappa\ldotp e) : \forall\alpha::\kappa\overset{\gamma}{\rightarrow}\tau \mid I
}
\and
\inferrule[T-TyApp]{
    \Gamma;\Sigma\vdash e : \forall\alpha::\kappa\overset{\gamma}{\rightarrow}\tau \mid \gamma_e\\
    \Gamma;\Sigma\vdash \tau' :: \kappa\\
    \Gamma;\Sigma\vdash \tau[\tau'/\alpha] :: \star\\
}{
    \Gamma;\Sigma\vdash e[\tau'] : \tau[\tau'/\alpha] \mid \gamma_e\rhd\gamma[\tau'/\alpha]
}
\and
\inferrule[T-Bool]{b\in\{\mathsf{true},\mathsf{false}\}}{\Gamma;\Sigma\vdash b : \mathsf{bool} \mid I}
\and
\inferrule[T-Unit]{ }{\Gamma;\Sigma\vdash () : \mathsf{unit} \mid I }
\and
\inferrule[T-Sub]{
    \Gamma;\Sigma\vdash e : \tau \mid \chi\\
    \Gamma;\Sigma\vdash \chi\sqsubseteq\chi'\\
    \Gamma;\Sigma\vdash \chi' :: \mathcal{E}
}{
    \Gamma;\Sigma\vdash e : \tau \mid \chi'
}
\\
\fbox{$\sigma,e\rightarrow^{\gamma}_Q \sigma',e'$}
\quad
\inferrule[E-App]{ }{\sigma,(\lambda x\ldotp e)\;v\rightarrow^I_Q \sigma,e[v/x]}
\quad
\inferrule[E-TyApp]{ }{\sigma,(\Lambda\alpha::\kappa\ldotp e)[\tau]\rightarrow^I_Q \sigma,e[\tau/\alpha]}
\quad
\inferrule[E-PrimApp]{
    \llbracket p_i\;\overline{v}\rrbracket(\sigma)=(e',\gamma,\sigma')
}{
    \sigma,p_i\;\overline{v}\rightarrow^{\gamma}_Q \sigma',e'
}
\and
\inferrule[E-IfTrue]{ }{\sigma,\mathsf{if}\;\mathsf{true}\;e_1\;e_2\rightarrow^I_Q \sigma,e_1}
\quad
\inferrule[E-IfFalse]{ }{\sigma,\mathsf{if}\;\mathsf{false}\;e_1\;e_2\rightarrow^I_Q \sigma,e_2}
\quad
\inferrule[E-While]{ }{\sigma,\mathsf{while}\;e\;e_b\rightarrow^I_Q \sigma,\mathsf{if}\;e\;(e_b;\mathsf{while}\;e\;e_b)\;()}
\and
\inferrule[E-PrimArg]{
    \llbracket p_i~\overline{v}\rrbracket(\sigma)\uparrow\\
    \sigma,e\rightarrow^\gamma_Q\sigma',e'
}{
    \sigma,p_i~\overline{v}~e\rightarrow^\gamma_Q\sigma',p_i~\overline{v}~e'
}
\\
\fbox{$\sigma,e\dhxrightarrow{\gamma}_Q \sigma',e'$}
\and
\inferrule{ }{\sigma,e\dhxrightarrow{I}_Q\sigma,e}
\and
\inferrule{
    \sigma,e\dhxrightarrow{\gamma}_Q \sigma',e'\\
    \sigma',e'\rightarrow^{\gamma'}_Q \sigma'',e''
}{\sigma,e{\dhxrightarrow{\gamma\rhd\gamma'}}_Q\sigma'',e''}
\end{mathpar}
\caption{A generic core language for sequential effects, omitting straightforward structural rules from the operational semantics.}
\label{fig:syntax}
\end{figure}

Figure \ref{fig:syntax} gives the (parameterized) syntax of kinds, types, and terms for the core language.  Most of the structure should be familiar from standard effect systems and Systems F and F$\omega$ (with multiple kinds, as in the original polymorphic effect calculus~\cite{lucassen88}), plus standard while loops and conditionals with effects sequenced as in Section \ref{sec:bg}.  We focus on the differences.
Syntactic effects are elements of $S_{Q(\mathsf{Values})}$ --- the syntactic effect quantale expressions over $Q(\mathsf{Values})$, the indexed effect quantale over syntactic values. This provides for value-dependent effects and effect variables, further elaborated shortly.
The standard term syntax is extended by the syntactic language parameters, specifically extending types by the type constructors $T$ and terms by primitives $p\in P$.
One small matter important to the soundness proof: for any value $v$, the effect of $v$ is the unit effect $I$.

\paragraph{Runtime Typing}
Figure \ref{fig:syntax} gives the runtime type system of a syntactic soundness proof with state.
Recall that this permits values ``made valid'' at runtime (like newly allocated heap locations) to be well-typed.  
These typing judgments assume a type context $\Gamma$ containing value and type variable bindings, as well as a $\Sigma\in\mathsf{StateEnv}$ used for typing primitives.
To type check source programs, $\delta$ would be used as the choice of $\mathsf{StateEnv}$.

\paragraph{Value-Dependent Types and Effects}
The language includes a (value-)dependent product (function) type~\cite{swamy2011secure,pedrot2019fire}, which permits program values to be used in types and effects.  This is used primarily through effects --- elements of an indexed effect quantale may mention elements of the index set --- and through the singleton type constructor $\mathcal{S}(-)$, which associates a type with each well-typed program value. These singleton types are slightly unusual in that they do not classify any values, but instead are used to allow framework instantiations to have value-indexed types (such as reference types indexed by a particular guarding lock).  Use of the dependent function space is restricted to syntactic values (which includes variables in our call-by-value language) --- the application rule requires that either the argument is a syntactic value, or the function type's named argument does not appear in the effect or result type.  In the latter case, for concrete types we will use the standard $\tau\overset{\gamma}{\rightarrow}\tau'$ notation.  A minor item of note is that dependent function types and quantified types bind their argument in the function's effect as well as in the result type.  This permits uses such as a function acquiring the lock passed as an argument.

As mentioned above, the language is parameterized by an indexed effect quantale $Q$ to permit value-dependent effects. We assume the intended choice of index set for elements of $Q$ is the set of well-typed values in the current type environment --- including the current state environment.
By $Q(\Gamma;\Sigma)$ we denote $Q$ instantiated with the set of well-typed values under $\Gamma$ and $\Sigma$.
For some $\gamma\in Q(\Gamma,x:\tau;\Sigma)$, we define value substitution into the effect $\gamma[v/x]$ via the effect quantale morphism arising from the term substitution function that replaces $x$ with $v$.  Then  \mbox{$\gamma\in Q(\Gamma,x:\tau;\Sigma)\rightarrow \gamma[v/x]\in Q(\Gamma;\Sigma)$}. This extends naturally to the syntactic indexed effects of $S_{Q(\Gamma;\Sigma)}$.

We work with elements of ${Q}(\mathsf{Values})$ --- $Q$ instantiated with the set of all syntactic values --- rather than ${Q}(\Gamma;\Sigma)$, to avoid problematic circularity\footnote{A careful reader may notice that values may still contain effects (via type application inside a function), so there is still some kind of circularity.  But it is well-founded.  $Q$ is a \emph{co}variant functor, so only uses the index set positively.  And there is no direct recursion at the value level so for any value containing an effect that mentions another value, the value in the effect will be (informally) smaller than the outer value containing the effect.  This could be formalized using sized types~\cite{abelminiagda}, though requiring an extension with addition on sizes~\cite{barthe2008type} in order to formalize substitution (when substituting a value for a variable in a term, the size of the result is the size of the value plus the size of the term that would contain it post-substitution).}. 
Despite relaxing the syntax to work with the full set of syntactic values rather than only well-typed syntactic values, the type system (specifically, kinding) enforces that the effects written correspond to the ${Q}(\Gamma;\Sigma)$ subset of ${Q}(\mathsf{Values})$ by checking that each argument to an effect constructor is a value that is additionally well-typed under $\Gamma$ and $\Sigma$.  Later we show that well-kinded closed non-trivial effects are in fact elements of $Q(\Gamma;\Sigma)$. 

\paragraph{Operational Semantics}
The operational semantics is mostly standard: a labeled transition system over pairs of states and terms, where the label is the effect of the basic step.  For brevity we omit the common structural rules that simply reduce a subexpression and propagate state changes and the effect label in the obvious way.
The only other subtlety of the single-step relation is handling primitives, for which we give both the actual reduction rule \textsc{E-PrimApp} and the structural rule \textsc{E-PrimArg}.
\textsc{E-PrimApp} reduces a fully-applied primitive \emph{the arguments to which are all values or types} by using the primitive semantics $\llbracket-\rrbracket$.
\textsc{E-PrimArg} is the structural rule for reducing primitive arguments to values, rather than requiring primitives to be applied to only syntactic values.  For example, if the standard heap allocation primitive $\mathsf{ref}$ were introduced to the language we would prefer to allow writing $\mathsf{ref}\;e$ rather than requiring expansion to $\mathsf{let}\;x=e\;\mathsf{in}\;(\mathsf{ref}\;x)$.
In both rules, $p\;\overline{v}$ is actually an abuse of notation to mean a series of nested application and type application expressions, where application is only used with value arguments ---  the grammar:
\[Papp ::= p \mid (Papp\;v) \mid Papp[\tau]\]
The operational rule \textsc{E-PrimArg} uses the notation $\llbracket p\;\overline{v}\rrbracket\uparrow$ indicating that the primitive semantics are not defined on that primitive application because it has too few arguments reduced to values.

We also give a reflexive transitive reduction relation $\dhxrightarrow{\gamma}_Q$ which accumulates the effects of each individual step.

\paragraph{Effect Syntax, Equivalence, Subeffecting, and Kinding}
As already discussed, the combination of partial effect operators and effect variables leads to some difficulties avoided in prior generic approaches to sequential effects, which did not address effect polymorphism.  Without effect polymorphism, simply using $Q$ would be sufficient, as derivations requiring undefined compositions would simply be ill-formed.
Section \ref{sec:effpoly} suggests addressing this by using syntactic effect expressions including variables, and this generally works well.  The equivalence relation of Figure \ref{fig:approx} also offers a way to derive subeffecting ($\chi\sqsubseteq\chi'\Leftrightarrow\chi\sqcup\chi'\equiv\chi'$). However, three additional restrictions must be imposed. First, kinding must ensure only bound effect variables are used. Second, kinding must ensure only well-typed values are used in effects, in order to actually work with the $Q(\Gamma;\Sigma)$ subset of $Q(\mathsf{Values})$.
Finally, we must restrict equivalence to only those derivations using well-kinded elements of $S_{Q(\mathsf{Values})}$.

The results of these restrictions appear in Figure \ref{fig:kinding-equiv}. Kinding for effects (kind $\mathcal{E}$) is mostly straightforward structural rules, additionally checking in \textsc{K-Concrete} that all values used in value-dependent effects are well-typed in the current environment (i.e., in the $Q(\Gamma;\Sigma)$ subset of $Q(\mathsf{Values})$), and checking in \textsc{K-Var} that effect variables are bound at the correct kind (or more accurately, any effect expression containing a variable will eventually constrain the variable to be of kind $\mathcal{E}$, which is only possible if $\Gamma$ binds the variable as an effect variable).
Finally, the equivalence relation of Figure \ref{fig:kinding-equiv} (used also to derive subeffecting) carries the type environment and state type through to the simplification rules, which only equates operations of $Q(\mathsf{Values})$ where all parameters to $Q$'s operators are valid in the sense of \textsc{K-Concrete}. This ensures equivalence derivations \emph{only} use well-kinded elements of $Q(\mathsf{Values})$, and since the equivalence rules cannot introduce or hide effect variables, equivalence preserves kinding.

\begin{lemma}[Equivalence Preserves Effect Variables]
    \label{lem:kind_preserves_effvars}
    If $\Gamma;\Sigma\vdash \chi\equiv \chi'$, the set of effect variables appearing $\chi$ and $\chi'$ is the same.
\end{lemma}
\begin{proof}
    By induction on the equivalence derivation.
\end{proof}
Of particular note, if two effects are equivalent and one is closed, this implies the other is also closed.

\begin{lemma}[Kinding of Equivalent Syntactic Effects]
    \label{lem:kind_equiv}
    For any two syntactic effects $\chi$ and $\chi'$ such that $\Gamma;\Sigma\vdash\chi\equiv\chi'$, we have that $\Gamma;\Sigma\vdash\chi::\mathcal{E}\Leftrightarrow\Gamma;\Sigma\vdash\chi'::\mathcal{E}$.
\end{lemma}
\begin{proof}
    By induction on the equivalence derivation. The non-trivial base cases are the simplification rules, where the constraints on the values appearing in the concrete effects correspond to \textsc{K-Concrete}, immediately proving well-kinding of those concrete effects.
\end{proof}
Intuitively, the only way kinding of an effect could break under equivalence would be if equivalence could introduce new unbound effect variables (which would fail \textsc{K-Var}), or new ground effects mentioning syntactic values not known to be well-typed (which would fail \textsc{K-Concrete}).
Equivalence strictly preserves the set of effect variables (Lemma \ref{lem:kind_preserves_effvars}), and this does not appear explicitly in the proof of Lemma \ref{lem:kind_equiv} because no rules directly manipulate or inspect effect variables.
While equivalence may expand an element of $Q(\mathsf{Values})$ into other elements which compose or iterate to produce the original (taking a right-to-left reading of the simplification rules plus symmetry), the antecedents of those rules enforce adequate restrictions to prove kinding using \textsc{K-Concrete} for the ``newly introduced'' elements.

Note that subeffecting \emph{could} introduce new concrete effects or effect variables: if $\Gamma;\Sigma\vdash\chi\sqsubseteq\chi'$, $\chi'$ may contain effect variables and concrete effects not present in $\chi$. This is why the type rule for effect subsumption (\textsc{T-Sub}) explicitly checks that the new upper bound is well-kinded.

\begin{figure}
\begin{mathpar}
\fbox{$\Gamma;\Sigma\vdash \tau :: \kappa$}
\and
\inferrule[K-Param]{T_i\in T }{\Gamma;\Sigma\vdash T_i :: K(T_i)}
\and
\inferrule[K-Var]{\Gamma(\alpha)=\kappa}{\Gamma;\Sigma\vdash\alpha::\kappa}
\and
\inferrule[K-TyCon]{
    \Gamma;\Sigma\vdash\tau::\kappa\Rightarrow\kappa'\\
    \Gamma;\Sigma\vdash\tau'::\kappa
}{
    \Gamma;\Sigma\vdash\tau\;\tau'::\kappa'
}
\and
\inferrule[K-$\Pi$]{
    \Gamma;\Sigma\vdash\tau::\star\\
    \Gamma,x:\tau;\Sigma\vdash\gamma::\mathcal{E}\\
    \Gamma,x:\tau;\Sigma\vdash\tau'::\star
}{
    \Gamma;\Sigma\vdash(\Pi x:\tau\overset{\gamma}{\rightarrow}\tau') :: \star
}
\and
\inferrule[K-Bool]{ }{\Gamma;\Sigma\vdash\mathsf{bool}::\star}
\and
\inferrule[K-Unit]{ }{\Gamma;\Sigma\vdash\mathsf{unit}::\star}
\and
\inferrule*[left=K-Sing]{\Gamma;\Sigma\vdash v : \tau \mid I}{\Gamma\vdash \mathcal{S}(v) :: \star}
\and
\inferrule*[left=K-$\forall$]{
    \Gamma,\alpha::\kappa;\Sigma\vdash\gamma::\mathcal{E}\\
    \Gamma,\alpha::\kappa;\Sigma\vdash\tau :: \star
}{
    \Gamma;\Sigma\vdash \forall\alpha::\kappa\overset{\gamma}{\rightarrow}\tau :: \star
}
\and
\inferrule*[left=K-Concrete]{
    \forall v\in\overline{v}\ldotp \exists \tau\ldotp \Gamma;\Sigma\vdash v:\tau\mid I
}{\Gamma;\Sigma\vdash E(\overline{v}) :: \mathcal{E}}
\and
\inferrule*[left=K-$\rhd$]{
    \Gamma;\Sigma\vdash \chi_1 :: \mathcal{E}\\
    \Gamma;\Sigma\vdash \chi_2 :: \mathcal{E}
 }{\Gamma;\Sigma\vdash \chi_1\rhd\chi_2:: \mathcal{E}}
\and
\inferrule*[left=K-$\sqcup$]{
    \Gamma;\Sigma\vdash \chi_1 :: \mathcal{E}\\
    \Gamma;\Sigma\vdash \chi_2 :: \mathcal{E}
 }{\Gamma;\Sigma\vdash \chi_1\sqcup\chi_2:: \mathcal{E}}
\and
\inferrule*[left=K-$*$]{\Gamma;\Sigma\vdash \chi :: \mathcal{E}}{\Gamma;\Sigma\vdash \chi^* :: \mathcal{E}}
\\
    \fbox{$\Gamma;\Sigma\vdash \chi \equiv \chi$}
    \and
    \inferrule[Eq-Refl]{\Gamma;\Sigma\vdash \chi :: \mathcal{E} }{\Gamma;\Sigma\vdash \chi \equiv \chi}
    \and
    \inferrule[Eq-Sym]{\Gamma;\Sigma\vdash\chi\equiv\gamma}{\Gamma;\Sigma\vdash\gamma\equiv\chi}
    \and
    \inferrule[Eq-$\rhd$-Cong]{\Gamma;\Sigma\vdash \chi \equiv \chi' \\ \Gamma;\Sigma\vdash \gamma \equiv \gamma' }{\Gamma;\Sigma\vdash \chi\rhd\gamma \equiv \chi'\rhd\gamma'}
    \and
    \inferrule[Eq-$*$-Cong]{\Gamma;\Sigma\vdash \chi \equiv \chi'}{\Gamma;\Sigma\vdash \chi^* \equiv \chi'^*}
    \and
    \inferrule[Eq-$\sqcup$-Cong]{\Gamma;\Sigma\vdash \chi \equiv \chi' \\ \Gamma;\Sigma\vdash \gamma \equiv \gamma' }{\Gamma;\Sigma\vdash \chi\sqcup\gamma \equiv \chi'\sqcup\gamma'}
    \and
    \inferrule[Eq-$\sqcup$-Idem]{ }{\Gamma;\Sigma\vdash \chi\sqcup\chi\equiv\chi}
    \and
    \inferrule[Eq-Comm]{ }{\Gamma;\Sigma\vdash \gamma\sqcup\chi \equiv \chi\sqcup\gamma}
    \and
    \inferrule[Eq-$\rhd$-Assoc]{ }{\Gamma;\Sigma\vdash \chi\rhd(\chi'\rhd\chi'') \equiv (\chi\rhd\chi')\rhd\chi''}
    \and
    \inferrule[Eq-$\sqcup$-Assoc]{ }{\Gamma;\Sigma\vdash \chi\sqcup(\chi'\sqcup\chi'') \equiv (\chi\sqcup\chi')\sqcup\chi''}
    \and
    \inferrule[Eq-UnitL]{ }{\Gamma;\Sigma\vdash I\rhd\chi\equiv\chi}
    \and
    \inferrule[Eq-UnitR]{ }{\Gamma;\Sigma\vdash \chi\rhd I\equiv\chi}
    \and
    \inferrule[Eq-$\rhd$-Simp]{
        \forall v\in\overline{v}\cup\overline{v'}\cup\overline{v''}\ldotp \exists \tau\ldotp \Gamma;\Sigma\vdash v:\tau\mid I\\\\
        E(\overline{v})\rhd_Q E'(\overline{v'})= E''(\overline{v''})}{\Gamma;\Sigma\vdash E(\overline{v})\rhd E'(\overline{v'})\equiv E''(\overline{v''})}
    \and
    \inferrule[Eq-$\sqcup$-Simp]{
        \forall v\in\overline{v}\cup\overline{v'}\cup\overline{v''}\ldotp \exists \tau\ldotp \Gamma;\Sigma\vdash v:\tau\mid I\\\\
        E(\overline{v})\sqcup_Q E'(\overline{v'})= E''(\overline{v''})}{\Gamma;\Sigma\vdash E(\overline{v})\sqcup E'(\overline{v'})\equiv E''(\overline{v''})}
    \and
    \inferrule[Eq-$*$-Simp]{
        \forall v\in\overline{v}\cup\overline{v'}\ldotp \exists \tau\ldotp \Gamma;\Sigma\vdash v:\tau\mid I\\
         E(\overline{v})^*= E'(\overline{v'})\in Q(\mathsf{Values})}{\Gamma;\Sigma\vdash  E(\overline{v})^*\equiv E'(\overline{v'})}
    \and
    \inferrule[Eq-DistL]{ }{\Gamma;\Sigma\vdash \chi\rhd(\chi'\sqcup\chi'') \equiv (\chi\rhd\chi')\sqcup(\chi\rhd\chi'')}
    \and
    \inferrule[Eq-DistR]{ }{\Gamma;\Sigma\vdash (\chi\sqcup\chi')\rhd\chi'' \equiv (\chi\rhd\chi'')\sqcup(\chi'\rhd\chi'')}
    \and
    \inferrule[Eq-$*$-Idemp]{ }{\Gamma;\Sigma\vdash(\chi^*)^*\equiv\chi^*}
    \and
    \inferrule[Eq-$*$-Extensive]{ }{\Gamma;\Sigma\vdash \chi\sqcup(\chi^*)\equiv\chi^*}
    \and
    \inferrule[Eq-$*$-Mono]{\Gamma;\Sigma\vdash \chi\sqcup\gamma\equiv\gamma}{\Gamma;\Sigma\vdash\chi^*\sqcup\gamma^*\equiv\gamma^*}
    \and
    \inferrule[Eq-$*$-Empty]{ }{\Gamma;\Sigma\vdash I\sqcup\chi^*\equiv\chi^*}
    \and
    \inferrule[Eq-$*$-Fold]{ }{\Gamma;\Sigma\vdash (\chi^*\rhd\chi^*)\sqcup\chi^*\equiv\chi^*}
    \and
    \inferrule[Eq-Trans]{\Gamma;\Sigma\vdash\chi\equiv\chi' \\ \Gamma;\Sigma\vdash\chi'\equiv\chi''}{\Gamma;\Sigma\vdash\chi\equiv\chi''}
    \\\\
    \fbox{$\Gamma;\Sigma\vdash \chi \sqsubseteq\chi$}
    \and
    \inferrule*[left=Sub-Def]{\Gamma;\Sigma\vdash \chi\sqcup\gamma\equiv\gamma}{\Gamma;\Sigma\vdash\chi\sqsubseteq\gamma}
\end{mathpar}
    \caption{Kinding, equivalence, and subeffecting
    }
    \label{fig:kinding-equiv}
\end{figure}

Kinding for types (of kind $\star$) is standard, including checking that arguments to the singleton type former are well typed values (standard for languages with singleton/value-dependent types~\cite{aspinall1994subtyping,swamy2011secure}).

\paragraph{Kinding and $Q(\Gamma;\Sigma)$}
Earlier we mentioned that while the syntax of effects uses elements of $Q(\mathsf{Values})$ to avoid circularity, the type system really only works with the subset $Q(\Gamma;\Sigma)\subseteq Q(\mathsf{Values})$.  This is informally apparent from \textsc{K-Concrete}, but we can now prove this.
\begin{proposition}[Well-Kinded Closed Effects]
    \label{prop:closed_kinded_eff_Q}
    For any type environment $\Gamma$, state environment $\Sigma$, and effect $\chi$, if $\Gamma;\Sigma\vdash \chi :: \mathcal{E}$, and $FV(\chi)=\emptyset$, and $\mathsf{NonTrivial}(\chi)$, there exists a $E(\overline{v})$ 
    such that $\Gamma;\Sigma\vdash\chi\equiv E(\overline{v})$ and $E(\overline{v})\in Q(\Gamma;\Sigma)$.
\end{proposition}
\begin{proof}
    By induction on the structure of $\chi$, similar to the proof of Lemma \ref{lem:closed_concrete_or_invalid}, picking up membership in $Q(\Gamma;\Sigma)$ from uses of \textsc{K-Concrete}, and exploiting the fact that since $Q$ is monotone, all operations must be defined even for the index set containing only the values appearing in inputs to sequencing, join, or iteration, so must not introduce new syntactic values that are not already well-kinded.
\end{proof}
All effects in a typing conclusion are well-kinded:
\begin{proposition}
    \label{prop:wk_eff_concl}
    For any type environment $\Gamma$, state environment $\Sigma$, term $e$, type $\tau$ and effect $\chi$, if $\Gamma;\Sigma\vdash e : \tau \mid \chi$ then $\Gamma;\Sigma\vdash \chi :: \mathcal{E}$.
\end{proposition}
\begin{proof}
    By induction on the typing derivation, relying on substitution lemmas given shortly showing that term and type substitution preserves kinding.
\end{proof}
This implies that if a typing judgment gives an expression some closed non-trivial effect, then that effect actually is from $Q(\Gamma;\Sigma)$, even though it is constructed in $Q(\mathsf{Values})$:
\begin{proposition}
    \label{prop:closed_eff_concl_Q}
    For any type environment $\Gamma$, state environment $\Sigma$, term $e$, type $\tau$ and effect $\chi$, if $\Gamma;\Sigma\vdash e : \tau \mid \chi$ and $FV(\chi)=\emptyset$ and $\mathsf{NonTrivial}(\chi)$, there exists a $E(\overline{v})$ 
such that $\Gamma;\Sigma\vdash\chi\equiv E(\overline{v})$ and $E(\overline{v})\in Q(\Gamma;\Sigma)$.
\end{proposition}
\begin{proof}
    By Propositions \ref{prop:wk_eff_concl} and \ref{prop:closed_kinded_eff_Q}.
\end{proof}

We will employ Propositions \ref{prop:closed_kinded_eff_Q} and \ref{prop:closed_eff_concl_Q} as coercions from syntactic effects to semantic effects when they are known to be closed and non-trivial.

They allow us to further relate syntactic equivalence and subtyping to semantic equivalence and subtyping.

\begin{lemma}[Closed Syntactic Effect Equivalence]
    \label{lem:closed_syneff_equiv}
    For any two closed non-trivial syntactic effects $\chi$ and $\chi'$ where $\Gamma;\Sigma\vdash\chi :: \mathcal{E}$ and $\Gamma;\Sigma\vdash\chi'::\mathcal{E}$, if $\Gamma;\Sigma\vdash\chi\equiv\chi'$ then $\chi$ and $\chi'$ are semantically equivalent as elements of $Q(\Gamma;\Sigma)$.
\end{lemma}
\begin{proof}
    By induction on the equivalence derivation. Most cases follow from the fact that if $\chi$ and $\chi'$ are well-kinded, then when cases break them into compositions of smaller effects we can obtain kinding and non-triviality for those, which are then equivalent to semantic effects by Proposition \ref{prop:closed_kinded_eff_Q}.  By non-triviality of $\chi$ and $\chi'$,
     the relevant semantic compositions are defined on each side, and equivalent because each rule axiomatizes a semantic equivalence up to definition.
    Only congruence rules use equivalence inductively, and in those cases the inductive hypothesis serves to obtain strict equalities. We show just a couple representative cases here.
    \begin{itemize}
        \item Case \textsc{Eq-$\rhd$-Cong}: Here $\chi=\chi_1\rhd\chi_2$ and $\chi'=\gamma_1\rhd\gamma_2$, where $\Gamma;\Sigma\vdash\chi_1\equiv\gamma_1$ and $\Gamma;\Sigma\vdash\chi_2\equiv\gamma_2$.  Because $\chi$ and $\chi'$ are well-kinded and non-trivial, by inversion on the kinding assumptions we obtain that $\chi_1$, $\chi_2$, $\gamma_1$, and $\gamma_2$ are well-kinded, and by definition of non-triviality we obtain that they are also non-trivial, and that the semantic equivalents of $\chi_1\rhd\chi_2$ and $\gamma_1\rhd\gamma_2$ are defined. Using the inductive hypotheses on the assumed equivalences, we obtain that semantically $\chi_1=\gamma_1$ and $\chi_2=\gamma_2$, and since the semantic $\rhd$ is a partial function, applying it to pairs of equal arguments will yield equal results.
        \item Case \textsc{Eq-$*$-Mono}: Note that this rule syntactically axiomatizes the monotonicity requirement that $x\sqsubseteq y\Rightarrow x^*\sqsubseteq y^*$ when the iterations are defined, using the definition of $\sqsubseteq$ in terms of $\sqcup$.
        In this case, $\chi=\gamma^*\sqcup\gamma'^*$ and $\chi'=\gamma'^*$. As before we can obtain that $\gamma$ and $\gamma'$ are well-kinded and non-trivial, and that their iterations, and join of iterations are semantically defined.
        If we can obtain the equivalence corresponding to the antecedent of \textsc{Eq-$*$-Mono}, this is an ordering relationship we can use, along with knowledge that both sides of the conclusion are defined and the monotonicity of $(-)^*$, to prove the semantic ordering equivalent to the conclusion. But we must first know that $\gamma\sqcup\gamma'$ is non-trivial, which does not follow directly from non-triviality of $\gamma^*\sqcup\gamma'^*$, because the former is not a ``subexpression'' of the latter.  Fortunately, because the latter is non-trivial and closed, we know $\gamma^*$ and $\gamma'^*$ are semantically defined, so we can conclude that semantically $\gamma\sqsubseteq\gamma^*$ and likewise for $\gamma'$ because the iteration operator is extensive.  Then because join's definition is downward-closed (Corollary \ref{coro:join_def_anti}), we know that $\gamma\sqcup\gamma'$ must also be semantically defined, allowing us to conclude the antecedent equivalence holds semantically, thus by monotonicity of $(-)^*$ the conclusion holds semantically.
        \item Case \textsc{Eq-Trans}: Here Lemma \ref{lem:kind_equiv} implies the intermediate effect is also well-kinded, and Lemma \ref{lem:kind_preserves_effvars} implies the intermediate effect is also closed, so two uses of the inductive hypothesis give that the $\chi$ and the intermediate element are equal in $Q(\Gamma;\Sigma)$, and likewise the intermediate element is equal to $\chi'$ in $Q(\Gamma;\Sigma)$, so both $\chi$ and $\chi'$ are equal.
    \end{itemize}
\end{proof}
\begin{corollary}[Closed Syntactic Subeffecting]
    \label{lem:closed_syneff_sub}
    For any two closed non-trivial syntactic effects $\chi$ and $\chi'$ where $\Gamma;\Sigma\vdash\chi :: \mathcal{E}$ and $\Gamma;\Sigma\vdash\chi'::\mathcal{E}$, if $\Gamma;\Sigma\vdash\chi\sqsubseteq\chi'$ then considering $\chi$ and $\chi'$ as elements of $Q(\Gamma;\Sigma)$, $\chi\sqsubseteq\chi'$.
\end{corollary}
\begin{proof}
    By inversion on the syntactic subeffecting judgment, $\Gamma;\Sigma\vdash\chi\sqcup\chi'\equiv\chi'$. By Lemma \ref{prop:closed_kinded_eff_Q} both $\chi$ and $\chi'$ are equivalent to (possibly different) elements of $Q(\Gamma;\Sigma)$.
    Because $\chi'$ is non-trivial, it is not syntactically equivalent to the syntactic join of two effects whose join is semantically undefined, thus $\chi\sqcup\chi'$ is also semantically defined.
    By Lemma \ref{lem:closed_syneff_equiv}, $\chi\sqcup\chi'$ and $\chi'$ are semantically equivalent, which means that semantically $\chi\sqsubseteq\chi'$.
\end{proof}

\subsection{Syntactic Safety}
\label{sec:syntactic_safety}
Syntactic type safety proceeds in the normal manner (for a language with mutually-defined types and terms), with only a few wrinkles due to effect quantales.
Here we give the major lemmas involved in the type safety proof, with outlines of the proofs themselves.

Systems with any kind of uncontrolled subsumption (e.g., arbitrary subtyping, or our rule for ascribing an expression an arbitrary larger effect) introduce some extra complexity in cases where inversion on a typing derivation is desired, because typing is no longer syntax-directed --- every na\"ive inversion yields two subgoals: one specific, and one corresponding to a use of subtyping.  This can be managed more cleanly and uniformly with a helpful lemma:

\begin{lemma}[Derivations Ending Without Subsumption]
\label{lem:without}
If $\Gamma;\Sigma\vdash e : \tau \mid \chi$, then there exists a derivation $\Gamma;\Sigma\vdash e : \tau \mid \chi'$ whose last inference is not due to subsumption, where $\Gamma;\Sigma\vdash\chi'\sqsubseteq\chi$.
\end{lemma}
\begin{proof}
    By induction on the derivation.  In all cases that specialize to a particular syntactic element, the result is immediate by applying the relevant typing rule again, and letting $\chi'=\chi$.  In the case of subsumption, we are left with a derivation of $\Gamma;\Sigma\vdash e : \tau \mid \gamma$ and $\Gamma;\Sigma\vdash\gamma\sqsubseteq\chi$.  We do not know if this new typing derivation ends with a use of subsumption.  But by the inductive hypothesis, for some $\chi'$, $\Gamma;\Sigma\vdash e : \tau \mid \chi'$ and $\Gamma;\Sigma\vdash\chi'\sqsubseteq\gamma$.  Since subeffecting is transitive, $\Gamma;\Sigma\vdash\chi'\sqsubseteq\chi$.
\end{proof}

Substitution lemmas are proven by induction on the expression's type derivation, exploiting the fact that all values' effects before subeffecting are $I$:
\begin{lemma}[Term Substitution]
\label{lem:term_subst}
If $\Gamma;\Sigma\vdash v : \tau \mid I$, then:
\begin{itemize}
\item if $\Gamma,x:\tau;\Sigma\vdash e : \tau' \mid \gamma$ then $\Gamma;\Sigma\vdash e[v/x] : \tau'[v/x] \mid \gamma[v/x]$
\item if $\Gamma,x:\tau;\Sigma\vdash \tau' :: \kappa$ then $\Gamma;\Sigma\vdash\tau'[v/x] :: \kappa$
\end{itemize}
\end{lemma}
\begin{proof}
By simultaneous induction on the typing and kinding relations.
The only subtle case is substitution of a variable occurring in a ground effect where $\Gamma,x:\tau;\Sigma\vdash E(\overline{v'}) :: \mathcal{E}$.
By inversion (via \textsc{K-Concrete}), the arguments $\overline{v'}$ are well-typed under $\Gamma,x:\tau;\Sigma$.
In this case, uses of $x$ in arguments to the effect constructor are being replaced by $v$, which we can do using the effect quantale homomorphism induced by $Q$ being indexed, which must necessarily yield $E(\overline{v'[v/x]})$ replacing uses of $x$ in the constructor arguments with $v$. Since by the term-substitution assumption the resulting arguments are well-typed under $\Gamma;\Sigma$, \textsc{K-Concrete} applies.
\end{proof}

\begin{lemma}[Type Substitution]
\label{lem:type_subst}
If $\Gamma;\Sigma\vdash \tau' :: \kappa$, then:
\begin{itemize}
\item if $\Gamma,\alpha::\kappa;\Sigma\vdash e : \tau \mid \gamma$, then $\Gamma;\Sigma\vdash e[\tau'/\alpha] : \tau[\tau'/\alpha] \mid \gamma[\tau'/\alpha]$
\item 
if $\Gamma,\alpha::\kappa;\Sigma\vdash \tau :: \kappa'$, then $\Gamma;\Sigma\vdash\tau[\tau/\alpha]::\kappa'$
\end{itemize}
\end{lemma}
\begin{proof}
By simultaneous induction on the typing and kinding relations, similar to Lemma \ref{lem:term_subst}.
\end{proof}
\begin{lemma}[Canonical Forms]
\label{lem:canonical}
If $\epsilon;\Sigma\vdash v : \tau \mid \gamma$ then:
\begin{itemize}
\item If $\tau=\Pi x:\tau'\overset{\gamma'}{\rightarrow}\tau''$, then $v$ is a primitive ($c$) or $v$ is of the form $(\lambda y\ldotp e)$ and $\epsilon;\Sigma\vdash I\sqsubseteq\gamma$.
\item If $\tau=\forall\alpha::\kappa\overset{\gamma'}{\rightarrow}\tau'$ then $v$ is a primitive or $v$ is of the form $(\Lambda\alpha::\kappa\ldotp e)$ and $\epsilon;\Sigma\vdash I\sqsubseteq\gamma$.
\item If $\tau=\mathsf{bool}$, $v=\mathsf{true}\vee v=\mathsf{false}$ and $\epsilon;\Sigma\vdash I\sqsubseteq\gamma$.
\item If $\tau=\mathsf{unit}$, $v=()$ and $\epsilon;\Sigma\vdash I\sqsubseteq\gamma$.
\end{itemize}
\end{lemma}
\begin{proof}
    Standard, with the exception that the boolean and unit cases rely on the restrictions imposed on the codomain of $\mathsf{StateEnv}$ to not give any primitives those closed base types.
\end{proof}
\begin{lemma}[Value Typing]
\label{lem:valtyping}
    If $\Gamma;\Sigma\vdash v : \tau \mid \gamma$, then $\Gamma;\Sigma\vdash v: \tau \mid I$ and $\Gamma;\Sigma\vdash I \sqsubseteq \gamma$.
\end{lemma}
\begin{proof}
    By Lemma \ref{lem:without}, $\Gamma;\Sigma\vdash v : \tau \mid \chi$ and $\Gamma;\Sigma\vdash\chi \sqsubseteq \gamma$.
    Proceed by inversion on the new typing derivation that does not end with subsumption).  For core values (functions, booleans, unit) this is direct from the type rule.  For primitives, Section \ref{sec:params} required $\Sigma\in\mathsf{StateEnv}$ assign effect $I$ to values.
\end{proof}

\begin{lemma}[Primitive Application Typing]
    \label{lem:primappty}
    For any well-typed primitive application $p\;\overline{v}$ of at most $\mathsf{Arity}(p)$ depth, if $\Gamma;\Sigma\vdash p\;\overline{v} : \tau \mid \chi$, then 
    \begin{itemize}
        \item $\tau=\mathsf{LastResult}(\Sigma(p))[\overline{v}/\mathsf{args}(p)]$
        \item $\Gamma;\Sigma\vdash \mathsf{LastEffect}(\Sigma(p))[\overline{v}/\mathsf{args}(p)] \sqsubseteq \chi$
        \item $\Gamma;\Sigma\vdash p\;\overline{v} : \tau \mid \mathsf{LastEffect}(\Sigma(p))[\overline{v}/\mathsf{args}(p)]$
    \end{itemize}
    where $-[\overline{v}/\mathsf{args}(p)]$ is the iterated substitution of argument values and types.
\end{lemma}
\begin{proof}
    By induction on $\mathsf{Arity}(p)$.
    In the case $\mathsf{Arity}(p)=0$, $p$ is a constant, there are no arguments, and the result follows from Lemma \ref{lem:without}.
    In the case $\mathsf{Arity}(p) > 0$: Let $n$ be the nesting depth of $p\;\overline{v}$, bounded by $\mathsf{Arity}(p)$.  Then this follows from proving a strengthened inductive hypothesis for different numbers of applications.
    To state this we need mild generalizations of $\mathsf{LastEffect}$ and $\mathsf{LastResult}$ (Section \ref{sec:params}).
    Let $\mathsf{ResultType}(i,\tau)$ return the codomain type for some (functional) type $\tau=B_0(x_0)\xrightarrow{\chi_0}\ldots B_n(x_n)\xrightarrow{\chi_n}\tau'$ after supplying $n$ arguments --- i.e., 
    \[\mathsf{ResultType}(i,B_0(x_0)\xrightarrow{\chi_0}\ldots B_n(x_n)\xrightarrow{\chi_n}\tau')=B_i(x_i)\xrightarrow{\chi_{i}}\ldots B_n(x_n)\xrightarrow{\chi_n}\tau'\]
    Intuitively, $\mathsf{ResultType}(i,\tau)$ peels off $n$ binders and latent effects from the functional effect $\tau$, returning the result. It is undefined for non-functional types, or when $i$ exceeds the number of curried arguments.
    Then let $\mathsf{Effect}(i,\tau)$ select the $i$-th latent effect $\chi_i$ of such types.
    To see how this generalizes $\mathsf{LastResult}$ and $\mathsf{LastEffect}$, note that
    $\mathsf{LastEffect}(\Sigma(p))=\mathsf{Effect}(\mathsf{Arity}(p),\Sigma(p))$,
    $\mathsf{LastResult}(\Sigma(p))=\mathsf{ResultType}(\mathsf{Arity}(p),\Sigma(p))$.
    Note also that $\mathsf{LastResult}(0,\tau)=\tau$.

    We can now state a strengthened inductive hypothesis
    \[ \begin{array}{l}
        \forall n\le\mathsf{Arity}(p)\ldotp \Gamma;\Sigma\vdash p\;\overline{v}^n : \tau \mid \chi \Rightarrow\\
        \qquad \tau=\mathsf{ResultType}(n,\Sigma(p))[\overline{v}^n/\mathsf{args}(n,\Sigma(p))] \land \Gamma;\Sigma\vdash\mathsf{Effect}(n,\Sigma(p))[\overline{v}^n/\mathsf{args}(n,\Sigma(p))] \sqsubseteq \chi 
        \\\qquad\land\Gamma;\Sigma\vdash p\;\overline{v}^n : \tau \mid \mathsf{Effect}(n,\Sigma(p))[\overline{v}^n/\mathsf{args}(n,\Sigma(p))]
    \end{array}\] 
    and prove the lemma by specializing to the case where $n=\mathsf{Arity}(p)$.
    Note that since we required $\Sigma$ to assign latent effect $I$ for all non-full applications of primitives, the use of $\mathsf{Effect}(-,-)$ in the lemma above returns $I$ for all $i<\mathsf{Arity}(p)$. Because $n$ is bounded by the arity of $p$, the constraints on $\Sigma$ from Section \ref{sec:params} ensure all uses of $\mathsf{Effect}$ and $\mathsf{ResultType}$ are defined.

    This generalization itself is proven by induction on $n$.  When $n=0$ the result follows as above.  When $n>0$, $p\;\overline{v}^n=(p\;\overline{v}^{n-1})\;v'$ or $p\;\overline{v}^n=(p\;\overline{v}^{n-1})[\tau']$.

    In the case of the value application, Lemma \ref{lem:without} followed by inversion on typing and use of the inductive hypothesis gives a (dependent) function type for $p\;\overline{v}^{n-1}$, whose result types is 
    $\mathsf{ResultType}(n-1,\Sigma(p))[\overline{v}^{n-1}/\mathsf{args}(n-1,\Sigma(p))]$, and whose effect $\chi'$ is lower-bounded by $\mathsf{Effect}(n-1,\Sigma(p))[\overline{v}^{n-1}/\mathsf{args}(n-1,\Sigma(p))]$.  Since by assumption all latent effects of primitive types except the last are $I$, $\Gamma;\Sigma\vdash I\sqsubseteq\chi'$.
    It also gives that the effect $\chi$ is an upper bound on \[\chi'\rhd\mathsf{Effect}(n,\Sigma(p))[\overline{v}^{n-1}/\mathsf{args}(n-1,\Sigma(p))][v'/\mathsf{arg}(n,\Sigma(p))]\]
    which is itself an upper bound on 
    \[I\rhd\mathsf{Effect}(n,\Sigma(p))[\overline{v}^{n-1}/\mathsf{args}(n-1,\Sigma(p))][v'/\mathsf{arg}(n,\Sigma(p))]\]
    which is exactly \[\mathsf{Effect}(n,\Sigma(p))[\overline{v}^n/\mathsf{args}(n,\Sigma(p))]\]  Similarly, the result type of this application will be 
    \[\begin{array}{l}
        \mathsf{ResultType}(n,\Sigma(p))[\overline{v}^{n-1}/\mathsf{args}(n-1,\Sigma(p))][v'/\mathsf{arg}(n,\Sigma(p))] \\
        \qquad = \mathsf{ResultType}(n,\Sigma(p))[\overline{v}^n/\mathsf{args}(n,\Sigma(p))]
    \end{array}\]
    The new typing derivation is obtained from the inductive hypothesis and the fact that all effects but the last are $I$ for a primitive, and thus using the stated $n$th latent effect as the effect of the application.

    The case of type application is similar, but with a type lambda and substitution for a type-level variable instead of a value variable.
\end{proof}

We give type preservation below.  
\begin{lemma}[One Step Type Preservation]
\label{lem:onestep_preservation}
For all $Q$, $\sigma$, $e$, $e'$, $\Sigma$, $\tau$, $\gamma$, and $\gamma'$, if
\[\epsilon;\Sigma\vdash e : \tau \mid \gamma
\qquad\textrm{and}\qquad
\vdash \sigma : \Sigma
\qquad\textrm{and}\qquad
\sigma,e\rightarrow^{\gamma'}_Q \sigma',e'
\qquad\textrm{and}\qquad
\mathsf{NonTrivial}(\gamma)
\]
then there exist $\Sigma'\in\mathsf{StateEnv}$, $\gamma''\in Q(\epsilon;\Sigma)$ such that
\[\epsilon;\Sigma'\vdash e' : \tau \mid \gamma''
\quad\textrm{and}\quad
\vdash\sigma':\Sigma'
\quad\textrm{and}\quad
\Sigma\le\Sigma'
\quad\textrm{and}\quad
\gamma'\rhd_Q\gamma''\sqsubseteq_Q\gamma
\quad\textrm{and}\quad
\mathsf{NonTrivial}(\gamma'')
\]
\end{lemma}
Before giving the proof, we point out a couple small subtleties of the statement related to the distinction between syntactic and semantic effects.  First, note that because this lemma applies to the empty type environment, and all effects concluding a type judgment with the empty type environment are necessarily closed, Proposition \ref{prop:closed_eff_concl_Q} applies, implying that $\gamma$ is equivalent to an element of $Q(\epsilon;\Sigma)$, and $\gamma''$ is equivalent to an element of $Q(\epsilon;\Sigma')$.  Because the operational semantics only emit elements of $Q(\mathsf{Values})$ that are $I$ or an element restricted by primitive preservation to be valid for the environments of interest, we can assume a similar restraint there.
Thus in the penultimate conjunct of the conclusion --- $\gamma'\rhd_Q\gamma''\sqsubseteq_Q\gamma$ we abuse notation and take this to mean the relationship holds with $\gamma'$, $\gamma''$, and $\gamma$ considered as elements of $Q(\epsilon;\Sigma')$ (using the inclusion $Q(\epsilon;\Sigma)\hookrightarrow Q(\epsilon;\Sigma')$ due to $Q$ being monotone) without explicitly employing Proposition \ref{prop:closed_eff_concl_Q} (that is, $\gamma$ and $\gamma''$ are still syntactic effects from $S_{Q(\mathsf{Values})}$).
This requires the proof to ensure the claimed composition is defined in the effect quantale.
\begin{proof}
By induction on the reduction relation.
We present only the non-trivial reducts (i.e., we omit structural rules that simply reduce a subterm and propagate state and effect labels).
\begin{itemize}
\item Case \textsc{E-App}: Here we know $e=(\lambda x\ldotp e_b)\;v$, $e'=e_b[v/x]$, $\sigma=\sigma'$ and $\gamma'=I$.  
By Lemma \ref{lem:without}, there is a derivation not ending with subsumption for some effect $\chi$ where $\epsilon;\Sigma\vdash \chi\sqsubseteq\gamma$.
By inversion on the new typing derivation and consequences of the effect ordering:
\[\begin{array}{l@{\qquad}l}
\epsilon;\Sigma\vdash(\lambda x\ldotp e_b):\Pi x:\tau_{arg}\overset{\gamma_f}{\rightarrow}\tau_{res} \mid \gamma_a
&
\tau=\tau_{res}[v/x]
\\
\epsilon;\Sigma\vdash v : \tau_{arg}\mid\gamma_v
&
x\not\in\mathsf{FV}(\gamma_f,\tau_{res})\vee\mathsf{Value}(v)
\\
\gamma_a\rhd\gamma_v\rhd(\gamma_f[v/x])=\chi
&
\mathsf{NonTrivial}(\gamma_a\rhd\gamma_v\rhd(\gamma_f[v/x]))
\end{array}
\]
By value typing (Lemma \ref{lem:valtyping}), the function and argument are each typeable with effect $I$, and $I\sqsubseteq\gamma_a$ and $I\sqsubseteq\gamma_v$.
So $\epsilon;\Sigma\vdash (\gamma_f[v/x])\equiv I\rhd I\rhd(\gamma_f[v/x])\sqsubseteq\gamma_a\rhd\gamma_v\rhd(\gamma_f[v/x])=\chi$.
From the inversion on the function's type derivation, $\epsilon,x:\tau_{arg};\Sigma\vdash e_b : \tau_{res} \mid \gamma_f$.
By term substitution (Lemma \ref{lem:term_subst}) we then have $\epsilon;\Sigma\vdash e_b[v/x] : \tau_{res}[v/x] \mid \gamma_f[v/x]$.
We then know $\epsilon;\Sigma\vdash I\rhd\gamma_f[v/x]\sqsubseteq\gamma_f[v/x]\sqsubseteq\chi\sqsubseteq\gamma$, and as the state did not change, let $\Sigma'=\Sigma$ and the state remains well-typed.
By the consequences above, $\mathsf{NonTrivial}(\gamma_f[v/x])$.
It remains to show that $\epsilon;\Sigma\vdash I\rhd\gamma_f[v/x]\sqsubseteq\gamma$, which follows from transitivity with $\chi$ as a mid-point.
By Proposition \ref{prop:closed_eff_concl_Q} and the underlying effect quantale $I\rhd_Q\gamma_f[v/x]\sqsubseteq_Q\gamma\in Q(\epsilon;\Sigma')$.
\item Case \textsc{E-TApp}: Here $e=(\Lambda\alpha::\kappa\ldotp e_b)[\tau_\alpha]$, $\gamma'=I$, $e'=e_b[\tau_\alpha/\alpha]$, $\sigma=\sigma'$.
By Lemma \ref{lem:without} and inversion on the non-subsuming typing assumption:
\[\begin{array}{l@{\qquad}l}
\epsilon;\Sigma\vdash(\Lambda\alpha::\kappa\ldotp e_b) : \forall\alpha::\kappa\overset{\gamma_f}{\rightarrow}\tau_{res}\mid\gamma_e
&
\tau=\tau_{res}[\tau_\alpha/\alpha]
\\
\epsilon;\Sigma\vdash \tau_\alpha :: \kappa
&
\chi=\gamma_e\rhd\gamma_f[\tau_\alpha/\alpha]\\
\mathsf{NonTrivial}(\gamma_e\rhd\gamma_f[\tau_\alpha/\alpha]) & 
\epsilon;\Sigma\vdash\chi\sqsubseteq\gamma\\
\gamma_e\rhd\gamma_f[\tau_\alpha/\alpha]=\chi
\end{array}\]
By value typing, $\epsilon;\Sigma\vdash I\sqsubseteq\gamma_e$, so 
\[\epsilon;\Sigma\vdash  \gamma_f[\tau_\alpha/\alpha]\equiv I\rhd\gamma_f[\tau_\alpha/\alpha]\sqsubseteq\gamma_e\rhd\gamma_f[\tau_\alpha/\alpha] =\chi\sqsubseteq\gamma
\]
Also from the inversion on the function's type derivation, $\epsilon,\alpha::\kappa\mathbin{;}\Sigma\vdash e_b : \tau_{res} \mid \gamma_f$.
By type substitution (Lemma \ref{lem:type_subst}) we then have $\epsilon;\Sigma\vdash e_b[\tau_\alpha/\alpha] : \tau_{res}[\tau_\alpha/\alpha] \mid \gamma_f[\tau_\alpha/\alpha]$, establishing $\gamma''=\gamma_f[\tau_\alpha/\alpha]$.  We also know $\epsilon;\Sigma\vdash I\rhd\gamma_f[\tau_\alpha/\alpha]\equiv\gamma_f[\tau_\alpha/\alpha]\sqsubseteq\chi\sqsubseteq\gamma$, and as the state did not change, $\Sigma'=\Sigma$ and the state remains well-typed.
$\gamma''$ is nontrivial because it is less than the assumed-nontrivial $\gamma$.
By Proposition \ref{prop:closed_eff_concl_Q}, $I\rhd_Q\gamma_f[\tau_\alpha/\alpha]\sqsubseteq_Q\gamma\in Q(\epsilon;\Sigma')$.
\item Case \textsc{E-Prim}: Here $e=p\;\overline{v}$, and $\llbracket p\;\overline{v}\rrbracket(\sigma)=(e',\gamma',\sigma')$.
Because $\llbracket p\;\overline{v}\rrbracket$ is defined, the initial term is a series of nested applications and type applications rooted at primitive $p$, and this is a full application of $p$.
So by Lemma \ref{lem:primappty}, $\tau=\mathsf{LastResult}(\Sigma(p))[\overline{v}/\mathsf{args}(\Sigma(p))]$, and $\epsilon;\Sigma\vdash\mathsf{LastEffect}(\Sigma(p))[\overline{v}/\mathsf{args}(\Sigma(p))]\sqsubseteq\gamma$, and moreover:
\[
    \epsilon;\Sigma\vdash p\;\overline{v} : \tau \mid \mathsf{LastEffect}(\Sigma(p))[\overline{v}/\mathsf{args}(\Sigma(p))]
\]
Primitive preservation then gives that $\gamma'\sqsubseteq \mathsf{LastEffect}(\Sigma(p))[\overline{v}/\mathsf{args}(\Sigma(p))]$, and that $\epsilon;\Sigma'\vdash v' : \tau \mid I$
for some $\Sigma'\ge\Sigma$ such that $\vdash\sigma':\Sigma'$.
Since $\gamma'$ is less than the assumed-non-trivial $\gamma$, it is also non-trivial.
Proposition \ref{prop:closed_eff_concl_Q} ensures the effects from the type judgments can be viewed as elements of $Q(\epsilon;\Sigma')$, and since the dynamic effects generated by $\llbracket-\rrbracket$ are restricted to depend only on values passed at the call site, and since those values were well typed under $\epsilon;\Sigma$, the dynamic effect can be as well.  Since the dynamic effect is less than the static effect (again, by primitive preservation) and both are concrete elements, then since $\gamma''=I$, the requirement that $\gamma'\rhd_Q\gamma''\sqsubseteq_Q\gamma$ is satisfied.
\item Case \textsc{E-IfTrue}: Here $e=\mathsf{if}\;\mathsf{true}\;e_1\;e_2$, $e'=e_1$, $\sigma=\sigma'$, $\gamma'=I$.  By Lemma \ref{lem:without} and inversion on typing:
\[\begin{array}{l@{\qquad}l}
\epsilon;\Sigma\vdash\mathsf{true}:\mathsf{bool}\mid\gamma_{\mathsf{true}}
&
\epsilon;\Sigma\vdash e_1 : \tau \mid \gamma_1
\\
\epsilon;\Sigma\vdash e_2 : \tau \mid \gamma_2
&
\chi=\gamma_{\mathsf{true}}\rhd(\gamma_1\sqcup\gamma_2)
\\
\mathsf{NonTrivial}(\gamma)
&
\mathsf{NonTrivial}(\gamma_{\mathsf{true}}\rhd(\gamma_1\sqcup\gamma_2))
\\
\epsilon;\Sigma\vdash \chi\sqsubseteq\gamma
\end{array}\]
By value typing, $\epsilon;\Sigma\vdash I\sqsubseteq\gamma_\mathsf{true}$.
By local hypothesis, immediately $\epsilon;\Sigma\vdash e_1 : \tau \mid\gamma_1$, and by effect quantale laws $\epsilon;\Sigma\vdash I\rhd\gamma_1\sqsubseteq\gamma_{\mathsf{true}}\rhd(\gamma_1\sqcup\gamma_2)=\chi\sqsubseteq\gamma$, which holds semantically (i.e., in $Q(\epsilon;\Sigma)$) as well because the effects involved are non-trivial closed effects.
State is unchanged, and remains well-typed under $\Sigma'=\Sigma$.
Thus the results hold for $\gamma''=\gamma_1$, which is non-trivial as a subexpression of a non-trivial effect.
\item Case \textsc{E-IfFalse}: Analogous to \textsc{E-IfTrue}.
\item Case \textsc{E-While}: Here $e=\mathsf{while}\;e_c\;e_b$, $\gamma'=I$, $\sigma=\sigma'$, and $e'=\mathsf{if}\;e_c\;(e_b;(\mathsf{while}\;e_c\;e_b))\;()$.
By Lemma \ref{lem:without} and inversion on typing:
\[\begin{array}{l@{\qquad}l@{\qquad}l@{\qquad}l}
\epsilon;\Sigma\vdash e_c:\mathsf{bool} \mid \gamma_c
&
\epsilon;\Sigma\vdash e_b:\tau_b \mid \gamma_b
&
\chi=\gamma_c\rhd(\gamma_b\rhd\gamma_c)^*
\\
\tau=\mathsf{unit}
&
\mathsf{NonTrivial}(\gamma)
&
\mathsf{NonTrivial}(\gamma_c\rhd(\gamma_b\rhd\gamma_c)^*)
\\
\epsilon;\Sigma\vdash \chi\sqsubseteq\gamma
\end{array}\]
The \textsf{NonTrivial} assumptions imply that all subexpressions of $\chi$ are also non-trivial.
By \textsc{T-If}, \textsc{T-Unit}, desugaring $;$ to function application, and weakening:
\[\epsilon;\Sigma\vdash\mathsf{if}\;e_c\;(e_b;(\mathsf{while}\;e_c\;e_b))\;() : \mathsf{unit} \mid \gamma_c\rhd(((\gamma_b\rhd\gamma_c)\rhd(\gamma_b\rhd\gamma_c)^*)\sqcup I) \]
State remains unchanged, so the final obligation in this case is to prove the effect just given for $e'$ (technically, preceded by $I\rhd$) is semantically a subeffect of $\chi=\gamma_c\rhd(\gamma_b\rhd\gamma_c)^*$ (in turn a subeffect of $\gamma$), which relies crucially on iteration being foldable and possibly-empty:
\[\begin{array}{rcl}
\epsilon;\Sigma\vdash \gamma_c\rhd((\gamma_b\rhd\gamma_c\rhd(\gamma_b\rhd\gamma_c)^*)\sqcup I)
&\sqsubseteq&
\gamma_c\rhd(((\gamma_b\rhd\gamma_c)^*\rhd(\gamma_b\rhd\gamma_c)^*)\sqcup I)\\
&\sqsubseteq&
\gamma_c\rhd(((\gamma_b\rhd\gamma_c)^*)\sqcup I)\\
&\sqsubseteq&
\gamma_c\rhd((\gamma_b\rhd\gamma_c)^*)\\
&\sqsubseteq&
\gamma
\end{array}
\]
\end{itemize}
\end{proof}

\begin{theorem}[Type Preservation]
    \label{thm:preservation}
For all $Q$, $\sigma$, $e$, $e'$, $\Sigma$, $\tau$, $\gamma$, and $\gamma'$, if
\[\epsilon;\Sigma\vdash e : \tau \mid \gamma
\qquad\textrm{and}\qquad
\vdash \sigma : \Sigma
\qquad\textrm{and}\qquad
\sigma,e\dhxrightarrow{\gamma'}_Q \sigma',e'
\qquad\textrm{and}\qquad
\mathsf{NonTrivial}(\gamma)
\]
then there exist $\Sigma'$, $\gamma''$ such that
\[\epsilon;\Sigma'\vdash e' : \tau \mid \gamma''
\quad\textrm{and}\quad
\vdash\sigma':\Sigma'
\quad\textrm{and}\quad
\Sigma\le\Sigma'
\quad\textrm{and}\quad
\gamma'\rhd_Q\gamma''\sqsubseteq_Q\gamma
\quad\textrm{and}\quad
\mathsf{NonTrivial}(\gamma'')
\]
\end{theorem}
\begin{proof}
By straightforward induction on the transitive reduction relation, applying Lemma \ref{lem:onestep_preservation} in the inductive case.
\end{proof}

\begin{theorem}[One Step Progress]
    \label{thm:one_progress}
For all $\Sigma$, $e$, $\tau$, $\gamma$, $\sigma$, if $\epsilon;\Sigma\vdash e : \tau \mid \gamma$ and $\vdash\sigma:\Sigma$, then either $e$ is a value, $e$ is an incomplete primitive application that is not defined, $e$ is a complete application of a blocking primitive $p\in B$, or there exists some $e'$, $\sigma'$, and $\gamma'$ such that $\sigma,e\rightarrow^{\gamma'}_Q \sigma',e'$.
\end{theorem}
\begin{proof}
By induction on the typing derivation.
\begin{itemize}
\item Case \textsc{T-Var} holds vacuously because the empty type environment contains no variable types.
\item Cases \textsc{T-Lambda}, \textsc{T-Bool}, \textsc{T-TyLambda}, \textsc{T-Unit}: These are all immediately values.
\item Case \textsc{T-Subsume}: By inductive hypothesis.
\item Case \textsc{T-Prim}: Either the expression is a primitive value ($c$), or it is a primitive application. Either the application is incomplete, or for complete applications the primitive progress language parameter ensures the application should reduce or is a stuck application of a primitive in $B$.
\item Case \textsc{T-App}: First apply the inductive hypothesis for $e_1$'s typing derivation.
  $e_1$ is either a value, a partial primitive application, or can reduce.  If it can reduce, simply use the first structural rule for reducing the left side of an application.

  Otherwise no further reductions are possible on $e_1$, so apply the inductive hypothesis to $e_2$, which is either a value, an incomplete primitive application with no further arguments (which is impossible, by the assumption that all primitives appear fully applied), or can reduce.  If $e_2$ can reduce, then this application reduces by either the second structural rule for applications or by \textsc{E-PrimArg}.

  This leaves the cases where $e_2$ is a value, and $e_1$ is a partial primitive application or a function.  If $e_1$ is a value, by Canonical Forms (Lemma \ref{lem:canonical}) $e_1=(\lambda y\ldotp e_b)$ for some variable and function body, so the application steps to $e_b[e_2/y]$ with the unit effect by \textsc{E-App}.
  If $e_1=p\;\overline{v}$, the next step depends on whether $\llbracket p\;\overline{v}\;e_2\rrbracket$ (where $e_2$ is a value) is defined. If so, this expression reduces by \textsc{E-Prim}.  Otherwise this is a partial primitive application (which by assumption must occur in a context where more arguments are supplied).
\item Case \textsc{T-If}: Similar to \textsc{T-App} in reducing the condition.  When the condition is a value, Lemma \ref{lem:canonical} produces the result that the condition is either \textsc{true} or \textsc{false}, and either \textsc{E-IfTrue} or \textsc{E-IfFalse} applies.
\item Case \textsc{T-While}: While loops are macro-expanded to conditionals by \textsc{E-While} when they are in reduction position.
\item Case \textsc{T-TyApp}: Similar to $\textsc{T-App}$, but using the type application structural rule, or when no reductions are possible on the nested expression, applying Canonical forms and stepping by \textsc{E-TyApp} when it is a value, or concluding this is a partial primitive application.

\end{itemize}
\end{proof}

\begin{theorem}[Type Safety]
    \label{thm:soundness}
For all $Q$, $\sigma$, $e$, $e'$, $\Sigma$, $\tau$, $\gamma$, and $\gamma'$, if
\[\epsilon;\Sigma\vdash e : \tau \mid \gamma
\qquad\textrm{and}\qquad
\vdash \sigma : \Sigma
\qquad\textrm{and}\qquad
\sigma,e\dhxrightarrow{\gamma'}_Q \sigma',e'\]
then either (a) $e'$ is a value of type $\tau$ (under some $\Sigma'\ge\Sigma$ such that $\vdash\sigma':\Sigma'$) and $\gamma'\sqsubseteq\gamma$, or (b) there exists $\sigma'',\gamma'',e''$ such that $\sigma',e'\longrightarrow^{\gamma''}_Q \sigma'',e''$.
\end{theorem}

Note also that with Lemma \ref{lem:subsume} relating the join semilattice model of traditional commutative effect systems to effect quantales, Theorem \ref{thm:soundness} recovers the traditional syntactic type safety for those systems as well.

\subsection{From Syntactic to Semantic Soundness}
\label{sec:interp}
As discussed earlier, the type of syntactic result given in Theorem \ref{thm:soundness} ensures a useful consistency criterion relating effects and reduction, but when effects are intended to imply some semantic properties of execution, this form only proves part of what is desired. This style of type safety result is useful for trace-based systems like history effects, where the semantic property is actually a constraint on the reduction sequence, but less useful for effect systems where effects concern state change or other properties.

We can strengthen the previous result by taking inspiration from Katsumata's framing~\cite{katsumata14} of the relationship between preordered monoids and graded monads: an element of the effect algebra specifies a subset of program semantics. Katsumata phrases the intuition behind this nicely, as an effect acting as a refinement type on general semantics.  Sticking to safety properties, we give a relational interpretation to each effect.  Intuitively, each effect corresponds to a set of possible state modifications --- expressed as binary relations on states --- and the effect quantale operations must agree in a natural way with combinations of these relations.  Formally:

\begin{definition}[Effect Quantale Interpretation]
    We define an \emph{effect quantale interpretation} for an effect quantale $Q$ over a set $\mathsf{State}$ to be a function $\mathcal{I}: Q \rightarrow \mathcal{P}(\mathsf{State}\times\mathsf{State})$ with the following properties:
    \begin{itemize}
        \item $\forall x,y\in Q\ldotp \mathcal{I}({x\sqcup y})=\mathcal{I}({x})\cup\mathcal{I}({y})$
        \item $\forall x,y\in Q\ldotp \mathcal{I}({x\rhd y})=\{(a,c) \mid \exists b\ldotp (a,b)\in\mathcal{I}({x})\land (b,c)\in\mathcal{I}({y})\}$
        \item $\mathcal{I}(I) = \{ (a,a) \mid a \in \mathsf{State} \}$
    \end{itemize}
    A homomorphism between effect quantale interpretations $\mathcal{I}$ for $Q$ over $\mathsf{State}$ and $\mathcal{I'}$ for $Q'$ over $\mathsf{State}'$, is a pair of an effect quantale homomorphism $f$ and a function $s:\mathsf{State}\rightarrow\mathsf{State}'$
    such that $\langle s, s\rangle\circ\mathcal{I} = \mathcal{I'}\circ f$.
\end{definition}
From these we can also prove the desirable property $\forall x,y\in Q\ldotp x \sqsubseteq y \Rightarrow \mathcal{I}({x}) \subseteq \mathcal{I}({y})$. That the interpretation respects the relevant associativity, commutativity, and distributivity laws follows from the fact that binary relations form an effect quantale:
$(\mathsf{Rel}(\mathsf{State}), \cup, \circ, =_{\mathsf{Rel}(\mathsf{State})})$ is itself an effect quantale, so an interpretation is \emph{almost} an effect quantale morphism into this relational effect quantale.  The reason it cannot be for our purposes is that in our proof framework (Section \ref{sec:params}) states are defined separately from effects, not mutually, so we cannot use this directly in the proof framework.  But a more sophisticated framework making use of guarded recursive types perhaps could.
For arbitrary unrelated effect quantales, homomorphisms between their interpretations are not necessarily interesting.

Sequential effects frequently have very straightforward relational interpretations in terms of runtime states.

\begin{example}[History Effect Interpretation]
    \label{ex:history_effect_interp}
    \citet{skalka2008types} prove soundness of their system by using state that accumulates a trace of events. Their notion of soundness is then that the runtime trace is included in the static trace set.  
Recalling their denotation of history effects in terms of sets of finite traces $\llbracket-\rrbracket:H\rightarrow \mathcal{P}(C^*)$ for events drawn from $C$, we can interpret history effects as relations on the accumulating state by specifying that a history effect corresponds to appending \emph{some} finite trace in its denotation to the dynamic trace:
\[ \mathcal{I}(H)=\{(\overline{x},\overline{y}) \mid \exists w\in\llbracket{H}\rrbracket\ldotp \overline{y}=\overline{x}\mathrel{++}w \} \]
\end{example}
Example \ref{ex:history_effect_interp} suggests a natural interpretation for finite trace effects in such a system as well.
Other effect systems have similarly straightforward interpretations, such as $\mathcal{L}(A)$ interpreting pre- and post-multi-sets of locks held as the set of locks held in the pre- and post-states (including counts stored in the state for recursive acquisitions):

\begin{example}[Lock-Multiset Effects]
    For a given index set $A$, there is a natural interpretation of effects from $\mathcal{L}(A)$ as relating the count of claims held on each lock before and after a computation. Assume states are a pair of thread ID and partial map from lock names to ownership information ($\mathsf{LockNames}\rightharpoonup((\mathsf{ThreadID}\times\mathbb{N})~\mathsf{option})$). This is a projection of the actual state used for a single thread's reduction relation in a language with reentrant locks, which typically also includes a mutable heap that is irrelevant to interpreting lock ownership claims. In this case, taking $|a_\ell|$ as the multiplicity of $\ell$ in the multiset $a$, and we can interpret locking effects by:
    \[
        \mathcal{I}(a,a') = \left\{((i,m),(i,m')) \middle| 
            \begin{array}{rl}
                \forall \ell\in A\ldotp & (|a_\ell|>0 \Rightarrow \exists c\ldotp m(\ell)=(i,c)\land c \ge |a_\ell|) \\
                                        & \land(|a'_\ell|>0 \Rightarrow \exists c\ldotp m'(\ell)=(i,c)\land c \ge |a'_\ell|) \\
                &\land (|a_\ell|-|a'_\ell|>0 \Rightarrow \mathsf{ReleasedNTimes}(i,m,m',|a_\ell|-|a'_\ell|)) \\
                &\land (|a'_\ell|-|a_\ell|>0 \Rightarrow \mathsf{AcquiredNTimes}(i,m,m',|a'_\ell|-|a_\ell|)) \\
                & \land (|a_\ell|=|a'_\ell|\land|a_\ell|>0 \Rightarrow m(\ell)=m'(\ell))
            \end{array}\right\}
    \]
    where
    \[
        \mathsf{ReleasedNTimes}(i,m,m',n) = \begin{array}{l}
            (\exists c,c'\ldotp m(\ell)=(i,c)\land m'(\ell)=(i,c')\land c-c'=n)\\
            \vee(m(\ell)=(i,n)\land (m'(\ell)=\mathsf{None}\vee(\exists c,j\ldotp j\neq i\land m'(\ell)=(j,c))))
        \end{array}
    \]
    \[
        \mathsf{AcquiredNTimes}(i,m,m',n) = \begin{array}{l}
            (\exists c,c'\ldotp m(\ell)=(i,c)\land m'(\ell)=(i,c')\land c'-c=n)\\
            \vee((m(\ell)=\mathsf{None}\vee(\exists c,j\ldotp j\neq i\land m(\ell)=(j,c)))\land m'(\ell)=(i,n) )
        \end{array}
    \]
    That is, for any lock in the index set, the multiplicity of that lock's occurrence in the pre-state (resp.\ post-state) of the effect is a lower bound on the number of actual claims held by the current thread, and the change in multiplicities in the multisets matches the change in claims in the semantics (critically also requiring that if the number of static claims was non-zero and unchanged, then the number of semantic claims remains unchanged).
\end{example}

Of course both of these have used particular instantiations of indexed effect quantales, so we must extend our interpretations to specify their interactions with functions between index sets:
\begin{definition}[Indexed Effect Quantale Interpretation]
    We define an \emph{indexed effect quantale interpretation} for an indexed effect quantale $Q$ over a set $\mathsf{State}$ as a mapping that functorially assigns to each set $S$
    an effect quantale interpretation $\mathcal{I}(S)$ of $Q(S)$ over $\mathsf{State}$, and to each function $g:S\rightarrow T$ an effect quantale interpretation homomorphism 
    $\mathcal{I}(g) : \mathcal{I}(S)\rightarrow\mathcal{I}(T)$.

    We call an indexed effect quantale interpretation \emph{well-behaved} if it assigns to $g:S\rightarrow T$ the effect quantale interpretation homomorphism consisting of\footnote{Recall that Definition \ref{def:indexed_eq} gives indexed effect quantales actions on not only sets, but functions between sets.} $Q(g):Q(S)\rightarrow Q(T)$ and $\langle\mathsf{id},\mathsf{id}\rangle$.
\end{definition}
Note that a \emph{well-behaved} indexed effect quantale interpretation for a \emph{monotone} indexed effect quantale assigns an inclusion of interpretations to each inclusion $g:S\hookrightarrow T$: $\langle\mathsf{id},\mathsf{id}\rangle\circ \mathcal{I}(S) = \mathcal{I}(T)\circ Q(g)$, where $Q(g)$ is an inclusion (because $Q$ is monotone).
We will use well-behaved indexed interpretations of monotone effect quantales in the proof to transfer interpretations across state types: whenever $\Sigma\le\Sigma'$, then for the inclusion $f$ between well-typed values under $\epsilon;\Sigma$ and well-typed values under $\epsilon;\Sigma'$,
\[ 
    \mathcal{I}(f) : \mathcal{I}(\epsilon;\Sigma) \rightarrow \mathcal{I}(\epsilon;\Sigma')
\]
and because $Q$ is monotone, this implies that 
\[
    \mathcal{I}(\epsilon;\Sigma)(\gamma)(\sigma,\sigma') \Rightarrow \mathcal{I}(\epsilon;\Sigma')(\gamma)(\sigma,\sigma')
\]

Defining an interpretation expresses the intended meaning of effects, but to prove that the effect quantale leads the effect system to enforce the intended semantics, we must also relate this interpretation to the primitives of the system.

\begin{definition}[Interpretation-Consistent Primitives]
    We say a language parameter set as in Section \ref{sec:params} is \emph{interpretation-consistent} with an indexed interpretation $\mathcal{I}$ if 
    \[ \epsilon;\Sigma\vdash p~\overline{v} : \tau \mid \gamma \land \llbracket{p~\overline{v}}\rrbracket(\sigma)=(v',\gamma',\sigma')\Rightarrow \mathcal{I}(\epsilon;\Sigma)(\gamma')(\sigma,\sigma')\]
\end{definition}
Because earlier constraints on language parameters also ensure $\gamma'\sqsubseteq\gamma$, this along with the requirements on the interpretation itself ensure $\mathcal{I}(\epsilon;\Sigma)(\gamma)(\sigma,\sigma')$.

We actually do \emph{not} need to describe the interaction of interpretation with substitution.  In extending our soundness proof to additionally prove correct interpretation of effects, we will only interpret \emph{grounded} effects --- those which are closed and non-trivial --- because we only interpret dynamically-executed effects, which do not contain variables.  We do however require indexed interpretations to be well-behaved:

\begin{lemma}[Interpreted One-Step Safety]
    \label{lem:interp_one_step}
For all $Q$ (monotone), $\mathcal{I}$ (well-behaved), $\sigma$, $e$, $e'$, $\Sigma$, $\tau$, $\gamma$, $\gamma'$, if
\[
\epsilon;\Sigma\vdash e : \tau \mid \gamma
\qquad
\vdash \sigma : \Sigma
\qquad
\delta\le\Sigma
\qquad
\sigma,e\rightarrow^{\gamma'}_Q \sigma',e'
\qquad
\mathsf{NonTrivial}(\gamma)
\]
then there exist $\Sigma'$, $\gamma''$ such that
\[\epsilon;\Sigma'\vdash e' : \tau \mid \gamma''
\quad
\vdash\sigma':\Sigma'
\quad
\Sigma\le\Sigma'
\quad
\gamma'\rhd\gamma''\sqsubseteq\gamma
\quad
\mathcal{I}(\epsilon;\Sigma')(\gamma')(\sigma,\sigma')
\quad
\mathsf{NonTrivial}(\gamma'')\]
\end{lemma}
\begin{proof}
    This proceeds exactly as Lemma \ref{lem:onestep_preservation} (One Step Type Preservation) (which it extends), additionally applying properties of the interpretation as needed.  In particular, most structural reductions yield single-step effects of $I$ and do not modify the state, application reduction performs substitutions on the effects that preserve the intended interpretations, and the primitive reduction case follows from the assumption that the language parameters are interpretation-consistent with $\mathcal{I}$.
\end{proof}

\begin{lemma}[Interpreted Safety]
    \label{lem:interp_safety}
For all $Q$ (monotone), $\mathcal{I}$ (well-behaved), $\sigma$, $e$, $e'$, $\Sigma$, $\tau$, $\gamma$, and $\gamma'$, if
\[\epsilon;\Sigma\vdash e : \tau \mid \gamma
\qquad
\vdash \sigma : \Sigma
\qquad
\delta\le\Sigma
\qquad
\sigma,e\dhxrightarrow{\gamma'}_Q \sigma',e'
\qquad
\mathsf{NonTrivial}(\gamma)
\]
then there exist $\Sigma'$, $\gamma''$ such that
\[
\epsilon;\Sigma'\vdash e' : \tau \mid \gamma''
\quad
\vdash\sigma':\Sigma'
\quad
\Sigma\le\Sigma'
\quad
\gamma'\rhd\gamma''\sqsubseteq\gamma
\quad
\mathcal{I}(\epsilon;\Sigma')(\gamma')(\sigma,\sigma')
\quad
\mathsf{NonTrivial}(\gamma)
\]
\end{lemma}
\begin{proof}
    By induction on the transitive reduction relation and Lemma \ref{lem:interp_one_step}.
\end{proof}

This development considers the semantic interpretation of effects, but not types.  This is typical of syntactic proofs of type safety for effect systems applied to Java or ML like languages, which often incorporate a notion of a relationship between effects and state into their type safety proof, without going so far as to give semantic interpretations of types. For example, using Lemma \ref{lem:interp_safety} with Example \ref{ex:history_effect_interp} recovers a version of Skalka et al.'s type safety proof relating history effects to their accumulating semantics (see Section \ref{sec:modeling2}).

\section{Relationships to Semantic Notions of Effects}
\label{sec:semantics}

Our notion of an effect quantale is motivated by generalizing directly from the form of
effect-based type judgments (Section \ref{sec:bg}).  In parallel with our work, there has been a line of
semantically-oriented work to generalize monadic semantics to capture sequential effect systems
(indeed, this is where our use of the term ``sequential effect system'' originates).  Here we
compare to several recent developments: Tate's productors (and algebraic presentation as
effectoids)~\cite{tate13}, Katsumata's parametric (later, graded) monads~\cite{katsumata14}, and Mycroft,
Orchard, and Petricek's joinads (and algebraic presentation in terms of joinoids)~\cite{mycroft16}.
We also compare briefly to Atkey's parameterized monads~\cite{atkey2009parameterised}.

Most of this work is done primarily in the setting of category theory, by incrementally considering the categorical semantics of desirable effect combinations (in contrast to our work, working by abstracting actual effect systems).  Each piece of work also couples the semantic development with an algebraic structure that yields an appropriate categorical structure, and we can compare directly with those without detailing the categorical semantics. (Though as outlined in Section \ref{para:abstract_soundness}, this ignores the distinction between the algebraic structure of effects and the desired semantics of effects.)

None of the following systems consider value-dependent effects or effect polymorphism, or give more than a passing mention of iteration, though given the generality of the technical machinery, it is unlikely that any of the following are fundamentally incompatible with these ideas.  Since none of them have considered polymorphism explicitly, the systems with partial effect combinations~\cite{tate13,atkey2009parameterised} did not address the potential for possibly-invalid effects.
One line of work mentions iteration in a discussion of future work, and later assumes the existence of a particularly strong total iteration operator.
We showed (Section \ref{sec:soundness}) that effect quantales are compatible with these ideas.

Overall, Tate's work studies structures which are strict generalizations of effect quantales (i.e., impose fewer constraints than effect quantales), and any effect quantale can be translated directly to Tate's effectoids or (with a minor adjustment) Katsumata's pre-ordered effect monoid.  Tate and Katsumata demonstrate that their structures are \emph{necessary} to capture certain parts of any sequential effect system --- a powerful general claim.  By contrast, we demonstrate that with just a bit more structure than Tate and Katsumata, effect quantales become \emph{sufficient} to formalize a range of real sequential effect systems.
Mycroft et al.'s work does consider a full programming language, but studies different control flow constructs than we do (block-structured parallelism rather than iteration).
Atkey's earlier and influential work~\cite{atkey2009parameterised} covers an important class of sequential effect systems, but does not include examples that fit within the other frameworks, such as atomicity or Tate's \textsc{Crit} effects (Definition \ref{def:crit-effects}).
Figure \ref{fig:frameworksummary} summarizes some of the key aspects of these frameworks, though there are many nuances considered below.

\begin{figure}
\begin{tabular}{|p{2.5cm}|p{1.5cm}|c|p{1.7cm}|p{1cm}|c|c|}
\hline
Framework & Effect \mbox{Order} & Units & Sequencing & Distrib.\ Laws & Iteration & Parallelism\\
\hline
Effectoids / \mbox{Productoids} \cite{tate13} & Preorder & Per-element & Relation & M & $\Circle$ & $\Circle$ \\
\hline
Graded Monads \cite{katsumata14} & Preorder & Global & Total Monoid & M & $\Circle$ & $\Circle$ \\
\hline
Graded Joinads \cite{mycroft16} & Preorder & Global & Total Monoid & M,R & $\LEFTcircle$ & $\CIRCLE$\\
\hline
Parameterized Monads \cite{atkey2009parameterised} & Substate \mbox{Preorder} $\Rightarrow$ \mbox{Preorder} & Per-element & Partial Semigroup & M & $\Circle$ & $\Circle$ \\
\hline
Effect Quantales (this paper) & Partial Join-Semilattice & Global & Partial Monoid & M,L,R & $\CIRCLE$ & $\Circle$\\
\hline
\end{tabular}
\caption{Summary of the major characteristics of various frameworks for sequential effect systems. In the distributive laws column, M indicates sequencing is monotone with respect to the ordering,  R and L indicate the existence of right and/or left distributivity laws with respect to some kind of join.
For iteration and parallelism, $\Circle$ indicates the topic was not considered by the framework, $\LEFTcircle$ indicates incomplete speculation, and $\CIRCLE$ indicates a reasonably complete solution was offered.}
\label{fig:frameworksummary}
\end{figure}

\subsection{Productors and Effectoids}
\label{sec:tate}
Tate~\cite{tate13} sought to design the maximally general semantic notion of sequential composition, proposing a structure called \emph{productors}, and a corresponding algebraic structure for source-level effects called an \emph{effector}.  Effectors, however, include models of analyses that are not strictly modular (i.e., effectors can specialize certain patterns in source code for more precise effects than one would obtain by simply composing effects of subexpressions)~\cite[Section 5]{tate13}.
To model the strictly compositional cases like syntactic type-and-effect systems, he also defines a semi-strict variant called an \emph{effectoid} (using slightly different notation):
\begin{definition}[Effectoid~\cite{tate13}]
An \emph{effectoid} is a set \textsc{Eff} with a unary relation $\mathsf{Base}(-)$, a binary relation $-\le-$, and a ternary relation $-\fcmp-\mapsto-$, satisfying
\begin{itemize}
\item \textsf{Identity}:
$\forall\varepsilon,\varepsilon'\ldotp
	(\exists \varepsilon_\ell\ldotp \mathsf{Base}(\varepsilon_\ell)\wedge \varepsilon_\ell\fcmp\varepsilon\mapsto\varepsilon')
	\Leftrightarrow
	\varepsilon\le\varepsilon'
	\Leftrightarrow
	(\exists \varepsilon_r\ldotp \mathsf{Base}(\varepsilon_r)\wedge \varepsilon\fcmp\varepsilon_r\mapsto\varepsilon')
$
\item \textsf{Associativity}:
$\forall\varepsilon_1,\varepsilon_2\varepsilon_3,\varepsilon\ldotp
	(\exists\overline{\varepsilon}\ldotp \varepsilon_1\fcmp\varepsilon_2\mapsto\overline{\varepsilon}\wedge\overline{\varepsilon}\fcmp\varepsilon_3\mapsto\varepsilon)
	\Leftrightarrow
	(\exists\hat{\varepsilon}\ldotp \varepsilon_2\fcmp\varepsilon_3\mapsto\hat{\varepsilon}\wedge\varepsilon_1\fcmp\hat{\varepsilon}\mapsto\varepsilon)$
\item \textsf{Reflexive Congruence}:
	\begin{itemize}
	\item $\forall\varepsilon\ldotp \varepsilon\le\varepsilon$
	\item $\forall\varepsilon,\varepsilon'\ldotp\mathsf{Base}(\varepsilon)\wedge\varepsilon\le\varepsilon'\Longrightarrow\mathsf{Base}(\varepsilon')$
	\item $\forall\varepsilon_1,\varepsilon_2,\varepsilon,\varepsilon'\ldotp \varepsilon_1\fcmp\varepsilon_2\mapsto\varepsilon\wedge\varepsilon\le\varepsilon'\Longrightarrow\varepsilon_1\fcmp\varepsilon_2\mapsto\varepsilon'$
	\end{itemize}
\end{itemize}
\end{definition}
Intuitively, \textsf{Base} identifies effects that are valid for programs with ``no'' effect --- e.g., pure programs, empty programs (this is a generalization of the unit $I$ in effect quantales).  Tate refers to such effects as \emph{centric}.  The binary relation $\le$ is clearly a pre-order for subeffecting (the axioms do not imply antisymmetry), and $-\fcmp-\mapsto-$ is (relational) sequential composition.  
The required properties imply that the effectoid's sequential composition can be read as a non-deterministic function producing a minimal composed effect \emph{or any supereffect thereof}, given that the sequential composition relation includes left and right units for any effect, and that \textsf{Base} and the last position of composition respect the partial order on effects.  Note the use of ``minimal'' rather than ``minimum'' --- effectoids do not require a least element in any of these.

Given Tate's aim at maximal generality (while retaining enough structure for interesting reasoning about sequential composition), it is perhaps unsurprising that every effect quantale yields an effectoid by flattening the monoid and semilattice structure into the appropriate relations:
\begin{lemma}[Quantale Effectoids]
    \label{lem:quantale_effectoids}
For any nontrivial effect quantale $Q$, there exists an effectoid $E$ with the following structure:
\begin{itemize}
\item $\textsc{Eff}=E_Q$
\item $\mathsf{Base}(a) \overset{def}{=} I \sqsubseteq a$
\item $a \le b \overset{def}{=} a \sqsubseteq b$
\item $a\fcmp b \mapsto c \overset{def}{=} a\rhd b\sqsubseteq c$ \end{itemize}
\end{lemma}
\begin{proof}
The laws follow almost directly from the effect quantale laws.  In the identity property, both left and right units are always chosen to be $I$.  Associativity follows directly from associativity of $\rhd$ and monotonicity.  The reflexive congruence laws follow directly from the definition (and transitivity) of $\sqsubseteq$.
\end{proof}

Many effectoids directly correspond to effect quantales.
Tate calls out the class of effectoids with a least result for any defined sequential composition, which he calls \emph{principalled}, which are equivalent to congruently preordered partial monoids --- i.e., sequential composition is monotone (when defined) with respect to the partial order.
Principalled effectoids with additional structure on the preorder give rise to effect quantales:
\begin{lemma}[Effect Quantales from Effectoids]
    \label{lem:eq_from_effectoid}
For any principalled effectoid $E$ with a least centric element, where each pair of elements has a least upper bound or no common upper bound (hence a partial join), where sequencing an element (on either side) with a join result is equivalent to joining the results of sequencing, there exists an unique effect quantale $Q$ such that:
\begin{itemize}
\item $E_Q=\textsc{Eff}_E$
\item $\sqcup$ performs the assumed partial binary join where defined (least common upper bound when it exists)
\item $a\rhd b$ produces the least $c$ such that $a\fcmp b\mapsto c$ when defined
 such that $a\fcmp b\mapsto c$ (by assumption, $a\fcmp b$ is undefined or has a least element).
\item $I$ is assumed the least centric element
\end{itemize}
\end{lemma}
\begin{proof}
The existence of a partial join (returning a least, not minimal, upper bound) forces the preorder to a partial order.
The remaining restrictions essentially require a partial join such that sequencing distributes over join.
\end{proof}

Essentially, effect quantales correspond to effectoids where the preorder is a partial join-semilattice satisfying the distributivity properties (the distributivity properties imply the congruent ordering, so technically the principality requirement above is satisfied for all effect quantales).
Tate notes that principalled (i.e., monotone) effectoids are common (including his motivating example, given as an effect quantale in Definition \ref{def:crit-effects}) because they simplify type checking and inference, so the only ``extra'' requirements imposed above are the global unit, and refining the preorder to a partial join semilattice satisfying the distributivity requirements.
The distributivity requirement does not follow from other requirements: even if the partial order is a (total) lattice or complete lattice, this does not automatically satisfy the distributive laws~\cite{birkhoff}.\footnote{Consider taking the meet of a lattice as sequential composition; then the distributive laws hold only if the lattice is distributive (such as in the case of Must Effects in Definition \ref{def:must_effects}).}

Lemma \ref{lem:eq_from_effectoid}'s restrictions on the effectoids involved may seem severe, and in general they are.  However, we are not necessarily interested in \emph{all} effectoids: recall that Tate's focus was on seeking the most general form of sequential effects that admitted subsumption: actual computational effects and concrete sequential effect systems have considerably more structure than effectoids require, and some additional structure is trivial to add.
Adding a global unit where one does not exist is straightforward, and most known examples of effect systems already correspond to partial join-semilattices (we see shortly this is a common observation).  The distributivity properties hold trivially for traditional one-operation commutative effect systems (which are simply join semilattices), and for all the concrete sequential effect systems we have considered in Section \ref{sec:modeling}, which were proposed prior to the existence of most of these generic frameworks (as outlined in Section \ref{sec:quantales}, the distributivity laws correspond to common code transformations, so an effect system that does not satisfy them is likely to have unexpected behavior on real programs).
This strongly suggests that our generalization from the type judgments of a few specific effect systems, rather than from semantic notions, did not cost much in the way of generality.

\subsection{Parametric Monads, a.k.a.\ Graded Monads}
\label{sec:katsumata}
Katsumata~\cite{katsumata14} pursues an independent notion of general sequential composition, where effects are formalized semantically as a form of type refining monad: a $T\;e\;\sigma$ is a monadic computation producing an element of type $\sigma$, whose effect is bounded by $e$ (which classifies a subset of such computations).  Based on general observations, Katsumata speculates that sequential effects form at least a pre-ordered (total) monoid, and goes on to validate that this leads to a rich and broadly applicable theory (among other interesting results related to the notion of effects as refinements of computations).  Katsumata shows categorically that these \emph{parametric monads} (now discussed as \emph{graded monads}~\cite{fujii2016towards}, to avoid confusion with other forms of indexing and to follow earlier mathematical practice~\cite{smirnov2008graded}) are also a specialization of Tate's productors, exactly when the productor is induced by an effectoid derived from a preordered monoid --- i.e., a \emph{total} principalled effectoid~\cite[Section 6.1]{katsumata14}.

The primary technical differences between graded monads and effect quantales stem from the differences in sequencing and partial order.  Graded monads require a total sequencing operation, while effect quantales admit partial sequencing.  However, it is straightforward to add an extra element to the monoid to represent being ``undefined'' and simply reject programs given this ``extra'' effect, as the original work on effect quantales did~\cite{ecoop17}; so this is not a critical difference.
Graded monads support more relaxed notions of order --- pre-orders --- than the partial join-semilattice required by effect quantales.
As with effectoids, restricting a graded monad to satisfy the effect quantale laws requires the preorder to be a partial join-semilattice satisfying the distributivity laws.  This is a smaller distinction than with effectoids, since graded monads already possess a global unit.
Intuitively, the partial join-semilattice restriction means effect quantales are limited to cases where if there is \emph{any} common upper bound of two effects, there is a \emph{unique, best} (least) such upper bound, while graded monads impose no requirements on upper bounds (and without a join, no distributivity).

Katsumata notes that most interesting examples seem to not require the flexibility of pre-orders, and all examples used by Katsumata are total join semilattices satisfying the effect quantale distributivity laws (which are then expressible as effect quantales).
In addition to not knowing of systems requiring the (full) extra generality,\footnote{History effects (Section \ref{sec:history_effects}) are defined by a pre-order and then quotiented by an equivalence relation to obtain an effect quantale,  But they do not exploit the full generality of pre-orders, since quotienting by equivalence still ensures there is a unique least upper bound on any pair of effects (the equivalence class of joining the two effects with a $|$ operator), which is not necessarily true for all quotients of pre-orders.} joins (rather than general partial orders) play a significant role in most work on type inference, including for the one class of sequential effect systems with established inference results~\cite{skalka2008types,Skalka2008}.

The final major distinction between graded monads and effect quantales is that the former impose no distributivity requirements, while the latter require both distributive laws. The concrete effect systems that motivated effect quantales (Section \ref{sec:bg}) all satisfy the distributive law, but until very recently it was unclear whether there were meaningful systems that would nearly be effect quantales, except for the distributivity requirements.  While this paper was under review, \citet{ivaskovic2020dataflow} demonstrated that dataflow analyses can be modeled as effect systems, modeling a range of standard dataflow analyses as graded monads.  Along the way, they demonstrated an example of a meaningful effect system with a join on effects that is \emph{not} possible to model as an effect quantale: basic constant propagation. 
The reason is that basic constant propagation does not satisfy the effect quantale requirement that sequencing distributes over joins in both directions.  
The example they give~\cite[Figure 2]{ivaskovic2020dataflow} is the basic constant propagation for the programs
\begin{lstlisting}
    if (...) { x = 1; y = 2; } else { x = 2; y = 1; }; z = x + y;
\end{lstlisting}
and
\begin{lstlisting}
    if (...) { x = 1; y = 2; z = x + y; } else { x = 2; y = 1; z = x + y; }
\end{lstlisting}
In the first program, the effects of the individual branches track the values of \lstinline|x| and \lstinline|y| after each branch's execution, but joining those effects per the conditional rule results in an effect mapping both variables to $\top$ (non-constant). This in turn means the program's overall effect considers \lstinline|z| non-constant, as the sum on non-constant values.
In the second program, the assignment to \lstinline|z| is recognized \emph{within each branch} as holding the constant value 3, with both branches' individual effects reflecting this, so the join of those effects maps both \lstinline|x| and \lstinline|y| to $\top$ as before, but additionally maps \lstinline|z| to the constant value 3.

Iva{\v{s}}ković et al.\ take this as an indication that it is better not to require distributivity since it would preclude basic constant propagation, though they note that all of their other (more sophisticated) examples \emph{are} distributive.
However, basic constant propagation is also the canonical example of why non-distributive analyses are often avoided in compilers and static analysis: the results are brittle. The conditional-related refactorings we described to justify the distributivity laws are minor, local, and obviously behavior-preserving. Yet they change the results of non-distributive analyses (as in the example above), which for compiler optimizations can lead to unpredictable performance fluctuations from seemingly trivial changes, frustrating developers.  This is part of why the most influential dataflow analysis frameworks (e.g., IFDS~\cite{reps1995precise} and IDE~\cite{sagiv1996precise}) emphasize distributivity (the `D' in both acronyms stands for Distributive), and why distributive variants of otherwise non-distributive analyses (including constant propagation) are well-studied~\cite{grove1993interprocedural}.
Iva{\v{s}}ković et al.'s example is an important example for understanding the limits of effect quantales imposed by the distributive laws, but we feel their presence in the effect quantale axioms is justified by the fact that such counterexamples are avoided in practice.

\subsection{Parameterized Monads}
\label{sec:param_monads}
Prior to the appearance of productors and graded monads, \citet{atkey2009parameterised} proposed \emph{parameterized monads}: monads indexed by elements corresponding to state before and after the computation --- essentially computational pre- and postconditions.  Essentially, a type $M~a~b~\tau$ is a computation that, when run with an initial resource corresponding to $a$, produces a $\tau$ and resource corresponding to $b$.  
For easier comparison to other systems discussed here and our work, we will borrow \citet{tate13}'s reformulation of parameterized monads as an effectoid:
\begin{definition}[Parameterized Monads as Effectoids]
    \label{def:param_monad}
    A \emph{parameterized monad}~\cite{atkey2009parameterised} is presented as an effectoid with the following structure~\cite[Section 7]{tate13}, assuming a substate preorder $(S,\le)$:
    \begin{itemize}
        \item $\textsc{Eff}=S\times S$ (taking $(s,s')$ to be a computation transitioning from $s$ to $s'$)
        \item $\mathsf{Base}((s,s'))\;=\;s\le s'$
        \item $(s,s')\le(t,t') \;=\; t\le s\land s'\le t'$
        \item $(s,s')\fcmp(t,t')\mapsto(u,u') \; = \; u\le s\land s'\le t\land t'\le u'$
    \end{itemize}
\end{definition}
Sequential composition of parameterized monads essentially requires matching the postcondition of the first computation with the precondition of the second, but essentially incorporates Hoare's rule of consequence applied to both computations.
His proof-theoretic presentation lacks an explicit unit, but semantically the units of the parameterized monad are monadic units with identical pre- and post-conditions (modulo subsumption above).

Atkey demonstrates that parameterized monads are quite flexible, encoding examples as diverse as certain session types (pre-condition is the postcondition prefixed with the operation performed, which could be extended to history effects~\cite{skalka2008types}) and composable continuations.

There are two key differences between parameterized monads and effect quantales: a preorder on states instead of a partial join-semilattice on effects (plus distributivity), and lack of a single unit element.
A single global unit (least centric element in Tate's terminology) can be added easily to any given parameterized monad's algebra, as an element less than each centric effect (i.e., $\forall s\ldotp I\le(s,s))$.
The preorder on substates $S$ yields the above preorder on effects, though quotienting can transform both into partial orders (though not necessarily partial join-semilattices).
With similar restrictions to those we suggested for effectoids and graded monads, we can obtain an effect quantale from a parameterized monad:
\begin{lemma}[Effect Quantales from Parameterized Monads]
    \label{lem:eq_from_param_monad}
    Given a parameterized monad (as an effectoid $(S\times S,\mathsf{Base},\le,-\fcmp-\mapsto-)$) where the preorder $(S,\le)$ on substates is a partial order, and where:
            \begin{itemize}
                \item For any pair of states, there are either no common upper bounds or a least upper bound (i.e., states have a partial join $\vee$)
                \item For any pair of states, there are either no common lower bounds or a greatest lower bound (i.e., states have a partial meet $\wedge$)
\end{itemize}
    there exists an unique effect quantale where:
    \begin{itemize}
        \item $E=(S\times S)\uplus\{I\}$
        \item $(s,s')\sqcup(t,t')=(s\wedge t, s'\vee t')$ where the meet and join are both defined,
              and $I\sqcup(s,s')=(s,s')\sqcup I=(s,s')$ when $s\le s'$, otherwise undefined
        \item $(s,s')\rhd(t,t')=(s,t')$ when $s'\le t$, and is otherwise undefined, while $(s,s')\rhd I=I\rhd(s,s')=(s,s')$
        \item $I=I$
    \end{itemize}
\end{lemma}
\begin{proof}
    Notice that the partial join induces a partial order consistent with the effectoid:
    \[ (s,s')\sqcup(t,t')=(s\wedge t, s'\vee t')=(t,t')\Leftrightarrow s\wedge t=t~\textrm{and}~s'\vee t'=t' \Leftrightarrow t\le s~\textrm{and}~s'\le t' \Leftrightarrow (s,s')\le_{\textsc{Eff}}(t,t') \]
    Also, $(s,s')\fcmp(t,t')\mapsto(u,u')$ if and only if $(s,s')\rhd(t,t')\sqsubseteq(u,u')$, so $\rhd$ returns the least ``result'' of $-\fcmp-\mapsto-$ (so the effectoid is principalled).
    Both operators are appropriately associative (and for join, commutative).
    For the right distributivity law:
    \begin{itemize}
    \item $((s,s')\sqcup(t,t'))\rhd(u,u')=(s\wedge t, s'\vee t')\rhd(u,u')=(s\wedge t,u')$ if the meet and join used are both defined, $s'\le u$, and $t'\le u$
    \item $((s,s')\rhd(u,u'))\sqcup((t,t')\rhd(u,u')) = (s,u')\sqcup(t,u') = (s\wedge t, u')$ if the meet is defined, and $s'\le u$ and $t'\le u$
    \end{itemize}
    These are equivalent aside from the first requiring definition of the join $s'\vee t'$. But since both sides require $s'\le u$ and $t'\le u$, $u$ is a common upper bound, and by assumption there is a least upper bound for $s'$ and $t'$, which defines the partial join of those elements for the right side as well (or conversely, make the requirement in the first series of equations that the join is defined redundant).
    The left distributivity law is analogous.
\end{proof}
We believe these restrictions are similarly modest in practice like our claims for effectoids and graded monads. Adding the global unit is straightforward, and a preorder on substates can always be quotiented to obtain a partial order, making the only significant restrictions the requirement of having partial joins and meets of substates.
However, as with Katsumata, the cases Atkey considers have explicit or implicit partial joins (such as the choice operator for session types), sometimes trivially (in several cases the preorder is actually equality, which is trivially a partial order with our required properties).  So in practice this extra flexibility from using general pre-orders is also of unclear practical value.

Parameterized monads and effect quantales are technically incomparable: there are some parameterized monads inexpressible as effect quantales (per above, those lacking partial joins and meets), but there are many effect quantales inexpressible as parameterized monads.
Tate gives some reasons based on categorical semantics that effect systems like \textsc{Crit} cannot be modeled as parameterized monads: the categorical definition of parameterized monads assumes a function from pairs of states to semantics, which means there is exactly one program transitioning between any two states, which is incompatible with the existence of multiple distinct effects that can be run between two states (like $\varepsilon$ and $\mathsf{critical}$ in \textsc{Crit}).

But these mismatches also arise purely from the effect algebras, without regard to semantics.
In particular, global units have non-obvious consequences in parameterized monads:
\begin{lemma}[Global Units in Parameterized Monads Imply a Bounded Preorder and Minimal Unit]
    Taking a parameterized monad formulated as an effectoid $(S\times S,\mathsf{Base},\le,-\fcmp-\mapsto-)$, then if there exists a global unit $(s,s')$, then it is $\le$ all other effects, $s$ is an upper bound on all states, and $s'$ is a lower bound on all states.
\end{lemma}
\begin{proof}
Assume $(s,s')$ is a global unit.  This means that for any other effect $(u,u')$:
\[(s,s')\fcmp(u,u')\mapsto(u,u') \; \Leftrightarrow \; u\le s\land s'\le u\land u'\le u'\]
\[(u,u')\fcmp(s,s')\mapsto(u,u') \; \Leftrightarrow \; u\le u\land u'\le s\land s'\le u'\]
Thus, $\forall u,u'\ldotp u\le s \land s'\le u\land u'\le s\land s'\le u'$.
This means the preorder on $S$ has a bottom and a top (up to equivalence), making the unit $(s,s')=(\top,\bot)$. This is then the least element in the preorder order on effects (again, up to equivalence): it is less than \emph{all} effects.
\end{proof}
This is incompatible with examples like our locking effect quantale (Definition \ref{def:locking}), where taking states to be multisets of locks there is no upper bound. More generally, many reasonable notions of state have either no sensible lower bound or no sensible upper bound.  Consider the extension order in states modeled by separation algebras~\cite{calcagno2007local}: there is no greatest state, only a least (the empty state), so a system of pre- and postconditions as separation algebras cannot be expressed as a parameterized monad with global unit, even though one can exist (the effect with no footprint, where both pre- and postconditions are $\mathsf{emp}$) if sequencing could incorporate a version of framing similar to bi-abduction~\cite{calcagno2011compositional} akin to what our formulations of the locking-related effect quantales do (Definitions \ref{def:locking} and \ref{def:dlf_suenaga}).  In a sense, the inability for parameterized monads to express \emph{transformations} of states is a limitation.
So there are real effect systems expressible as effect quantales that are inexpressible as parameterized monads.
There also exist effect quantales where unit is the \emph{greatest} element of the partial order on effects:
\begin{definition}[Lower Bound Count Effects]
    The effect quantale that establishes lower bounds on the number of times an action of interest occurs is defined as:
    \begin{itemize}
        \item $E=\mathbb{N}$
        \item $x\rhd y=x+y$
        \item $x\sqcup y =\mathsf{min}(x,y)$
        \item $I=0$
    \end{itemize}
\end{definition}

\subsection{Joinads and Joinoids}
\label{sec:joinoids}
Mycroft, Orchard, and Petricek~\cite{mycroft16} further extend graded monads to \emph{graded conditional joinads}, and similar to Tate, give a class of algebraic structures (joinoids) that give rise to their semantic structures (joinads).  As their base, they take graded monads. They then further assume a ternary conditional operator $?{}:(-,-,-)$ (modeling conditionals whose branch approximation may depend on the conditional expression's effect), and parallel composition $\&$ suitable for fork-join style concurrency.

Their ternary operator is motivated by considerations of sophisticated effects, in particular control effects like backtracking search.  From their ternary operator, they derive a binary join, and therefore a partial order.  However, their required laws for the ternary operator include only a right distributivity law because effects from the conditional expression itself do not in general distribute into or out of the branches.  Thus their derived structure satisfies only the right distributivity law $(a\sqcup b)\rhd c = (a\rhd c)\sqcup(b\rhd c)$, and not, in general, the left-sided version.  They also do not require ``commutativity'' of the branch arguments. This means that joinoids, in general, do not give rise to effect quantales --- some amount of structure is not necessarily present --- and that in general they validate fewer equivalences between effects.
Their motivations justify these choices (they offer logic programming as a domain where the left distributivity and commutativity conditions would not hold), and our focus on strict sequential languages justifies our additional requirements.
As with graded monads, joinads require sequential composition to be total, yet relax the kind of ordering on the monoid.

Joinads originally arose as an extension to monads that captures a class of combinators typical of composing parallel and concurrent programs in Haskell, in particular a \emph{join} (unrelated to lattices) operator of type $M\;A\rightarrow M\;B\rightarrow M\;(A\times B)$.  This is a natural model of fork-join-style parallel execution, and gives rise to the $\&$ operator of joinoids, which appears appropriate to model the corresponding notion in systems like Nielson and Nielson's effect system for CML communication behaviors~\cite{nielson1993cml}, which is beyond the space of operations considered for effect quantales.  However, $\&$ is inadequate for modeling the unstructured parallelism (i.e., explicit thread creation and termination, or task-based parallelism) found in most concurrent programming languages, so we did not consider such composition when deriving effect quantales.  We would like to eventually extend effect quantales for unstructured concurrent programming: this is likely to include adapting ideas from concurrent program logics that join asynchronously~\cite{dodds09}, but any adequate solution should be able to induce an operation satisfying the requirements of joinoids' parallel composition.

Ultimately, any effect quantale gives rise to a joinoid, by using the effect quantale's join for both parallel composition and to induce the ternary operator outlined above (and possibly adding a point, as with Katsumata's work, to make the operations total), though collapsing parallel and alternative composition obviously fails to capture many intended concurrency semantics. Mycroft et al.\ point out Milner's suggestion of defining parallel composition as $\chi\&\chi'{=}(\chi\rhd\chi')\sqcup(\chi'\rhd\chi)$, which we could also use to induce a joinoid from an effect quantale, but this is only valid for some effect systems (e.g., it assumes each thread's effects are atomic, which is generally untrue).

\subsection{Fixed Points}
Mycroft et al.\ also give brief consideration to providing iteration operators through the existence of fixed points, noting the possibility of adding one type of fixed point categorically, which carried the undesirable side effect of requiring sequential composition to be idempotent: $\forall b\ldotp b\rhd b = b$.  This is clearly too strict, as it is violated by nearly every example from the literature we examined in Section \ref{sec:modeling}.  They take this as an indication that every operation should be explicitly provided by an algebra, rather than attempting to derive operators.

By contrast, our closure operator approach not only imposes semantics that are by construction compatible with a given sequential composition operator, but critically coincide with manual definitions for existing systems.  We have also shown that for broad classes of sequential effect systems, our construction yields the most precise possible iteration operator, so it need not be explicitly defined for the algebra (though implementations may still require an explicit definition).

\citet{mycroft16}'s results seem to inform contemporaneous work by overlapping authors~\cite{orchard2016effects}, who relate sequential effects to session types using an unnamed effect algebra with (total) join and sequencing and a separately defined (total) iteration operator.  As with joinoids, only right distributivity is assumed (though Orchard and Yoshida do assume join is commutative). For the iteration operator, they require much stronger axioms for iteration --- in our notation, they require $x^*=I\sqcup(x\rhd x^*)=I\sqcup(x^*\rhd x)$, which is equivalent to the requirements of iteration in Kleene Algebras~\cite{KOZEN1994366}; our discussion of Kleene Algebras in Section \ref{sec:kleene} explains why we feel this is still too strong a requirement, even though for principally iterable effect quantales the construction in Section \ref{sec:iteration} induces operators satisfying the stronger axioms.
In some sense, the key insight behind our resolution of the problem Mycroft et al.\ identify is the recognition that iteration should be partial.

\section{Modeling Prior Effect Systems in a Generic Framework}
\label{sec:modeling2}
This section demonstrates that we can model significant prior type systems by embedding into our core language from Section \ref{sec:soundness}.  Embedding here means a type-and-effect-preserving translation.  Our language is generic, but clearly lacks concurrency, exception handling, and other concrete computational effects.  
We demonstrate two instantiations.  First, we show how to instantiate the framework to model a significant fraction of a joint locking and atomicity effect system (without spawning new threads), mostly as a demonstration that the framework parameters permit non-trivial instantiations.
Second, we give an instantiation that captures the main constructs of Skalka et al.'s history effects, and show how to recover their soundness proof for those constructs in our framework.

\subsection{Types for Safe Locking and Atomicity}

This section develops a hybrid of Flanagan and Abadi's \emph{Types for Safe Locking}~\cite{safelocking99} and Flanagan and Qadeer's \emph{Types for Atomicity}~\cite{flanagan2003tldi}, further extended to track locks in a flow-sensitive manner (the former uses \texttt{synchronized} blocks, the latter does not track locks itself), in order to demonstrate a rich instantiation of our framework.
We model a single thread's perspective on the type system: we omit constructs for spawning new threads, but include locks and access to heap cells, whose types indicate a guarding lock.
The type system tracks the locks held at each program point, flow sensitively for explicit lock acquire and release primitives, rather than for the synchronized blocks treated in the original work.

\begin{figure}
\[
\begin{array}{l}
Q(X) = \mathcal{L}(X)\otimes\mathcal{A}\\
M \in \mathsf{LockNames}\rightharpoonup((\mathsf{ThreadID}\times\mathbb{N})~\mathsf{option})\\
H \in \mathsf{Location}\rightharpoonup\mathsf{Term}\\
\mathsf{State} = \mathsf{ThreadID}\times M\times H\\
\mathsf{StateEnv} = (\mathsf{LockNames}\cup\mathsf{Location})\rightharpoonup\mathsf{Type}\\
K(\mathsf{lock}) = \star\\
K(\mathsf{ref}) = \star\Rightarrow\star\Rightarrow\star\\
\end{array}
\begin{array}{c}
\inferrule{
    \forall l\in\mathsf{dom}(m)\ldotp \Sigma(l)=\mathsf{lock}\\\\
    \forall r\in\mathsf{dom}(h)\ldotp \epsilon;\Sigma\vdash h(r):\Sigma(r)\mid I
}{
    Q\vdash (m,h) : \Sigma
}
\end{array}
\]
\[
\begin{array}{rcl}
\delta(\mathsf{new\_lock}) &=& \mathsf{unit}\overset{B}{\rightarrow}\mathsf{lock}\\
\delta(\mathsf{acquire}) &=& \Pi x:\mathsf{lock}\xrightarrow{(\emptyset,\{x\})\otimes R}\mathsf{unit}\\
\delta(\mathsf{release}) &=& \Pi x:\mathsf{lock}\xrightarrow{(\{x\},\emptyset)\otimes L}\mathsf{unit}\\
\delta(\mathsf{alloc}) &=& \Pi x:\mathsf{lock}\overset{B}{\rightarrow}\forall \alpha::\star\overset{B}{\rightarrow} \tau \overset{B}{\rightarrow}\mathsf{ref}\;S(x)\;\tau\\
\delta(\mathsf{read}) &=& \Pi x:\mathsf{lock}\overset{B}{\rightarrow}\forall \alpha::\star\overset{B}{\rightarrow} \mathsf{ref}\;S(x)\;\tau {\xrightarrow{(\{x\},\{x\})\otimes B}} \tau\\
\delta(\mathsf{write}) &=& \Pi x:\mathsf{lock}\overset{B}{\rightarrow}\forall \alpha::\star\overset{B}{\rightarrow} \mathsf{ref}\;S(x)\;\tau \overset{B}{\rightarrow}\tau{\xrightarrow{(\{x\},\{x\})\otimes B}} \tau\\
\end{array}
\]
\[
\begin{array}{l}
\llbracket\mathsf{new\_lock}\;\_\rrbracket((i,m,h))(\Sigma) = l,(i,m[l\mapsto\mathsf{false}],h),\Sigma[l\mapsto\mathsf{lock}]~\textrm{for $l\not\in\mathsf{dom}(m)$}\\
\llbracket\mathsf{acquire}\;l\rrbracket((i,m[l\mapsto \mathsf{None}],h))(\Sigma) = (),(i,m[l\mapsto (i,1)],h),\Sigma\\
\llbracket\mathsf{acquire}\;l\rrbracket((i,m[l\mapsto (j,n)],h))(\Sigma) = (),(i,m[l\mapsto (j,n+1)],h),\Sigma~\textrm{when $j=i$}\\
\llbracket\mathsf{release}\rrbracket((i,m[l\mapsto (j,n)],h))(\Sigma) = (),(i,m[l\mapsto (j,n-1)],h),\Sigma~\textrm{when $n>1$ and $j=i$}\\
\llbracket\mathsf{release}\rrbracket((i,m[l\mapsto (j,1)],h))(\Sigma) = (),(i,m[l\mapsto \mathsf{None}],h),\Sigma~\textrm{when $j=i$}\\
\llbracket\mathsf{alloc}\;l\;\tau\;v\rrbracket((i,m,h))(\Sigma) = \ell,(i,m,h[\ell\mapsto v]),\Sigma[\ell\mapsto\mathsf{ref}\;S(l)\;\tau]~\textrm{for $\ell\not\in\mathsf{dom}(h)$}\\
\llbracket\mathsf{read\;l\;\tau\;\ell}\rrbracket((i,m,h))(\Sigma) = h(\ell),(i,m,h),\Sigma\\
\llbracket\mathsf{write\;l\;\tau\;\ell\;v}\rrbracket((i,m,h))(\Sigma) = v,(i,m,h[\ell\mapsto v]),\Sigma\\
\end{array}
\]
\caption{Parameters to model a sequential and reentrant variant of Flanagan and Abadi's \emph{Types for Safe Locking}~\cite{safelocking99} and Flanagan and Qadeer's \emph{Types for Atomicity}~\cite{flanagan2003tldi} in our framework.  We sometimes omit the locking component of effects when it is simply $(\emptyset,\emptyset)$ to improve readability.
}
\label{fig:locking_params}
\end{figure}

We define in Figure \ref{fig:locking_params} the parameters to the language framework needed to model locks, mutable heap locations, and lock-indexed reference types, and the primitives to manipulate them.
We define $T$ (new type families) by giving $K$ (kinding of type constructors, which is defined over $T$), and define the primitives $P$ as $\mathsf{LockNames}\uplus\mathsf{Location}\uplus\mathsf{dom}(\delta)$ (locks, heap locations, and primitive operations).
The state \textsf{State} consists of a thread ID, a lock heap mapping locks to the owner and a count of outstanding claims indicating how many times it has been recursively acquired (if held), and a standard mutable store.  The reference type is indexed by a lock (lifted to a singleton type).
Our framework does not include concurrency, but this models the typical single-thread step relation common to small-step formalizations of sequentially consistent shared-memory concurrency (typically coupled with an additional reduction relation that non-deterministically selects a thread to step).
Primitives include lock allocation; lock acquisition and release primitives whose effects indicate both the change in lock claims and the mover type;
allocation of data guarded by a particular lock; 
plus reads and writes, with effects requiring appropriate lock ownership.
Both acquisition and release have two cases --- one case for reentrant behaviors, and one for claiming or leaving an unowned lock.

The primitive types are largely similar, so we explain only two in detail; recall the details of movers were described in Section \ref{sec:atomicity_quantale}.
\textsf{acquire} takes one argument --- a lock --- that is then bound in the latent effect of the
type.  That effect is a product of the locking and atomicity quantales, indicating that the lock
acquisition is a \emph{right mover} ($R$), and that safe execution requires no particular lock
claims on entry, but finishes with the guarantee that the lock passed as an argument is held (we use
syntactic sugar for assumed effect constructors of appropriate arity).
The $\mathsf{read}$ primitive for well-synchronized reads is akin to a
standard dereference operator, but because it works for any reference --- which may be associated
with any lock and store values of any type --- the choice of lock and type must be passed as
arguments before the reference itself.  Given the lock, cell type, and reference, the final latent
effect indicates that the operation requires the specified lock to be held at invocation, preserves
ownership, and is a \emph{both mover} ($B$).

We give a stylized definition of the (partial) semantics function for primitives as acting on not only states but also state types, giving the monotonically increasing state type for each primitive, as required of the parameters (this essentially specifies the choice of a new larger \textsf{StateEnv} for each reduction as required for primitive preservation).  We also omit restating the dynamic effect in our $\llbracket-\rrbracket$ definition; we take it to be the final effect of the corresponding entry in $\delta$ with appropriate value substitutions made --- as required by the type system.
The definitions easily satisfy the primitive preservation property assumed by the type system (Section \ref{sec:params}), by virtue of handling reference types correctly, with the stylized definition of the semantics also giving the choice of new \textsf{StateEnv} for each case of the proof.
Letting the set of blocking primitives be $\{\mathsf{acquire},\mathsf{release}\}$, this instantiation satisfies primitive progress (Section \ref{sec:params}): any well-typed complete primitive application either reduces to a value, or is a stuck application of a blocking primitive (primitive preservation ensures that if a blocking primitive does reduce, the result is appropriately typed).

These parameters are adequate to write and type terms like the following atomic function that reads from a supplied lock-protected reference (permitting syntactic sugar for brevity):
\[
\begin{array}{@{\qquad}l}
\emptyset\vdash \lambda x\ldotp\lambda r\ldotp \mathsf{acquire}\;x; \mathsf{let}\;y=\mathsf{read}\;x\;[\mathsf{bool}]\;r\;\mathsf{in}\;(\mathsf{release}\;x; y)\\
\qquad : \left(\Pi x:\mathsf{lock}\xrightarrow{(\emptyset,\emptyset)\otimes B}
	\Pi r:\mathsf{ref}\;\mathcal{S}(x)\;\mathsf{bool}\xrightarrow{(\emptyset,\emptyset)\otimes A}\mathsf{bool}\right) \mid (\emptyset,\emptyset)\otimes B
\end{array}
\]
This embedding mostly only demonstrates the flexibility of the framework's parameter set. While tracking lock ownership for a single thread is not particularly useful (even if it may structurally resemble other problems),
Lemmas \ref{lem:onestep_preservation} and \ref{thm:one_progress} for this instantiation recover the per-thread portion of a syntactic type safety proof for a multithreaded language.

\subsection{History Effects}
Skalka et al.'s history effects~\cite{skalka2008types} use slightly less abstraction than we do for their technical machinery and soundness proof.  They include a set of constants (only atoms) with assumed singleton types, but also a primitive $\mathsf{ev}(e)$ which evaluates $e$ to a constant (enforced using classic singleton types, which classify exactly one value) --- effectively giving a family of event primitives for a set of events fixed a priori (by the set of constants in the language).
Their  language is otherwise similar to ours (a higher order functional language), though they support Hindley-Milner style inference using prenex polymorphism~\cite{Damas:1982:PTF:582153.582176} (rather than System F style explicit polymorphism as we do), and they use recursive functions rather than loops.

Figure \ref{fig:history} gives the type rules for a slightly restricted version of Skalka et al.'s $\lambda_\textrm{trace}$ language, making several simplifications for brevity.
First, we syntactically restrict the argument to \textsf{ev} to be a constant rather than using singleton types (any $\lambda_\textrm{trace}$ expression can be rewritten in such a form).
Second, we omit universal quantification.  Skalka et al.\ include quantification over singleton types, regular types, and history effects.  The second and third of these can be translated as one would normally translate from multi-kinded Hindley-Milner type schemes to multi-kinded System F. But because Skalka et al.'s system does not explicitly give kinds to type-level variables, adding this to our example would obscure the core ideas we are interested in for our translation.
The singleton types would be more complicated to embed, though also not as useful given our choice to restrict the event construct to constant literals.
Finally, Skalka et al.\ include a fixed point primitive which we omit here, but which could be encoded with a different state instantiation permitting Landin's knot.
Figure \ref{fig:history} also omits the straightforward details of boolean operators.
The types required to characterize this fragment of $\lambda_\textrm{trace}$ is a subset of our core language's (when instantiated with a history effect quantale), so we require no translation between types.

\begin{figure}\small
\begin{mathpar}
\inferrule[Var]{\Gamma(x)=\sigma}{\Gamma,\epsilon\vdash x : \sigma}
\and
\inferrule[Event]{ }{\Gamma,H;\mathsf{ev(c)}\vdash \mathsf{ev}(c) : \mathsf{unit}}
\and
\inferrule[Weaken]{\Gamma,H'\vdash e : \tau \\ H'\sqsubseteq H}{\Gamma,H\vdash e : \tau}
\\\\
\inferrule[If]{\Gamma,H_1\vdash e_1 : \mathsf{bool} \\ \Gamma,H_2\vdash e_2 : \tau \\ \Gamma,H_2\vdash e_3 : \tau}{\Gamma,H_1;H_2\vdash \mathsf{if}~e_1~\mathsf{then}~e_2~\mathsf{else}~e_3 : \tau}
\and
\inferrule[App]{\Gamma,H_1\vdash e_1 : \tau'\overset{H_3}{\rightarrow}\tau \\ \Gamma,H_2\vdash e_2 : \tau'}{\Gamma,H_1;H_2;H_3\vdash e_1 e_2 : \tau}
\\\\
\inferrule[Lambda]{\Gamma;x:\tau,H\vdash e : \tau'}{\Gamma,\epsilon\vdash\lambda x\ldotp e : \tau\overset{H}{\rightarrow}\tau'}
\and
\inferrule[Let]{\Gamma,\epsilon\vdash v : \sigma \\ \Gamma;x:\tau, H\vdash e : \tau}{\Gamma,H\vdash\mathsf{let}~x=v~\mathsf{in}~e : \tau}
\end{mathpar}
\[
\begin{array}{l}
Q(X) = \mathcal{H}(C)\\
c \in C\\
\mathsf{State} = \mathsf{seq}~C\\
K(\mathsf{event}) = \star\\
\end{array}
~~~~~~~~~~~~~~~~~~~~~~~~~~~~~~~~~~~~~~~~~~~~
\begin{array}{rcl}
\delta(c) &=& \mathsf{event}\\
\delta(\mathsf{ev}) &=& \Pi x:\mathsf{event}\xrightarrow{\mathsf{ev}[x]}\mathsf{unit}\\
\mathsf{StateEnv} &=& \{\delta\}
\end{array}\]
\[
\begin{array}{c}
\inferrule{ }{
    \vdash \overline{c} : \Sigma
}
\end{array}
\qquad
\begin{array}{l}
\llbracket\mathsf{ev}\;x\rrbracket(\overline{c})(\Sigma)=(),(\overline{c}\mathrel{++}[x]),\Sigma
\end{array}
\]
\[
\begin{array}{rclrcl}
\llparenthesis x \rrparenthesis & = & x &
\llparenthesis c \rrparenthesis & = & c\\
\llparenthesis \mathsf{ev}(c) \rrparenthesis & = & \mathsf{ev}(c) &
\llparenthesis \mathsf{if}~e_1~\mathsf{then}~e_2~\mathsf{else}~e_3 \rrparenthesis &=& \mathsf{if}~\llparenthesis{e_1}\rrparenthesis~\llparenthesis{e_2}\rrparenthesis~\llparenthesis{e_3}\rrparenthesis\\
\llparenthesis e_1 e_2 \rrparenthesis &=& \llparenthesis e_1 \rrparenthesis \llparenthesis e_2\rrparenthesis &
\llparenthesis (\lambda x\ldotp e) \rrparenthesis &=& (\lambda x\ldotp \llparenthesis e\rrparenthesis)\\
\llparenthesis \mathsf{let}~x=v~\mathsf{in}~e \rrparenthesis &=& (\lambda x\ldotp \llparenthesis e\rrparenthesis) \llparenthesis v\rrparenthesis
\end{array}
\]
\caption{Embedding Skalka et al.'s history effects~\cite{skalka2008types}.
}
\label{fig:history}
\end{figure}
Figure \ref{fig:history} also gives an instantiation of our framework and corresponding type-and-effect-preserving translation, which embeds the selected fragment of Skalka et al.'s system into our core language.
The primitives $P$ are exactly the $\mathsf{ev}$ primitive and the set of constant events $c\in C$.
The runtime state is a sequence of events from $C$ --- a dynamic accumulation of the events in the order they occur during execution.
$\delta$ gives each $c\in C$ the new type \textsf{event} (of kind $K(\mathsf{event})=\star$), and gives the \textsf{ev} primitive a simple dependent type indicating that its latent effect is the single event trace of the event provided as an argument.  Because no dynamic allocation of primitives occurs, $\mathsf{StateEnv}$ is simply the singleton set containing $\delta$, and all states (dynamic traces) are considered well-typed.  The semantics of the event primitive $\llbracket\mathsf{ev}\;c\rrbracket$ are to return unit and append the event to the current state, labeling the dynamic reduction with the singleton trace containing the emitted event (notice this is trivially consistent with the type ascribed by $\delta$).

Finally Figure \ref{fig:history} gives a term translation $\llparenthesis-\rrparenthesis$ translating the syntax of the fragment of Skalka et al.'s system we target into our core language.
Our omissions from Skalka et al.'s system are not trivial, but the details of embedding those aspects are orthogonal to the aspects of an embedding we are interested in.

For this restricted version of $\lambda_\textrm{trace}$, we can obtain the following result:
\begin{lemma}[Embedding of $\lambda_\textrm{trace}$]
    For any $\lambda_\textrm{trace}$ environment, term, type, and history effect such that $\Gamma,H\vdash e : \tau$, $\Gamma\vdash \llparenthesis e \rrparenthesis : \tau \mid H$.
\end{lemma}
\begin{proof}
    By induction on the $\lambda_\textrm{trace}$ derivation, taking advantage of the fact that the types require no translation.
\end{proof}

As a result of instantiating our framework in this way, we obtain two interesting results.  Of minor interest is the ``free'' addition of while loops to their core language.
However, the value of this is limited because we have done so for an instantiation of our core language without inspection of state, so only infinite loops are expressible this way.
Slightly more interesting is a new soundness proof for history effects, proving exactly the same property as Skalka et al.'s original proof. Our general type safety proof (Theorem \ref{thm:soundness}) accumulates traces in the proof, rather than in the semantics, with the $\xrightarrow{\gamma\rhd\ldots\rhd\gamma'}^*_Q$ relation. But we can still recover exactly their proof in terms of dynamic state by giving an effect quantale interpretation, as we did in Example \ref{ex:history_effect_interp}.  
Then Lemma \ref{lem:interp_safety} gives that for any well-typed expression $e$, we have $\epsilon,e\dhxrightarrow{H}\eta,e'$ implies $\eta\in\llbracket{H}\rrbracket$.

\section{Kleene Effect Systems}
\label{sec:kleene}
Readers familiar with Kleene Algebras (\emph{KA}s)~\cite{KOZEN1994366} may have noticed some similarities to effect quantales.  Both are semilattice-ordered monoids (total operations for KAs, partial for effect quantales).  Kleene Algebras also include a (finite) iteration operator with properties similar to those required of effect quantale iteration.  This section describes the relationship between these classes of systems in detail. In summary, every Kleene Algebra is a principally iterable effect quantale --- meaning Kleene Algebras can be used as effect systems ---  while effect quantales admit partiality (important for the examples in Section \ref{sec:modeling}), and have less demanding requirements on iteration when it is defined.

For completeness, we recall a definition of Kleene Algebra due to Kozen~\cite{KOZEN1994366}:
\begin{definition}[Kleene Algebra]
    \label{def:ka}
A Kleene Algebra $\mathcal{K}$ is a structure $(K,+,\cdot,{}^*,0,1)$ where:
\begin{itemize}
    \item $(K,+,0)$ is a total commutative idempotent monoid
    \item $(K,\cdot,1)$ is a total monoid
    \item $0$ is nilpotent for $\cdot$.
    \item $\cdot$ distributes over $+$ on both sides
    \item Deriving a partial order $x\le y \leftrightarrow x + y = y$, the ${}^*$ operator satisfies the following laws:
        \begin{itemize}
            \item[(a)] $1+x\cdot x^* \le x^*$
            \item[(b)] $1+x^*\cdot x \le x^*$
            \item[(c)] $b+a\cdot x\le x \rightarrow a^*\cdot b\le x$
            \item[(d)] $b+x\cdot a\le x \rightarrow b\cdot a^*\le x$
        \end{itemize}
\end{itemize}
\end{definition}
Kleene Algebras arise as a natural algebraic model of sequential computation: $\cdot$ is sequential composition (often elided, as in regular expressions), $+$ is used for alternatives / branching, $1$ is the unit / no-op program, $0$ is the program which always fails, and ${}^*$ models finite iteration.

Setting aside the iteration operator for a moment, it is immediate that every Kleene Algebra $\mathcal{K}$ gives rise to an effect quantale $EQ(\mathcal{K})$, since any idempotent semiring (Kleene Algebra) is also a partial idempotent semiring without 0 (effect quantale):
\begin{proposition}[Kleene Algebra Effect Quantales]
    \label{prop:kaeq}
For every Kleene Algebra $\mathcal{K}=(K,+,\cdot,{}^*,0,1)$, there is a corresponding effect quantale $EQ(\mathcal{K})=(K,+,\cdot,1)$.
\end{proposition}
\begin{proof}
    Satisfaction of the effect quantale laws follows immediately from the properties in Definition \ref{def:ka}.
\end{proof}

Effect quantales are essentially relaxations of Kleene Algebras, in three primary ways. We discuss each along with both its practical consequences, and the philosophical justification based on the fact that Kleene Algebras are intended for equational reasoning about programs, while effect systems are intended for sound bounding of program behavior.

\paragraph{Totality vs.\ Partiality}
Kleene Algebras use total operations, while effect quantales permit all operations (including iteration) to be partial.
Most examples of effect quantales from Section \ref{sec:modeling} are partial in at least one operator, so cannot be directly represented as Kleene Algebras.
Of course, as with the generic frameworks for sequential effects discussed in Section \ref{sec:semantics}, synthetic elements could be added to represent ``not defined'' in order to embed the partial functions into total functions, though it would be unpleasant to require this for most non-trivially sequential effect systems.

The relevant examples in Section \ref{sec:modeling} highlight the difference in purpose between Kleene Algebras and effect systems.  Kleene Algebras are specifically intended to support equational reasoning about (regular) programs. This use case requires knowing what all combinations do. In contrast, effect systems emphasize proving upper bounds on program behavior rather than exact characterizations.  Since their uses are typically to \emph{rule out undesirable behaviors}, once a program is known to contain undesirable behaviors there is no longer a reason for the effect system to classify the program at all --- it can simply not have a valid effect, by way of certain effect combinations being undefined.

Consider security properties.  Kleene Algebras for finite execution prefixes (i.e., the Kleene Algebra of traces, already described as an effect quantale in Section \ref{sec:finite_traces}) provide a good example.  That algebra can describe arbitrary trace sets.  But consider using this to enforce a particular policy (e.g., a linear-time safety property like a security policy) that a certain action occurs at most once during the program. This can be done by separately checking that the top-level effect of the program is a trace set that excludes the offending traces (roughly the approach taken by \citet{skalka2008types,Skalka2008}).  For better error localization, this check can be imposed uniformly on all latent effects. As a general tool for allowing arbitrary policy specifications, this is desirable.  However, effect systems are often intended for use with \emph{specific} policies, in which case there is no reason to allow effects corresponding to trace sets repeating the offending action to even exist: instead, sequencing any effect \emph{corresponding} to having executed the sensitive action, with any other such effect, can simply be undefined. Extending to iteration, this then makes iteration of such effects also undefined (there would be no subidempotent effects above such effects).  

Another example might concern tracking locks (an analogue of $\mathcal{L}(L)$, Definition \ref{def:locking}) via a Kleene Algebra of binary relations (on lock multisets). Kleene iteration is then the reflexive transitive closure of the binary relation, and always defined.  In the case of a relation intended as the effect of acquiring a lock $x$ (i.e., the relation of post-states identical to the pre-state except with an additional lock acquisition on $x$), the result of iterating this is the infinite union of relations modeling all possible numbers of acquisitions (0, 1, 2, \ldots).  This is semantically correct, but from the perspective of an effect system is not very useful: after such a loop, how many times should $x$ be released?  In fact, using this KA as an EQ is sound (Proposition \ref{prop:kaitereq}), but it permits writing large program fragments that have effects that cannot be used meaningfully in complete programs (e.g., iteration of a lock acquisition). Since any program that tries to acquire a lock an indeterminate number of times without releasing it will be problematic, it seems preferable to identify that as an invalid operation --- i.e., allow iteration on that effect to be undefined, as in $\mathcal{L}(L)$.

\paragraph{Least Elements}
\emph{None} of the sequential examples in Section \ref{sec:modeling} have a natural 0 element --- a least element according to join which is also absorbing / nilpotent for sequencing.
As with completing partial functions to total functions, a synthetic 0 could be added to an effect quantale --- but requiring every interesting sequential effect system to add a synthetic element suggests it simply should not be required.

Beyond the pragmatic matter of not wanting to add synthetic elements to most interesting examples, the lack of 0 requirement is also related to the difference in purpose.  Since effect systems are primarily used to rule out undesirable behaviors by not classifying programs with undesirable behaviors, and the 0 element of a Kleene algebra corresponds to programs that fail, effect systems generally have no use for a 0 element.

\paragraph{Weaker Iteration Axioms}

\begin{figure}
    \begin{tabular}{|c|c|c|l|}
        \hline
        & Kleene Algebra Iteration & Effect Quantale Iteration & Effect Quantale Notes\\ \hline
        \hline
        1 & $1\le a^*$ & $I\sqsubseteq a^*$ & Possibly-Empty \\ \hline
        2 & $a\le a^*$ & $a\sqsubseteq a^*$ & Extensive \\ \hline
        3 & $(a^*)^*=a^*$ & $(a^*)^*=a^*$ & Idempotent \\ \hline
        4 & $a\le b\rightarrow a^* \le b^*$ & $a\sqsubseteq b\rightarrow a^* \sqsubseteq b^*$ & Monotone \\ \hline
        5 & $a^*\cdot a^*=a^*$ & $a^*\rhd a^*= a^*$ & Lemma \ref{lem:iter_strict_idem} \\ \hline
        6 & $1+(a\cdot a^*)=a^*$ & $1\sqcup (a\rhd a^*)\sqsubseteq a^*$ & Follows from 1, 2, and 4\\ \hline
        7 & $1+(a^*\cdot a)=a^*$ & $1\sqcup (a^*\rhd a)\sqsubseteq a^*$ & Follows from 1, 2, and 4 \\ \hline
    \end{tabular}
    \caption{Properties of Kleene Algebra iteration~\cite{KOZEN1994366}, alongside analogous properties of effect quantale iteration.}
    \label{fig:compare_iter}
\end{figure}

Figure \ref{fig:compare_iter} recalls some basic properties of Kleene Algebra's iteration operator, alongside the most similar property of effect quantales.  The first 4 properties are identical to axioms of effect quantale iteration (modulo notation and partiality).
The final iteration axiom for effect quantales is weaker than the Kleene Algebra property in row 5, but Lemma \ref{lem:iter_strict_idem} strengthened that ordering to an equality.
In rows 6 and 7, properties which are equalities for Kleene Algebras are weakened to ordering in only one direction.
This is because the Kleene Algebra axioms for iteration specifically imply iteration is a specific least fixed point, while effect quantale iteration does not specifically require this minimality.

The Kleene Algebra axioms determine iteration uniquely --- as the least fixed point of $F(x)=1+x\cdot F(x)$ --- while effect quantales may admit many iteration operators with no hard requirement that any be a least fixed point. In particular, effect quantales may exist with iteration even when $F(x)=1+x\cdot F(x)$ has no unique least fixed point.
There are effect quantales with meaningful iteration, but which are not principally iterable, and where no least iteration operator in the sense of Proposition \ref{prop:maxprecise} exists.
Example \ref{ex:non-principal} defines such an effect quantale, $Q_{NP}$, where principality fails because there is an infinite descending chain of idempotent elements above unit.
There, even with the addition of a synthetic 0, the semiring structure cannot be made a Kleene Algebra because there is no least fixed point for
 $F(x)=1+(x\cdot F(x))$ and $G(X)=1+(G(x)\cdot x)$.

In general, the weaker requirements for effect quantale iteration are motivated by the fact that a sound effect system only requires sound upper bounds. This relaxation admits the possibility that an effect system may have multiple different iteration operators of interest (e.g., Proposition \ref{prop:desc_chain_iter}), and for those that are not principally iterable, may not have an optimally-precise operator.
While this paper does not directly consider implementation or usability issues, there may also be practical reasons to select less precise iteration operators, if they are easier or faster to compute or solve constraints over, or if a less precise operator is easier to explain to developers. (Though naturally, all else being equal, the most precise operator is preferable.)

Since the Kleene Algebra iteration properties directly imply the corresponding effect quantale iteration axioms, we can strengthen the relationship between Kleene Algebras and effect quantales:

\begin{proposition}[Iterable Kleene Algebra Effect Quantales]  
    \label{prop:kaitereq}
For every Kleene Algebra $\mathcal{K}=(K,+,\cdot,{}^*,0,1)$, there is a corresponding \emph{principally iterable} effect quantale $IEQ(\mathcal{K})=(K,+,\cdot,1,{}^*)$.
\end{proposition}
\begin{proof}
    By Proposition \ref{prop:kaeq} and the fact that facts derivable from the Kleene Algebra axioms~\cite{KOZEN1994366} (recalled above) imply the effect quantale axioms.
Because the iteration axioms for Kleene Algebras imply that Kleene iteration is the \emph{least} fixed point of $F(x)=1+x\cdot x^*$, and Kleene iteration is total, this means every subiterable element $x$ has a least iterable element above it (specifically, $x^*$), making this effect quantale \emph{principally iterable}.
\end{proof}

Thus \emph{every Kleene Algebra can be used as a sequential effect system with (principal) iteration}.
Given a Kleene Algebra describing the semantics of some set of primitives, our framework in Section \ref{sec:soundness} shows that those operations can be used as primitives in a higher order functional language with parametric polymorphism over behaviors, and the resulting effect system will yield (for a closed program) a Kleene Algebra expression describing the possible behaviors of the program.

\section{Related and Future Work}
\label{sec:relwork}
The closely related work is split among three major groups: generic effect systems, algebraic models of sequential computation, and concrete effect systems.

\subsection{Generic Effect Systems}
We have already discussed generic treatments of sequential effect systems in Section \ref{sec:semantics}.
We know of only three generic characterizations of traditional single-operator commutative effect systems that aim to capture the full structure common to a range of concrete effect systems.  None are extensible with new primitives.

Marino and Millstein give a generic model of a static traditional (single-operator) commutative effect system~\cite{marino09} for a simple extension of the lambda calculus.  Their formulation is motivated explicitly by the view of effects as capabilities, which pervades their formalism --- effects there are sets of capabilities, values can be tagged with sets of capabilities, and subeffecting follows from set inclusion.  They do not consider type system support for effect polymorphism (they consider the naive substitution of let bindings during type checking).  They do however also parameterize their development by an insightful choice of an \emph{adjust} operation to change the capabilities available within some evaluation context and a \emph{check} operation to check the capabilities required by some redex against those available, allowing great flexibility in how effects are managed.

Henglein et al.~\cite{henglein2005effect} give a simple expository effect system to introduce the technical machinery added to a standard typing judgment in order to track (single-operator commutative) effects.  Like Marino and Millstein they use qualifiers as a primitive to introduce effects.  Because their goals were instructional rather than technical, the calculus is only used to gradually introduce effect systems before presenting a full typed region calculus~\cite{talpin1992polymorphic}.

Rytz et al.~\cite{rytz12} offer a collection of insights for building manageable effect systems, notably the relative effect polymorphism~\cite{rytz12a} mentioned earlier in Section \ref{sec:effpoly} (inspired by anchored exceptions~\cite{vanDooren2005}) and an approach for managing the simultaneous use of multiple effect systems with modest annotation burden.  The system was given abstractly, with respect to a lattice of effects.  Toro and Tanter later implemented this as as a polymorphic extension~\cite{toro2015customizable} to Schwerter et al.'s gradual effect systems~\cite{BanadosSchwerter2014gradual}.  Their implementation is again parameterized with respect to an effect lattice.

\subsection{Algebraic Approaches to Computation}
Our effect quantales are an example of an algebraic approach to modeling sequential computation, though as discussed in Section \ref{sec:kleene} these typically emphasize exact characterization of semantics, rather than establishing bounds on behavior, which leads to stronger axioms (particularly for iteration).
In addition to Kleene Algebras~\cite{KOZEN1994366}, this includes Kleene Algebras with Tests~\cite{kozen1997kleene} (\textsc{KAT}s), action logic~\cite{pratt1990action}, and others.
Each of these examples has (total) operations for joining or sequencing behaviors, with sequencing distributing over joins.  They differ in their approaches to the range of operations modeled, from the primitives of Kleene Algebras to the extension to include blocking tests in \textsc{KAT}, or the extension to including residual actions in action logic. Much of our discussion of Kleene Algebras applies to these systems as well.

\subsection{Concrete Effect Systems}
\label{sec:concrete}
We discussed several example sequential effect systems throughout.
We often appealed to a new flow-sensitive variant of Flanagan and Abadi's \emph{Types for Safe Locking}~\cite{safelocking99} (the precursor to \textsc{RCC/Java}~\cite{rccjava00}), and Flanagan and Qadeer's \emph{Types for Atomicity}~\cite{flanagan2003tldi} (again a precursor to a full Java version~\cite{flanagan2003atomicity}).  This atomicity work is one of the best-known examples of a sequential effect system.  Coupling the atomicity structures developed there with a sequential version of lockset tracking for unstructured locking primitives gives rise to interesting effect quantales, which can be separately specified and then combined to yield a complete effect system.
\citet{skalka2008types,Skalka2008}'s work on history effects is similarly instructive, and when contrasted against the atomicity and locking effect quantales, highlights the use of effects for reasoning about internal details of program behaviors, rather than only external summaries.

One interesting point of comparison between our work and the variety of concrete locking-related effect systems is how the combination of recursive lock acquisition and substitution of lock variables is handled.
Many effect systems that track lock ownership with effects do so with \emph{traditional single-operator commutative} effect systems~\cite{rccjava00,safelocking99,Abadi2006,objtyrace99,boyapati01,boyapati02} that treat \textsf{synchronized} blocks.  In those cases, the use of \textsf{synchronized} blocks pushes the counting for recursive acquisition into the runtime semantics rather than the type system (which cannot track multiplicities with only a single commutative effect operator).
For example, in \textsc{RCC/Java}~\cite{rccjava00} the dynamic semantics permit recursive acquisition and count recursive claims in the evaluation contexts. Acquiring a lock twice nests evaluation contexts that mention the same lock, and since all counting is done dynamically, effects need only track a \emph{set} of locks held.

Taking such sets and naively trying to use them in effects can lead to problems with substitution.
Consider a variant of our locking effect quantale using sets rather than multisets, and consider
the term $f=(\lambda l_1\ldotp\lambda l_2\ldotp \mathsf{acquire}\;l_1;\mathsf{acquire}\;l_2)$, which would have type $\Pi l_1:\mathsf{lock}\overset{I}{\rightarrow}\Pi l_2:\mathsf{lock}\overset{(\emptyset,\{l_1,l_2\})}{\rightarrow}\mathsf{unit}$ (using only locking, not atomicity).
Intuitively, applying this function to the same lock $x$ twice ($f\;x\;x$) would eventually substitute the same value for $l_1$ and $l_2$, yielding an expected overall effect of $(\emptyset,\{x\})$ after type/effect-level substitution --- the number of locks acquired shrank because the set would collapse, though the underlying term would try to acquire the same lock twice.  Moreover, after reducing the second application, the resulting term would no longer by type-correct, as $(\emptyset,\{x\})\rhd(\emptyset,\{x\})$ must be undefined when using non-reentrant locks!  This is because the substitution loses information (that a lock was acquired twice) that is not duplicated dynamically.
In terms of indexed effect quantales, this hypothetical broken example 
fails to be an indexed effect quantale because the mapping induced by variable substitution cannot map all valid effects to valid effects. 

We know of two ways to deal with this interaction between substitution and effects that have some notion of multiplicity.
Our locking effect quantale (Definition \ref{def:locking}) uses multisets to avoid losing information: if a function merges two locks, their multiplicities are summed, so the total number of lock acquisitions remains the same.
The alternative is to impose constraints on what substitutions are permitted.
\citet{suenaga2008type} does this: in his system the term $f$ above would be rejected if the two locks had the same level (as they must for caller passing the same lock for both to type check), because simultaneously held locks must have distinct levels (Definition \ref{def:dlf_suenaga}).  Giving the two formal parameters distinct levels would then make the application ill-typed.  Notice however, that Suenaga's system does not form an indexed effect quantale according to Definition \ref{def:indexed_eq}.  While it is naturally monotone, the substitution of $x$ for both locks is still semantically defined, but would not yield an effect quantale homomorphism between instantiations of $\mathcal{DL}(-)$ (since the substitution would map defined compositions to undefined compositions).  Instead Suenaga's system relies crucially on the type system only requiring functorial behavior for certain classes of functions \emph{constrained by the type environment} --- intuitively it relies on the fact that the type system only performs type-preserving substitutions, and that substituting $x$ for lock variables at two distinct levels would be ill-typed, so the type system does not allow corresponding substitutions, and the interaction of the pre-indexed effect quantale with such functions is irrelevant.
This is the only example we know of that requires a more refined, type-environment-dependent domain of sensible behaviors for working with substitution.  Identifying general principles here seems worthwhile future work, but seems likely to require additional examples to avoid over-specializing. Such generalization would require the definition of the appropriate variation on indexing to be defined mutually with the relevant restrictions on substitution imposed by the particular type system.
Another approach could be employing HM(X)-style constraints~\cite{odersky1999type} to ensure functions are only called with arguments that would not lose track of information; like an approach closer to Suenaga's system, this would require refining how indexed effect quantales handle their index sets.

Many other systems that are not typically presented as effect systems can be modeled as sequential
effect systems.  Notably this includes systems with flow-sensitive additional contexts (e.g., sets
of capabilities, or systems like Suenaga's) as alluded to in Section \ref{sec:bg}, or fragments of type information in systems
that as-presented perform strong updates on the local variable contexts (e.g., the state transitions
tracked by typestate~\cite{WolffGTA11,Garcia2014FTP}, though richer systems require dynamic reflection
of typestate checks into types~\cite{Sunshine2011}, which is a richer form of dependent effects than our framework currently tracks).  Other forms of behavioral type systems have at least a close correspondence to known effect systems, which are likely to be adaptable to our framework or an extension in the future. \citet{orchard2016effects} demonstrated that the long-recognized similarity between session types~\cite{Honda2008} and Nielson and Nielson's effect system for communication in CML~\cite{nielson1993cml} could be made precise.
Resource usage type systems~\cite{igarashi2002resource} have similar notions of join and sequencing to track resource usage, but also make significant use of substructural restrictions on type contexts which are key to enforcing intended resource constraints; it may be possible to restructure these uses into a single notion of effect as we did for Suenaga's system in Section \ref{sec:dlf}, or these may be better described as a graded coeffect~\cite{gaboardi2016combining}.

Many of the systems above, as described in Section \ref{sec:modeling}, have partial sequencing and/or joins of effects.  To the best of our knowledge, no prior work has dealt with the combination of effect polymorphism and partiality in effect operations.  Concurrently with this work, \citet{jones2020partial} described an extension of Haskell with partial type constructors, which is then integrated with Haskell's support for constraint contexts. They add a new type constraint that a given type constructor is defined for a specific input, and show the presence of this constraint can generally be inferred from usage once the core type constructors (e.g., function type formation) are annotated.  This is useful for cases like Haskell's \verb|UArray|, whose second argument must be a type the compiler knows how to unbox. This is not unlike the consideration given in Section \ref{sec:effpoly} to the possibility of handling the partiality of effect operators via constraints. 
We expect doing something similar for operators rather than constructors would introduce new challenges, but could be another promising approach to explore.

\subsection{Limitations and Future Work}
There remain a few important aspects of sequential effect systems that neither we, nor related work on semantic characterizations of sequential effects, have considered.
One important example is the presence of a masking construct~\cite{lucassen88,gifford86} that
locally suppresses some effect, such as try-catch blocks or \texttt{letregion} in region calculi.
Another is serious consideration of control effects, which are alluded to in Mycroft et al.'s work~\cite{mycroft16}, but otherwise have not been directly considered in the algebraic characterizations of sequential effects.
We have taken steps in this direction in a parallel line of work~\cite{ecoop20b} addressing tagged delimited continuations, but there remains further work to be done for supporting constructs like \texttt{finally} blocks or the fact that exiting a \texttt{synchronized} block releases the corresponding lock.

Our generic language carries some additional limitations.  
It lacks sub\emph{typing}, which enhances usability of the system, but should not present any new technical difficulties, especially since we do support effect subsumption.
It also lacks support for adding new evaluation contexts, which is important
for modeling constructs like \texttt{letregion}.
Allowing this would require more sophisticated machinery for composing partial semantic
definitions~\cite{birkedal2013intensional,Delaware2013MLC,Delaware2013MMM}.

Beyond the effect-flavored variation~\cite{lucassen88,talpin1992polymorphic} of parametric polymorphism and the polymorphism arising from singleton types as we consider here, the literature contains bounded~\cite{Grossman2002Cyclone} (or more generally, constraint-based) effect polymorphism, and unusual ``lightweight'' forms of effect polymorphism~\cite{rytz12,ecoop13} with no direct parallel in traditional approaches to polymorphism.  Extending  our approach for these seems sensible and feasible.  \citet{Skalka2008} includes forms of effect polymorphism specific to objects, with some superficial resemblance to dependent object types~\cite{amin2012dependent,rompf2016type}.

Finally, we have not considered concurrency and sequential effects, beyond noting the gap between joinoids' fork-join style operator and common concurrency constructs.  
As a result we have not directly proven that our multiset-of-locks effect quantale ensures data race freedom or atomicity for a true concurrent language.

\section{Conclusions}
We have given a new algebraic structure --- effect quantales --- for effects in sequential effect systems, and shown it sufficient to implement complete effect systems, unlike previous approaches that focused on a subset of real language features.  We used them to model classic examples from the sequential effect system literature, and gave a syntactic type safety proof for the first generic sequential effect system, including an extension to cover some semantic properties.  Moreover, we give the first investigation of the generic interaction between (singleton) dependent effects and algebraic models of sequential effects, discuss some subtleties in mixing effect polymorphism with partial effect operators, and give a way to derive an appropriate iteration operator on effects for many effect quantales (recovering manually-designed operators from prior concrete systems).  We have also discussed the relationship between Kleene Algebras and effect quantales in some depth, highlighting that Kleene Algebras can in fact be used as sequential effect systems.  We believe these results form an important basis for future work designing complete sequential effect systems, and for generic effect system implementation frameworks supporting sequential effects.

\section*{Acknowledgments}
We are deeply grateful to the anonymous TOPLAS reviewers for multiple rounds of remarkably careful, thorough, and constructive feedback on both presentation and technical developments in earlier drafts of this unusually long paper. This feedback has led to substantial improvements in the current version.

Portions of this work were funded by NSF Grant \#2007582.

\bibliographystyle{ACM-Reference-Format}

\appendix
\section{Relationship to Original Effect Quantales Publication}
\label{apdx:comparison}

The original effect quantales publication~\cite{ecoop17} used slightly different definitions for several concepts.  This appendix surveys the differences and their consequences.

Broadly speaking, the primary changes are:
\begin{itemize}
    \item This paper changes the definition of effect quantale to an equivalent form that handles partiality differently, replacing the simulated partiality of the prior paper (which used a distinguished element $\top$ to represent undefined results) with actual partial functions.  As part of this, the definitions of indexed effect quantales and effect quantale homomorphisms changed to use this new definition, which better highlights key ideas.
    \item This paper gives more general criteria for inducing an iteration operator, covering more concrete systems from the literature (which appear as additional examples in this version), and further strengthens the results for iteration by identifying two broad classes of effect quantales that always satisfy the criteria (finite effect quantales and those with non-empty meets above unit), and by showing that when defined the construction gives an optimally-precise iteration operator.
    \item This paper improves the discussion of soundness.  First, it more clearly articulates throughout that our soundness results are soundness for \emph{safety} properties, because we use syntactic type safety on an inductively-defined reduction relation.  Second, it improves upon the the original paper, which established only the
    syntactic consistency between dynamic reductions and static effects (Theorem \ref{thm:soundness}) that is typical of
    syntactic type-safety proofs for other sequential effect systems: assuming an instrumented semantics that labels every reduction with the primitive effect of that step, the static effect of an expression over-approximates the sequencing of all individual steps' effects in order (for finite executions).
    This paper also gives a way to interpret individual effects as relations on pre- and post-states of execution, extending the type safety proof to show this semantic interpretation is respected by the framework (Lemma \ref{lem:interp_safety}).
    \item This paper adds discussion of using Kleene Algebras as effect quantales.
    \item This paper drops an instantiation of the parameters for a combination of race freedom and atomicity, which was originally used to give a non-semantics preserving translation from Flanagan and Qadeer's CAT language~\cite{flanagan2003tldi} to our framework. Instead we give an instantiation for \citet{skalka2008types} that preserves semantics.  The original paper failed to preserve semantics for spawning new threads; this paper drops the embedding, and better explains the relationship between our locking instantiation of the framework and proofs for a typical shared-memory concurrent language.
\end{itemize}
Additionally, Section \ref{sec:modeling2} gives a different semantics for locking primitives than the original ECOOP paper, which incorrectly gave non-reentrant semantics (a boolean flag for each lock) which satisfied primitive preservation, but not primitive progress. This version of the work correctly gives re-entrant locking semantics (i.e., counting claims) satisfying both.
And one additional minor change is the clarification to the use of the term ``commutative effect system'' --- the original paper ignores the existence of effect systems with two distinct commutative operators (e.g., Section \ref{sec:must}).

We discuss the first two items, which are significant changes rather than additions or corrections, in more detail.
In addition, Appendix \ref{apdx:quantales} gives a more complete accounting of the relationship between quantales and effect quantales.

\subsection{Real vs.\ simulated partiality}
The original effect quantales paper~\cite{ecoop17} gave a \emph{total} formulation of effect quantales --- sequencing and joins were always defined --- using a distinguished top element that was nilpotent for $\rhd$.  This meant that if any subexpression of an effect calculation yielded $\top$, so would the overall calculation, just as an effect expression with undefined subexpressions is considered undefined.  This was essentially completing our partial definition to a total one, with $\top$ acting as a sentinel element to indicate undefined results.
Switching between the two definitions does not affect the expressive power of effect quantales as long as programs whose effect simplifies to $\top$ are rejected (as in the original work).
The original paper stated an implicit side condition on each type rule that the effect was not equivalent to $\top$, which doubled for both this purpose as well as for the issues missing effect variables with partial effect operators now discussed in Section \ref{sec:soundness} (the previous paper did not highlight the dual purpose of this check).

But there are several advantages to dropping the ``synthetic partiality'' of the original definition in favor of using partial operators.

From an expository standpoint, it was problematic that many effect systems (notably most commutative effect systems) already have a greatest element, and adding a new synthetic element makes distinguishing the two mildly confusing (see the discussion of the \citet{flanagan2003tldi} atomicity example in the original paper, which distinguished $\top_{FQ}$ as the original top element and $\top$ as the synthetic error element).

From a technical standpoint, using true partiality leads to some simplifications.  The original paper used a different notion of effect quantale homomorphism, which was simultaneously too strict in some ways and too lax in others.  On the side of being overly strict, the definition of the original paper required homomorphisms to preserve top (the distinguished error element), meaning homomorphisms could not map one effect quantale into another that permitted more sequencing and join inputs to have defined outputs (i.e., it prohibited morphisms into upwards extensions of the join semilattice structure).  It also required the homomorphism to be a monoid homomorphism and semilattice homomorphism, strictly preserving $\rhd$ and $\sqcup$ results, rather than allowing the result in the codomain to be less than the original. These did not pose problems for the uses intended for this paper (and the original), but would make the original definition less useful for comparing the precision of different effect systems.  

The original definition was too lax in that it permitted mapping some non-$\top$ (non-error) elements of the domain to $\top$ in the codomain, introducing \emph{more} errors.  Untreated, this would cause problems for soundness (for example, value substitution into an indexed effect quantale using such homomorphism could lead to effects becoming undefined, breaking subject reduction).
The original paper introduced an additional class of indexed effect quantales called \emph{collapsible} indexed effect quantales, where the homomorphism assigned to a function always reflected $\top$ --- i.e., if a sequencing or join of two elements was undefined ($\top$) in the codomain quantale, it must also be undefined ($\top$) in the domain.  This is implied by the new definition of effect quantale homomorphism given as Definition \ref{def:morphism}, making separate identification of collapsible indexed effect quantales unnecessary here.

The new Definition \ref{def:morphism} also makes it clearer that an effect quantale homomorphism is essentially an embedding of one effect quantale into another (possibly more nuanced) effect quantale, both by incorporating the equivalent of collapsibility into the main definition, and by permitting the embedding of sequencing and joins to be more precise (less) than the embedding of the original result.

\subsection{Principally vs.\ Distributively Iterable Effect Quantales}
\label{sec:lax_dist}
The earlier version of this work used a different approach to constructing an iteration operator.   Like Section \ref{sec:iteration} it used closure operators (on the total version of effect quantales with synthetic top), but beyond the change in partiality it differed from the approach in Section \ref{sec:iteration} in three key respects:
\begin{itemize}
    \item It defined iteration in terms of $\mathsf{Idem}(Q)$ rather than $\mathsf{SubIdem}(Q)$.
    \item It required (what it called) \emph{iterable} (here called \emph{distributively iterable}, defined below) effect quantales' strictly idempotent elements to be closed under joins --- that the join of any two strictly idempotent elements was itself strictly idempotent. The construction of a useful iteration operator was then only claimed for these slightly more restrictive effect quantales (the construction itself was the same closure operator construction on this more restricted closure subset).
    \item It guaranteed a strong distributivity property of the (same) closure operator constructed above, that joins distribute strictly over iteration --- $\forall a,b\ldotp a^*\sqcup b^*=(a\sqcup b)^*$ --- which held because of the extra restrictions on the closure operator.
\end{itemize}
Our generalization to subidempotents was mainly to support the proofs of Propositions \ref{prop:finite_iterable} and \ref{prop:complete_iterable}, which are slightly easier to prove with the more open-ended description, though it turns out the strictly idempotent and subidempotent elements \emph{above unit} coincide (Lemma \ref{lem:iter_strict_idem}).

The second and third distinctions are more obviously meaningful, and the reason we defined \emph{principally iterable} effect quantales above.
They are more general than the \emph{iterable} effect quantales originally proposed~\cite{ecoop17}, which were too strict for \emph{behavioral}~\cite{PGL-031} effect systems like history effects~\cite{Skalka2008} or trace sets~\cite{Koskinen14LTR} that expose internal behaviors of code, rather than (informally) summaries.
Note that the prior paper defined an iterable effect quantale as one satisfying that paper's requirements for inducing an iteration operator, whereas in this paper an iterable effect quantale is simply an effect quantale with a choice of iteration operator (which always exists).
In comparison, \emph{principally} iterable effect quantales satisfy a slight relaxation of the distributivity property given above: $a^*\sqcup b^*\sqsubseteq (a\sqcup b)^*$. 
The proof of this weaker property is straightforward using monotonicity of ${-}^*$ and idempotence of joins:  $a^*\sqcup b^*\sqsubseteq (a\sqcup b)^*\sqcup(a\sqcup b)^* = (a\sqcup b)^*$.
We now call what prior work referred to as simply iterable as \emph{distributively iterable effect quantales}: principally iterable effect quantales where additionally the subidempotent elements are closed under joins.  In those effect quantales the weaker law above becomes the stronger distributive equality:
\[\begin{array}{rcl}
a^*\sqcup b^* &=& \mathsf{min}(a\uparrow\cap(I\uparrow\cap\mathsf{SubIdem}(Q))) \sqcup \mathsf{min}(b\uparrow\cap(I\uparrow\cap\mathsf{SubIdem}(Q)))\\
&=& \mathsf{min}((a\sqcup b)\uparrow\cap(I\uparrow\cap\mathsf{SubIdem}(Q)))\\
&=& (a\sqcup b)^*
\end{array}
\]

Our earlier work by chance only formally considered distributively iterable examples like the atomicity and locking quantales, as well as the product construction --- which preserves distributive iterability (as well as principal iterability).  As a result, we conjectured there that all meaningful effect quantales were (in this paper's terminology) distributively iterable.  Here we refine our conjecture: we believe all meaningful effect quantales are principally iterable, a weaker condition.
As in the initial version of this work, we cannot make the claim precise. But unlike the earlier version, we can claim a level of generality via Propositions \ref{prop:finite_iterable} and \ref{prop:complete_iterable}.

A key motivating example for switching to principally iterable effect quantales is given by the finite trace sets (Section \ref{sec:finite_traces}).  In that case, for events (not effects) $a$ and $b$, $\{a\}^*=\bigcup_{i\in\mathbb{N}} {a^i}$ (assuming $a^0=\epsilon$) and similarly for $b$.  The join of these sets contains strings with any finite number of $a$s or any finite number of $b$s, but no mixed strings.  In contrast, $(\{a\}\sqcup\{b\})^*=\{a,b\}^*$ contains all finite strings composed of $a$ and $b$, including strings containing both.
These sets are not equivalent, though the former is a subset of the latter.
Informally, \emph{behavioral} effects that expose internal behaviors of computations tend to be only principally iterable, while \emph{summarizing} effects that essentially only give a summary of externally-visible behavior (e.g., locking, atomicity) tend to be distributively iterable.  However this remains an informal claim --- a proof would require a formal distinction between these types of systems.

\subsection{Quantales vs.\ Effect Quantales}
\label{apdx:quantales}
Quantales~\cite{mulvey1986,mulvey1992quantisation} are an algebraic structure originally proposed to generalize some concepts in topology to the non-commutative case.  They later found use in models for non-commutative linear logic~\cite{yetter1990quantales} and reasoning about observations of computational processes~\cite{abramsky1993quantales}, among other uses. \citet{abramsky1993quantales} give a thorough historical account.  They are almost exactly the structure we require to model a sequential effect system, but just slightly too strong.  Here we give the original definition, which was relaxed to the definition of effect quantales (Definition \ref{def:eq}).

\begin{definition}[Quantale~\cite{mulvey1986,mulvey1992quantisation}]
A \emph{quantale} $Q=(E,\wedge,\vee,\cdot)$ is a complete lattice $(E,\wedge,\vee)$ and an associative product $\cdot$ that distributes on both sides over arbitrary (including infinite) joins:\\
\centerline{
$a\cdot \left(\bigvee b_i\right) = \bigvee (a\cdot b_i)$
and
$\left(\bigvee b_i\right)\cdot a = \bigvee (b_i\cdot a)$
}
Additionally, a quantale is called \emph{unital} if it includes an element $I$ that acts as left and right unit for the product --- $I\cdot a=a=a\cdot I$, or in other words $(E,\cdot,I)$ is a monoid.
\end{definition}
Because of the similarity to rings, the join is often referred to as the additive element, while the semigroup or monoid operation is typically referred to as the multiplicative operation.  Because the lattice is complete, it is bounded, and therefore contains both a greatest and least element.

A unital quantale is close to what we require, but just slightly too strong.  
In particular, every quantale has a least element (the join of the empty set).
The \emph{complete} lattice structure combined with the associative multiplication distributing over the joins makes all quantales residuated lattices.
In any residuated lattice with a bottom element, the bottom element is always nilpotent for multiplication~\cite[\textsection 2.2]{galatos2007residuated} --- $\bot\cdot x = \bot = x\cdot\bot$, for all $x$.
This conflicts with the use of the bottom element (if present) as the unit for composition in traditional commutative effect systems, 
and there are sequential effect systems in the literature with no natural bottom element (such as the lockset example of Section \ref{sec:locking}).
So both for natural treatment of sequential effects, and because we would additionally like to subsume traditional commutative effect systems, we require a slightly more general structure.

Our framework does not require a bottom element, nor a meet operation, though those could be useful in some contexts (e.g., for type inference in the presence of subtyping~\cite{oopsla16}).  
The need to join over empty or infinite sets is also not required by any effect system we know of, making structures requiring complete lattices (e.g., quantales) too strong.  Thus we replace the complete lattice of a standard quantale with a partial join semilattice (binary joins only), in addition to requiring the unit to exist.
\end{document}